\documentclass{article}
\usepackage[utf8]{inputenc}
\usepackage{pgfplots}
\pgfplotsset{compat=newest}
\usepgfplotslibrary{fillbetween}
\usetikzlibrary{arrows.meta}
\usepackage[english]{babel}
\usepackage{amsmath,amsthm}
\usepackage{graphicx}
\usepackage{amssymb}
\usepackage{enumerate}
\usepackage{mathtools}
\usepackage{natbib}
\usepackage[T1]{fontenc}
\usepackage[multiple]{footmisc}
\usepackage{lmodern}
\usepackage[autostyle]{csquotes}
\usepackage{setspace}
\usepackage{appendix}
\usepackage{tikz-cd}
\usepackage{url}
\usepackage{hyperref}
\usepackage[normalem]{ulem}
\theoremstyle{definition}
\newtheorem{definition}{Definition}
\theoremstyle{plain}
\newtheorem{theorem}{Theorem}
\newtheorem{lemma}{Lemma}
\newtheorem{corollary}{Corollary}

\newtheorem{remark}{Remark}

\newtheorem{proposition}{Proposition}

\newtheorem{result}{Result}

\newtheorem{principle}{Principle}

\newtheorem{construction}{Construction}
\DeclareMathAlphabet{\mathpzc}{OT1}{pzc}{m}{it}

\begin{document}
\title{Statistical and Numerical Convergence in Stochastic Equilibrium}

\author{David Staines}
\date{\today}
\maketitle

\begin{abstract}
This paper sets out the most general 
computational and econometric implications of the rigorous stochastic equilibrium theory from SELCKE (\cite{staines2024stochastic}) \url{https://arxiv.org/abs/2312.16214}. The analytical backbone is the discovery that the system converges geometrically to \\ long-run equilibrium, 
at a rate given by the greater of the eigenvalue or inverse eigenvalue (from outside) closest to the unit circle and the maximum shock persistence. 
High-order shocks converge faster. I develop a simulation procedure to test, with asymptotic power, whether stochastic equilibrium exists for a particular model.
\par The fundamental approximation result asserts that, whatever the \\ order of expansion or loss function, the stochastic steady state delivers the most accurate perturbation solution. I also show that super-consistent parameter 
estimators $O(1/T)$  arise whenever second-order terms vanish.
\par Besides Calvo, I study stochastic equilibrium in two alternative pricing models. Dynamics simplify considerably. I bound the time 
the impulse response peaks, by the maximum lag in the errors. This lends empirical support to Taylor contracts, 
although there are issues surrounding unit roots and the strong cost-channel. \par For menu costs, I demonstrate that the initial price distribution decays away super-exponentially, producing a system equivalent to Calvo with an endogenous reset probability. The impact of idiosyncratic disturbances appears as an additional wedge between actual and efficient output.
\\ Blow-up of the objective function at the boundary is proven, with the help of new distributional arguments, so the model meets existing eigenvalue existence conditions for the recursive equilibrium. Along the way, new light is shone on existing theoretical models and statistical procedures.
\end{abstract}
Arxiv: (Primary) Econometrics \newline (Secondary) 
Theoretical Economics, Probability, \newline Classical Analysis and ODE, Numerical Analysis.
Statistical Estimation, Super-Consistency, Impulse Response Function, \\ Menu Costs, Taylor Contracts.
60B10, 60B12, 60F05, 
60F25, 60K99, 62L20, 62P20, 62P99, 65Q10, \\ 65Q20, 65Z99, 91A15, 91A16, 91A50, 91B64, 91B84. 
\tableofcontents
\section{Introduction}
Dynamic Stochastic General Equilibrium (DSGE) are the class of models used to aggregate between individual incentives and macroeconomic outcomes. They play a prominent role in policy analysis at central banks and international \\ institutions, where they are used to simulate policy counterfactuals, forecast and qualitatively explain macroeconomic phenomenon. Previously criticized on the grounds of empirical fit and technical construction, there has been a renewal in theoretical and empirical interest following SELCKE. This paper refines the technical arsenal of stochastic equilibrium theory with a view to econometric and computational application.\footnote{Stochastic equilibrium theorem is a continuation of non-rigorous work referring to the "risky steady state" \cite{coeurdacier2011risky}.}
\par The paper focuses on three themes. 
The first concerns the numerical \\ analysis of stochastic equilibrium. The second derives restrictions on the return to stochastic equilibrium in terms of the primitives of the recursive equilibrium. The third relates to the asymptotic distribution of key test statistics. 
\par Stochastic equilibrium theory implies two fundamental approximation \\ results. The first is a test for the existence, using only the short term response to a small white noise shock, after a long enough burn in. 
The second is a guarantee, that for any order of approximation, the stochastic equilibrium is the neighborhood from which the most accurate approximation comes. This implies a simulation strategy from which the stochastic equilibrium expansion can be obtained through basic regression analysis.
\par These results address two erstwhile separate strands of literature. The first entails numerical algorithms developed for DSGE and other dynamic nonlinear macroeconomic models, without watertight mathematical justification. This includes the variety of solution methods detailed in \cite{maliar2011solving}, \cite{judd2014smolyak}, \cite{maliar2021deep}, \cite{bilal2023solving}, \cite{fernandez2023dynamic} and \cite{mennuni2024dynamic}.
The second relates to mathematical proofs concerning approximations to mean field games, such as \cite{Achdou2020mean}, \cite{carmona2021convergence}, \cite{carmona2022convergence}, \cite{gianatti2023approximation} and \cite{bertucci2024mean}. However, the focus here has typically been on discretizing continuous processes and dealing with the limit as the number of players becomes infinite. Finally, \cite{brumm2017recursive} provide an efficient algorithm but it pertains only to a simple production economy, where existence is guaranteed.
\par The shape of impulse response functions is the second theme. 
A large body of empirical evidence indicates that output and in particular inflation respond slowly to monetary shocks (see \cite{christiano1999monetary}, \cite{piazzesi2002fed}, \cite{romer2004new}, \cite{uhlig2005effects}, \cite{barnichon2018functional}, \cite{jorda2020effects}, \cite{palma2022real} and \cite{pruser2024large}.\footnote{Comparable dynamics have been observed in response to fiscal policy shocks (\cite{auerbach2013output}, \cite{cloyne2013discretionary}, \cite{guajardo2014expansionary}, \cite{acconcia2014mafia} and \cite{ramey2019ten}).} \footnote{The results have broader implications. There are similar patterns in exchange rate time series and the technical arguments would extend to flexible price models.} 
Similar results have been uncovered for a variety of shocks, including oil prices (\cite{hamilton1983oil}, \\ \cite{kilian2009not} and \cite{baumeister2019structural}),  volatility (\cite{bloom2009impact} and \cite{baker2024using}) and technology (\cite{christiano2004response}, \cite{feve2009response} and\cite{alexopoulos2011read})
, where the focus has typically been on capital and labor adjustment. 
Indeed, indications of hump-shaped response patterns can be detected through pure time series methods (\cite{perron1993hump}).
\par Note that this discussion is in some sense non-rigorous because it is not clear all the structural models I present or could be derived hereafter obey all possible identifying restrictions. For this reason, I favor estimates like \cite{jorda2020effects} and \cite{palma2022real} derived from natural experiments, which suggest a longer time to peak response (for both output and inflation). \cite{jorda2020effects} finds typically around four to eight years compared to nearer one year in \cite{christiano1999monetary}. There is considerable imprecision, nevertheless the majority comfortably exceed the two quarter upper bound, derived here for the benchmark Calvo model. 
\par My main result here bounds the peak time of the impulse response by the maximum lag present in the model. This favors \cite{taylor1979staggered} 
pricing and its generalizations, which unlike their counterparts allow for two lags or more. This chimes with a literature suggesting that longer contracts are need to fit the response of output, employment and inflation to monetary and other prominent macroeconomic disturbances (see \cite{taylor1980aggregate}, \cite{chari2000sticky}, \cite{dixon2010can}, \cite{dixon2011contract} and \cite{dixon2012generalised}).
\par In addition, the new limiting stochastic equilibrium theory yields advances in precision. I prove that the vanishing of second order effects gives rise to consistency of order $1/T$ rather than the standard $1/\sqrt{T}$. This should yield substantive gains in empirical analysis of macroeconomic data, where degrees of freedom are short and power weak. This is a macroeconometric breakthrough, hitherto super-consistency results have been confined to time series or panel with time trends (\cite{beenstock2019econometric}).\footnote{There are examples of super-consistency in microeconometrics, for instances, relating to the order statistics governing optimal bidding in auctions (\cite{donald2002superconsistent}). } The drawback is that the test is only valid against small noise limit alternatives.\footnote{This limitation is well-known in statistics, where it has been shown that super-consistency is an appropriate sense a measure zero phenomenon (\cite{le1960some} and a recent extension \cite{bredahl2024superconsistency}).}  
\par I touch on other areas.
I am able to demonstrate that the menu costs model of nominal rigidity, pioneered by \cite{sheshinski1977inflation} and \cite{mankiw1985small} has common dynamics with Calvo, as it can be viewed as Calvo with an optimal choice of reset parameter.\footnote{Some technical conditions are required to ensure the economy-wide price level does not become flexible.} Alongside the limiting equivalence between Calvo and Taylor pricing discussed in SELCKE, 
this creates a universality class of New Keynesian models. 
\par Moreover, I am able to evince unambiguous improvements in the clarity of dynamic predictions. Unlike with just the recursive form in stochastic \\ equilibrium, I am able to remove dependence upon reset pricing variables that may be more difficult for economists to track or households to appreciate. \\ Alongside, come quantitative reductions in the dimensionality of the model, large with Taylor contracts, this could prove significant in, for example, \\ forecasting exercises with limited observations.  
\par In fact, this account directly addresses several theoretical debates. It \\ undercuts the \cite{moll2025trouble} criticism that rational expectations requires tracking the entire distribution of idiosyncratic shocks. Nonetheless, his focus on rational learning and behavioral frictions both in that lecture and the subsequent paper \cite{moll2026mean} would be a good incorporation into my framework. The analysis of menu costs frustrates recent attempts to define small noise limit approximations, to at least this particular heterogeneous agent model. 
\par Finally, renewed justification is provided for the indirect inference approach to estimation. This is a popular means to estimate or test complicated non-linear models (\cite{gallant2010simulated}, \cite{schmitt2012s}, \cite{guvenen2014inferring} and \cite{berger2015consumption}). It involves \\ simulating the model and selecting parameters by how well they fit an \\ auxiliary data description like a Vector Autoregression (VAR). Traditionally \\ defended with reference to small sample properties (see \cite{gourieroux1993indirect} and \cite{meenagh2024indirect}), I give asymptotic support in two ways. Firstly, I demonstrate that the minimum distance estimator, common in applied work on indirect inference (\cite{le2016testing}), constitutes the efficient estimator in the limit where second order terms vanish. Secondly, I argue that the simulation \\ procedure for estimating approximate dynamics in stochastic equilibrium amounts to undertaking indirect inference with the appropriate auxiliary model, although this is not yet a common feature of the literature.
\section{Technical Overview}
This second introductory section is split into two units. The main part outlines the main mathematical results, explains the proof strategies and relates these to the mathematical literature. A short subsection explains the organization of the paper.
\subsection{Mathematical Arguments}
Stochastic equilibrium is a primitive whenever anyone runs a regression with the aim of deploying standard econometric tests, a ubiquitous practice in applied work. The vocation of the paper is to garner all information relating to statistical and numerical approximation true throughout the widest class of DSGE model. The cornerstone is Theorem 1 a version of the fixed-point Theorem 3 (SELCKE) reworked to shed dependence on the patient limit $(\beta \rightarrow 1$ and specific behavior at the boundary of the state space.
\par The six overriding determinations of this paper are given below, with \\ reference to their formal statement and proof. There are two groups of three. The first relate directly to long-time return to the equilibrium distribution after a shock and how to observe this numerically. 
\begin{result} (\emph{[Theorem 3]} and \emph{[Theorem 9]})
A DSGE model, in stochastic \\ equilibrium, converges (almost surely) to its long-run stochastic steady state, at an exponential rate equal to the maximum of internal and external persistence, regardless of the size initial shock.
\end{result}
\begin{result}(\emph{[Theorem 8]})
Every convergent metric between actual outcomes and those predicted by a polynomial of degree $n$ is minimized (asymptotically) by taking the $n^{th}$-order Taylor expansion around the stochastic steady state.
\end{result}
\begin{result}(\emph{[Theorem 9]} and \emph{[Theorem 10]})
Every standard DSGE in stochastic equilibrium can be approximated with exponentially improving accuracy, which is amenable to super-consistent testing. 
\end{result}
The second set have a broader lens featuring long-run equilibrium and \\ short-run deviations.
\begin{result}(\emph{[Theorem 5]} and \emph{[Theorem 6]})For an economy facing small shocks. The time at which the impulse response function peaks is bounded by the oldest shock that enters its recursive equilibrium, so hump-shaped responses peaking after more than two periods are only possible for Taylor contracts. 
\end{result}
\begin{result}(\emph{[Theorem 2]} and \emph{[Theorem 4]})
The menu cost model follows \\ standard eigenvalue existence results because the initial state of the distribution of prices dies away super-exponentially. 
\end{result}
\begin{result}(\emph{[Theorem 7]})
In the limit where the shock size $\vert \varepsilon \vert \rightarrow 0$, parameter estimators governed by a standard central limit converge at rate $O(1/T)$, rather than $O(1/\sqrt{T})$. 
\end{result}
\par The paper contributes to two domains of pure mathematics. The first is \\ statistical theory whilst the other spans a broader field of analysis. The \\ geometric convergence with known errors contrasts with the $O(1/\sqrt{T})$ pace offered by a standard functional central limit theorem, which would still apply if they were estimated. Consult \cite{billingsley1995probability} Theorem 37.8 for the basic theorem and \cite{bischoff1998functional} for the regression theory.
\par The second is mean field game theory, a joint enterprise between \\ mathematics and allied disciplines instigated by \cite{lasry2007mean} and \cite{caines2006large}. Thus far, the literature has centered on the case of \\ Brownian motion (in continuous time) with bounded objective functions. 
\\ \cite{alvarez2023price} studies impulse response functions in continuous time, where the are normally called Green functions (\cite{teschl2012ordinary}). His results \\ appear similar to discrete time findings that strong externalities are needed to generate any hump shape. They use a different demand system with \\ complementarities absent here.\footnote{The demand curves are displayed as (10) for Calvo, Taylor and Rotemberg and (51) for menu costs, which is derived in detail in Appendix A.1.}
It is well-known that monotonicity \\ assumptions common in the mean field game literature rule out these dynamics, \cite{mou2022mean} offer a synopsis. This and my previous treatise represent the dawn of a new probabilistic approach to mean field games, characterized by universal \emph{a priori} estimates and the potential for deep quantitative existence result with scientific implications. 
\par There are three progressive aspects to the proofs. Theorem 3 engenders statistical deductions from the abstract existence result Theorem 3 (SELCKE). Theorem 2 involves breaking the population into measurable chunks, which allows me to extend previous existence results to the case of idiosyncratic noise, in a fashion that should prove generic. Finally, the limiting distribution behind Theorem 7 may be surprising to some.
\par In the course of my mathematical labors, there are several 
references to \\ results from advanced textbooks, such as \cite{billingsley1995probability} (probability),\\ \cite{hatcher2002algebraic} (algebraic topology), 
\cite{aliprantis2007infinite} and \cite{villani2021topics} (both functional analysis), with the aim of clarifying technical points or validating limiting logic. Nevertheless, the broad thrust of the proof rely on a \\ familiar triumvirate of stochastic process theory, 
refined fixed-point Theorem 1 and classical analysis. Thus, the paper is designed to be mathematically accessible.  
\subsection{Road Map}
Section 3 lays out the benchmark Calvo model. Section 4 is devoted to menu costs. Section 5 concerns the theory of Taylor pricing, Section 6 looks towards applications. Section 7 focuses on foundational convergence behavior. 
Section 8 harnesses these results for numerical analysis and testing. 
Section 9 concludes. 
\par There are four Supplementary Appendices dedicated to lengthy and \\ predictable workings. Section A details further DSGE derivations. Section B is given over to the solution of the Taylor Phillips, without restriction on the degree of nominal rigidity. Section C contains a multitude of 
special cases. Section D focuses on characteristic equation calculations supporting empirical implementation. 
\section{Benchmark Calvo Model}
This section is divided in two. The first 
communicates the underlying structure of the basic \cite{calvo1983staggered} model. The second brings about the solution.  
The switch to persistent errors engenders novelty in the dynamic form, whilst an assumption is removed from the fixed-point theorem in SELCKE.  
\subsection{Framework}
This exposition closely tracks SELCKE, apart from a reduction in detail and an empirically motivated change to the assumptions on the error term. There are two components, the first deals with households and preferences. The second treats pricing and general equilibrium. 
\subsubsection{Households and Preferences}
There is a single representative household that chooses consumption $C$ and labor supply
$L$, so as to maximize the following objective function
\begin{equation}\max_{C_{t},\,  L_{t}} \: U_{t}= {\mathbb{E}_{t}\sum_{T=t}^{\infty}\beta^{T-t}\,\bigg[u(C_{T})-
\nu(L_{T})\bigg ]\,\psi_{T}}\end{equation}
subject to the budget constraint 
\begin{equation} P_{t}\, C_{t}+B_{t+1}=(1+i_{t-1})B_{t}+P_{t}\, (1-\tau_{t})\, W_{t}\, L_{t}+ \int_{0}^{1}{\Pi_{t}(i)}\,{d}i +T_{t} \end{equation}
$\beta \in (0,\, 1)$ is the discount factor, $u$ the utility function and $\nu$ the cost of work. $\psi_{T}$ is the demand shock. It is a preference shock, such that a higher value induces the household to demand more consumption today and less tomorrow. It can be taken to encompass financial shocks. \,$B$ refers to the holding of one period risk free nominal bonds.\,$i_{T}$ is the risk-free nominal interest rate paid at the end of period $T$ on the bond.\,$P$ is the price level - bonds are the numeraire here.\,$W$ is the real wage. $T_{t}$ is a lump sum tax that can be used to fund a lump sum subsidy on wages $\tau_{t}$. This is a welfare feature irrelevant to dynamic analysis. 
\par There is a unit continuum of firms. $\Pi(i)$ is profit from an individual firm $i$, given by
\begin{equation}
\Pi_{t}(i)=p_{t}(i)\,y_{t}(i)-W_{t}\, l_{t}(i)
\end{equation}
The budget constraint states that the uses for nominal income (consumption and saving) must be equal to the sources of income (wealth, labor and dividend income). \par Finally, there are two constraints
\begin{equation}B_{T}=0
\end{equation}
\begin{equation}\lim_{T \rightarrow \infty}  \beta^{T-t} \, \mathbb{E}_{t}\, \psi_{T}\, u'(C_{T})\geq 0\end{equation}
These serve to enforce the dynamic budget constraint and ignore fiscal policy. Preferences are constructed so as to ensure interior solutions. Thus $u'>0$, so agents always wish to consume more and the transversality condition will bind with equality. $u''<0$ to incentivize consumption smoothing. It is costly for agents to work $\nu'>0$,  whilst $\nu''>0$ to encourage the agent to balance work and leisure.
\par Additional conditions are required to rule out boundary solutions. The \\ standard Inada condition for consumption is
\begin{equation}\lim_{C \rightarrow 0}u'(C)= \infty \end{equation} along with zero net wealth. This ensures the representative household will always work. 
To force them to take leisure 
\begin{equation}\lim_{L \rightarrow \bar{L}} \nu( L_{t}) \rightarrow \infty \end{equation}
where $\bar{L}$ is the maximum possible labor supply. I use the functional forms below\footnote{The inconsistency between Inada condition (7) and functional form (9) are discussed alongside remedies in SELCKE Footnote 23. Logarithmic utility is a requirement for the balanced growth in this model, consistent with standard de-trending procedure. It is only required for the policy rule.}
\begin{equation} u(C)=
\log(C)
\end{equation}
\begin{equation}
\nu(L)=\frac{L^{1+\eta}}{1+\eta}
\end{equation}
Each firm produces an individual variety for which demand is given by 
\begin{equation}y_{t}(i)=\bigg(\frac{p_{t}(i)}{P_{t}}\bigg)^{-\theta} Y_{t}\end{equation}
$\theta$ is the elasticity of demand.
There are three partial equilibrium relations stemming from the household's optimization.  \begin{equation}u'(C_{t})=\beta\, (1+i_{t}) \, \mathbb{E}_{t}\, u'(C_{t+1})\frac{\psi_{t+1}}{\psi_{t}}\frac{P_{t}}{P_{t+1}}\end{equation} 
\begin{equation}
u'(C_{t})\, W_{t}= \nu'(L_{t})\end{equation}
The consumption Euler is the household's intertemporal optimization condition balancing the marginal utility return to consumption today with that of the next period. Second is the intra-temporal optimal labor supply constraint equalizing the value of extra consumption with the marginal cost of working. Interest rates and wage rates are equilibriating mechanisms. The goods market clearing condition is simply
\begin{equation}C_{t}=Y_{t}\end{equation}
\subsubsection{Pricing and General Equilibrium}
\cite{calvo1983staggered} pricing is the most popular approach to inject nominal rigidity into a DSGE model. Reoptimization is governed by a stochastic process common across firms. With probability $1-\alpha$ each firm is free to reset its price (at no cost), whilst with probability $\alpha$ it keeps its price fixed and meets demand at its existing price.
Firms reset their prices to maximize the expected present value of profits through the lifetime of the price as follows: 
\begin{equation}\max_{p_{t}^{*}(i)}\mathbb{E}_{t}\sum_{T=t}^{\infty}\alpha^{T-t}\, Q_{t,\, T}\bigg[\frac{p_{t}(i)}{P_{T}}y_{T}(i)-C (y_{T})(i)\bigg] \end{equation}
subject to the individual demand (10). Here  \begin{equation}Q_{t,\, t+k}=\beta ^k \frac{\psi_{t+k}\, u'(C_{t+k})}{\psi_{t}\, u'(C_{t})}\Pi_{t,\, t+k}^{\theta}
\end{equation}
represents the real stochastic discount factor (SDF). It is the risk-adjusted present value of future consumption $k$ periods ahead which depends on the gross rate of inflation
\begin{equation}\Pi_{t,\, t+k}=\frac{P_{t+k}}{P_{t}}=(1+\pi_{t+1}) \cdots (1+\pi_{t+k})\end{equation} between today time $t$ and a future time $T > t$. The first order condition is
\begin{equation}
\mathbb{E}_{t}\sum_{T=t}^{\infty} (1-\alpha)^{T-t}\, Q_{t,\, T}\bigg(\frac{p_{t}^{*}}{P_{T}}\bigg)^{-\theta}Y_{T}\bigg[\frac{p_{t}^{*}}{P_{T}}-\frac{\theta}{\theta-1}MC_{T}(y_{T}(i))\bigg]=0
\end{equation}
The price level evolves as follows: 
\begin{equation}P_{t}^{1-\theta}=\alpha\,  P_{t-1}^{1-\theta}+(1-\alpha)\, (p_{t}^{*})^{1-\theta}
\end{equation}
The reset price can be expressed as
\begin{equation}\frac{p^{*}_{t}}{P_{t}}=\frac{\theta}{\theta -1} \frac{\aleph_{t}}{\beth_{t}}\end{equation}
where 
\begin{equation}\aleph_{t}=\mathbb{E}_{t}\sum_{T=t}^{\infty}(\alpha \, \beta)^{T-t}\, \Pi^{\theta}_{t,\, T}\, \psi_{T}\, u'(C_{T})\, Y_{T}\, MC_{T} \end{equation}
\begin{equation}
\beth_{t}=\mathbb{E}_{t}\sum_{T=t}^{\infty}(\alpha \, \beta)^{T-t}\, \Pi^{\theta-1}_{t,\, T}\, \psi_{T}\, u'(C_{T})\, Y_{T}
\end{equation}
both numerator and denominator have recursive forms 
\begin{equation}
\aleph_{t}=\psi_{t}\, u'(C_{t})\, Y_{t}\, MC_{t}+\alpha \, \beta \, \mathbb{E}_{t}\, (1+\pi_{t+1})^{\theta}\, \aleph_{t+1}
\end{equation}
\begin{equation}
\beth_{t}=\psi_{t}\, u'(C_{t})\, Y_{t}+\alpha \, \beta \, \mathbb{E}_{t}\, (1+\pi_{t+1})^{\theta-1}\, \beth_{t+1}
\end{equation}
Intuitively, $\aleph_{t}$ is a scale-weighted measure of marginal costs with a \\ discounting scheme that reflects the expected age of the price, whilst $\beth_{t}$ is a similarly weighted measure of the scale of future demand.
Nominal rigidity generates real distortions through the dispersion term.  
\begin{equation}\Delta= \int_{i}{\bigg(\frac{p(i)}{P}\bigg)}^{-\theta}\; \mathrm{d}\mu(i) \geq 1\end{equation}
The intuition is that consumers prefer variety and it is therefore costly to \\ substitute between high and low price goods. Therefore they cannot achieve the same utility when prices are dispersed which will always arise when prices are rigid and inflation variable. Here with Calvo pricing, $\Delta$ evolves according to the following relationship: 
\begin{equation} \Delta_{t}=\frac{({1-\alpha\, (1+\pi_{t})^{\theta-1})}^{\theta/(\theta-1)}}{(1-\alpha)^{1/(\theta-1)}} + 
\alpha\, (1+\pi_{t})^{\theta}\, \Delta_{t-1} \end{equation}
derived from repeated use of (18). Labor market clearing stipulates that 
\begin{equation} \Delta_{t}\, C_{t}=L_{t}\end{equation}
where I have normalized productivity to unity.
\par The final piece of the jigsaw is the monetary policy rule 
\begin{equation}i_{t}=i^{*}_{t} + a_{\pi}\, \hat{\pi}_{t}+ a_{y}\, \hat{y}_{t} 
\end{equation}
This so-called Taylor rule is an ad hoc stabilization condition motivated by \\ \cite{taylor1993discretion}.\footnote{His actual proposal contained an inertial element because inflation was measured relative to four quarters back. In the same conference (\cite{henderson1993comparison}) came up with a very similar formulation.} The reaction coefficients $(a_{\pi}, \, a_{y})\geq \bf{0}$, reflecting stabilization \\ motives. This setup is not fit for purpose. Nevertheless, an improvement lies beyond the scope of the paper and all results are robust to a nearby alternative. 
\par Finally, there is a change in the primitive properties of the errors relative to SELCKE. There I made the strict assumption that demand shocks, the only error present in the benchmark model, were white noise. This was to better delineate my discovery of endogenous persistence. Here, I adopt the more \\ conventional position that shocks are autoregressive with one lag, driven by white noise.
\begin{equation} \hat{\psi}_{t}=\rho \, \hat{\psi}_{t-1} +\hat{\xi}_{t}\end{equation}
\par The key point here is that the exogenous shocks peak on impact, so that they cannot cause hump shape dynamics on their own. If errors are regarded as shocks to the financial system, it may be quite natural to regard them as persistent (see for example \cite{guerrieri2017credit}). 
The $n^{th}$ moment of the disturbance process is denoted $\xi^{n}$. In general, the requirements for \\ existence will match the order of the perturbation analyzed. 
Distribution \\ functions will always be continuous.
A fuller understanding of the structure of DSGE disturbances remains an outstanding priority. 
\subsection{Benchmark Solution}
This subsection has two halves. The first crystalizes the critical equilibrium concept of the paper. The second unveils the central Keynesian inflation \\ relationship.
\subsubsection{Stochastic Equilibrium}
The starting point is the stochastic equilibrium of the model, given by the triplet
\begin{equation} \psi \, Y^{-1}= \beta \, (1+i) \, \mathbb{E}\bigg[\frac{\psi\,  Y^{-1} }{(1+\pi)}\bigg]\end{equation}
\begin{multline} \bigg(\frac{1-\alpha}{1-\alpha\, (1+\pi)^{\theta-1}}\bigg)^{1/(\theta-1)}= \\ \frac{\theta}{\theta-1}\bigg(\psi \, \nu'(\Delta \, Y) + \frac{\alpha \, \beta \,\mathbb{E} \, (1+\pi)^{\theta}\, \psi \, \nu'(\Delta \, Y)}{1-\alpha \, \beta \, \mathbb{E}\, (1+\pi)^{\theta}}\bigg) \Bigg/\bigg(\psi + \frac{\alpha \, \beta\, \mathbb{E}\, (1+\pi)^{\theta-1}\, \psi}{1-\alpha \,\beta \, \mathbb{E}\, (1+\pi)^{\theta-1}}\bigg) \end{multline}
\begin{equation}
\Delta=\mathbb{E}\, \Delta=\frac{1}{(1-\alpha)^{1/(\theta-1)}}\frac{\mathbb{E}\, (1-\alpha\, (1+\pi)^{\theta-1})^{\theta/(\theta-1)}}{(1-\alpha \, \mathbb{E} \, (1+\pi)^{\theta})}
\end{equation}
Now we arrive at the linchpin of the paper 
\begin{theorem} Fix a DSGE model with recursive equilibrium $$\mathbb{E}_{t}\, \mathbf{X}_{t+1}= f(\mathbf{X}_{t}, \, \boldsymbol{\gamma}, \, \mathbf{e}_{t}) \: \: \: \mu \: a.s. $$
where  $f \in C^{1}$. Decompose the endogenous variables into jumps (solved forward) and states (solved backwards), such that $\mathbf{X}_{t}=(\mathbf{X}^{J}_{t}, \, \mathbf{X}^{S}_{t})'$. Posit the following mutual dependence pattern  
$$\mathbb{E}_{t}\, \mathbf{X}_{t+1}(i)=f_{i}(\mathbf{X}_{t}(-i), \, \cdot)$$
$$\mathbb{E}_{t}\, \mathbf{X}^{J}_{t+1}= f_{J}(\mathbf{X}^{S}_{t}, \, \cdot)$$
$$\mathbb{E}_{t}\, \mathbf{X}^{S}_{t+1}= f_{S}(\mathbf{X}^{J}_{t}, \,\cdot)$$
which are non-constant $\mu$ a.s.
There exists a stochastic equilibrium if and only if the eigenvalues conform to the Blanchard-Kahn conditions (matching numbers of eigenvalues inside the unit circle with the total of jump variables) for a unique solution around the candidate stochastic steady state. 
\end{theorem}
\begin{proof} Follows directly from Theorem 3 in SELCKE. This result is in fact that theorem minus the demonstration that when the patient limit ($\beta \rightarrow 1$) \\ dominates there can be no other recursive equilibria.\end{proof}
\begin{remark} The mutual dependence condition is trivial for forward-looking \\ optimization models with state variables. It is essential to permit permutation of the eigenvalues and achieve the aggregate condition.\end{remark}
\begin{remark} The smoothness impositions of the theorems require one to accept the premise that central banks can mimic negative nominal interest rates, through mechanisms like quantitative easing. The empirical credence of this postulate is discussed at length in SELCKE. The culprit is aggregate regime-switching; mechanical challenges are covered in \cite{chavez2026structural}.\end{remark}
\subsubsection{Benchmark Phillips Curve}
Stochastic equilibrium approximations will in general contain expectation terms, evaluated at the ergodic invariant measure. For example, here is an expression for the aggregate demand equation 
\begin{multline} \bigg(1+ \beta\,  a_{y}\, \bigg(\frac{\mathbb{E}\, \psi/(1+\pi)\, Y}{\mathbb{E}\, \psi/Y}\bigg)\bigg)\hat{y}_{t}=\bigg(1-\rho \, \psi \, \bigg( \frac{\mathbb{E}\, 1/(1+\pi)\, Y}{\mathbb{E}\, \psi/Y}\bigg)\bigg)\hat{\psi}_{t} - \\ \beta a_{\pi}\bigg(\frac{\mathbb{E}\, \psi/(1+\pi)\, Y}{\mathbb{E}\, \psi/Y}\bigg)\hat{\pi}_{t} + Y\bigg(\frac{\mathbb{E}\, \psi/(1+\pi)\, Y^{2}}{\mathbb{E}\, \psi /(1+\pi)\, Y}\bigg)\mathbb{E}_{t}\, \hat{y}_{t+1} + \\ \bigg(\frac{\mathbb{E}\, \psi/(1+\pi)^{2} \, Y}{\mathbb{E}\, \psi /(1+\pi)\, Y}\bigg)\mathbb{E}_{t}\, \hat{\pi}_{t+1}
\end{multline}
This case which admits high order approximations will be fleshed out in the appendix.
\par Attention now turns to the $\vert{\varepsilon}\vert$, where the model approaches linearity. The labor market clearing condition (combining (12) and (26)) in linear form is 
\begin{equation} \hat{mc}_{t}=(1+\eta)\, \hat{y}
_{t} + \eta \, \hat{\Delta}_{t}\end{equation}
Price dispersion obeys 
\begin{equation} \hat{\Delta}_{t}=\alpha
\, \hat{\Delta}_{t-1}\end{equation}
Linearizing the two components (22) and (23) of the Phillips curve at 
ZINSS reveals a common root $\mathbb{L}=1/\alpha \, \beta$ in both lag polynomials. 
\begin{equation}\hat{\aleph}_{t}=(1-\alpha \, \beta)\bigg\{(1+\eta)\, \hat{y}^{e}_{t}+\eta \, \hat{\Delta}_{t} +  \hat{\psi}_{t}\bigg\} +\alpha \, \beta \, \mathbb{E}_{t}\, \hat{\aleph}_{t+1}\end{equation}
\begin{equation}\hat{\beth}_{t}=(1-\alpha \, \beta) \hat{\psi}_{t} +\alpha \, \beta \, \mathbb{E}_{t}\, \hat{\beth}_{t+1}\end{equation}
This gives rise to a singular surface. The induced cancellation of the error terms makes it a wall-crossing, in a particular a three-dimensional hole. Local \\ dynamics around the steady state are instead described by the limiting \\ stochastic equilibrium, which approaches but does not undertake the spurious cancellation step and agrees with the missing zero limit of the trend inflation Phillips curve, first noted by \cite{ascari2004staggered}.
\par Several simple substitution steps yield the accelerationist form 
\begin{multline} \pi_{t}=\frac{1}{\beta\, (1+\alpha)}\pi_{t-1} + \frac{(1-\alpha)(1-\alpha \, \beta)\, (1+\eta)}{(1+\alpha)}\hat{y}
_{t} + \frac{(1-\alpha)\, (1-\alpha \, \beta)\, \eta}{(1+\alpha)}\hat{\Delta}
_{t}- \\ \frac{(1-\alpha)(1-\alpha \, \beta)\, \eta}{\beta \, (1+\alpha)}\hat{\Delta}
_{t-1} -\frac{(1-\alpha)\, (1-\alpha \, \beta)\, (1+\eta)}{\beta\, (1+\alpha)}\hat{y}^{e}_{t-1} + \frac{\alpha \, \beta^{2}}{(1+\alpha)}\mathbb{E}_{t}\, \pi_{t+1}\end{multline}
The derivation is completed by substituting in the lagged version of aggregate demand curve (32), which at ZINSS is given by 
\begin{equation} (1+ \beta \, a_{y})\,\hat{y}
_{t}=(1-\rho)\, \hat{\psi}_{t} -  \beta \, a_{\pi}\, \pi_{t} + \mathbb{E}_{t}\, \hat{y}_{t+1} + \mathbb{E}_{t}\, \hat{\pi}_{t+1} 
\end{equation}
The symmetry in the error structure emerges from the cohomology induced by the efficiency of ZINSS. The reduction in the error coefficient by a factor of $\rho$ reflects the counteracting effect of the expected future preference shock, inside the Euler equation, absent under the extreme white noise pursued in SELCKE. Here is the form of the equation
\begin{equation}
\hat{\pi}_{t}=b_{0}\, \hat{\pi}_{t-1} +b_{1}\, \hat{y}_{t}+ b_{2}\, \hat{\Delta}_{t} + b_{3}\, \mathbb{E}_{t}\, \hat{\pi}_{t+1}+
b_{4}\, (\hat{\psi}_{t}-\hat{\psi}_{t-1}) \end{equation}
Each coefficient comprises a numerator indicated by a tilde superscript and an equation specific denominator so $b_{0}=\tilde{b}_{0}/b$ and so forth.
\begin{equation} b= 1+ \alpha + \frac{(1-\alpha)^{2}}{\alpha}\frac{(1+\eta)}{1+a_{y}}\end{equation}
\begin{equation} \tilde{b}_{0}= 1+ a_{\pi}\, \frac{(1-\alpha)^{2}}{\alpha}\frac{(1+\eta)}{1+a_{y}}\end{equation}
\begin{equation} \tilde{b}_{1}=(1-\alpha)^{2}\, (1+\eta)\, \bigg\{\frac{\alpha \, (1+a_{y})-1}{\alpha\, (1+a_{y})}\bigg\} \end{equation}
\begin{equation} \tilde{b}_{2}=-\eta\, (1-\alpha)^{3}\frac{(1+\alpha)}{\alpha^{2}}
\end{equation}
\begin{equation} \tilde{b}_{3}=\alpha \end{equation}
\begin{equation} \tilde{b}_{4}(\rho)= (1-\rho)\, \frac{(1-\alpha)^{2}}{\alpha}\frac{(1+\eta)}{1+a_{y}}\end{equation}
where I have sent $\beta \rightarrow 1$. It is the natural setting and it helps to compact the expressions. At the standard calibration $\alpha =2/3$ and $\eta=4$ we find that 
\begin{equation} \pi_{t}=0.575 \, \pi_{t-1}  + 0.25\, \hat{\Delta}_{t} + 0.3\, \mathbb{E}_{t}\, \pi_{t+1} +0.25\, (\hat{\psi}_{t}-\hat{\psi}_{t-1})\end{equation}
There is polydromy so a second "smaller" $\sqrt{\varepsilon}$ limit exists, where price dispersion vanishes. In this case, output developments can be summarized by the efficient output gap $\hat{y}^{e}_{t}$, which moves one-for-one with actual output. In both  instances the policy rule stipulates $a_{\pi} < 1$.\footnote{The coefficient expressions are quite non-linear in inflation, for example
\begin{multline} b = \bigg(\alpha\, (1+\pi)^{\theta-1}\, (2+\pi) + (1-\alpha)^{1/(\theta-1)}\frac{(1-\alpha\, (1+\pi)^{\theta-1})^{(\theta-2)/(\theta-1)}}{(1+\pi)^{\theta-2}} + \\ \frac{(1-\alpha\, (1+\pi)^{\theta-1})}{\alpha\, (1+\pi)^{\theta-2}}(1- (1+\pi)^{\theta})\frac{(1 + \eta)}{1 +  a_{y}}\bigg)- 
\frac{\eta \, \theta \, \pi \, (1-\alpha)^{1/(\theta-1)}} {(1-\alpha\, (1+\pi)^{\theta-1})^{2/(\theta-1)}}
\end{multline}
\begin{equation} \tilde{b}_{0}= 1 + a_{\pi}\, \frac{(1-\alpha\, (1+\pi)^{\theta-1})}{\alpha\, 
(1+\pi)^{\theta-2}}(1-\alpha \, (1+\pi)^{\theta})\frac{(1+ \eta)}{1+ a_{y}}\end{equation}
\begin{equation} \tilde{b}_{3}=\alpha \bigg( \alpha \, (1+\pi)^{2 \theta-1} + \frac{(1-\alpha\, (1+\pi)^{\theta-1})^{\theta/(\theta-1)}}{(1-\alpha)^{1/(\theta-1)}(1+\pi)^{\theta-2}}\bigg)\end{equation}
Thus, there is no guarantee that the qualitative features of the policy rule extend to typical positive rates of trend inflation.} The slope of the Phillips curve is zero \\ because the intertemporal distortions from excess discounting $\alpha < 1$  and the policy induced cost channel $a_{y}>0$ precisely cancels out at the chosen \\ values. This phenomenon is called output neutrality. More reactive policy would guarantee a positive slope.
\section{Menu Costs}
Menu costs provide a crucial bridge between classical economics based on \\ individual optimization and traditional Keynesian understanding that prices and wages respond slowly to shocks. Some notion of costs to making price adjustments is essential to rationalize price stickiness. All the alternatives \\ presented are in some sense a reduced form description of a suitable menu cost framework. An authoritative estimate of these costs by \cite{zbaracki2004managerial} reaches $1\%$ of output.\footnote{Negotiation and communication costs dominate physical expenses. There is empirical evidence that repricing becomes more frequent when inflation increases (\cite{alvarez2019hyperinflation} and \cite{blanco2022evolution}) consistent with the basic predictions of this approach.} Apart from simple cases, only simulation methods have been employed (see \cite{berger2019shocks}). There are 
four parts the first describes the unique features of the model. The second 
contains the main steps towards the stochastic equilibrium solution. The final subsection offers proofs of important properties.  
\subsection{Setup}
\par The setup of the model is that firms face a fixed cost $c$ whenever they change their price. This changes the resource constraint (13) to 
\begin{equation}C_{t}=Y_{t}-\alpha_{t}\, c\end{equation}
where $\alpha_{t}$ is the time varying frequency of price adjustment.
A second change is required. So as to prevent prices becoming flexible, it is necessary to introduce idiosyncratic shocks $b_{t}(i)$. These are independently and identically distributed across firms with mean $b$.\footnote{Technically, support and measurability assumptions are required. It is adequate to have $b(i)$ continuously distributed on $\mathbb{R}_{++}$.} This modifies the individual demand system \\ corresponding to (10) to \begin{equation}c_{t}(i)= \bigg(\frac{b_{t}(i)}{b}\bigg)^{\theta}\bigg(\frac{p_{t}(i)}{P_{t}}\bigg)^{-\theta} (Y_{t}-\alpha_{t}\, c)\end{equation}
These shocks are also sufficient to ensure standard blowup at the boundary \\ conditions.\footnote{Recall Proposition 16 (ii) undergirding Theorem 3 and its application to the Calvo model in Theorem 9 in SELCKE.} 
The consumer demand schedule is fully derived in 
the \\ Supplementary Materials Section A.1.
\par The value of the firm is 
\begin{multline} V_{t}(i)=\max_{p^{*}_{t}(i)}\Bigg\{\psi_{t}\, u'(Y_{t}-\alpha_{t}\, c)\bigg[\bigg(\frac{b_{t}(i)}{b}\bigg)^{\theta}\bigg(\frac{p^{*}_{t}(i)}{P_{t}}\bigg)-MC_{t}\bigg]\bigg(\frac{p^{*}_{t}(i)}{P_{t}}\bigg)^{\theta}(Y_{t}-\alpha_{t}c) \\ -c 
+
\beta \, \mathbb{E}_{t}\, V_{t+1}(p^{*}_{t}(i)), \, \\ \psi_{t}\, u'(Y_{t}-\alpha_{t}\, c)\bigg[\bigg(\frac{p_{t-1}(i)}{P_{t}}\bigg)^{1-\theta}-MC_{t}\bigg](Y_{t}-\alpha_{t}\, c) +  \beta \, \mathbb{E}_{t}\, V_{t+1}(p_{t-1}(i))\Bigg\}\end{multline}
Price dispersion previously (24) is now 
\begin{equation} \Delta=\int_{0}^{1}\bigg(\frac{b(i)}{b}\bigg)^{-\theta}\bigg(\frac{p(i)}{P}\bigg)^{-\theta}\, \mathrm{d}i\end{equation}
let \begin{equation} \Delta_{P}=\int_{0}^{1}\bigg(\frac{p(i)}{P}\bigg)^{-\theta}\, \mathrm{d}i\end{equation} be its pure price component. It follows that 
\begin{proposition} Idiosyncratic shocks increase price dispersion in particular $\Delta \geq \Delta_{P}$ with equality, if and only if $b(i)=b$, $\mu$ a.s. in $i$. Thus $\Delta >\Delta_{P}$, $\mu$ a.s.\end{proposition}
\begin{proof} A straightforward consequence of Jensen's inequality and the process of the shock properties, mimicking the steps in Proposition 2 (SELCKE). \end{proof}
This result readily applies to other pricing models. 
Note that the resource cost \begin{equation} \Delta_{b}=\int_{0}^{1}\bigg(\frac{b_{t}(i)}{b}\bigg)^{-\theta}\, \mathrm{d}\mu(i) \geq 1\end{equation} is not a figment of price rigidity but is intrinsic to the love of variety and monopolistic competition. More discussion arrives with Proposition 3.
\subsection{Equilibrium Construction}
The following characterization is the main contribution 
Although rigorous, the proof is circuitous, with space left for other pertinent observations and results. The recursive equilibrium form left implicit here, is deferred to Section 6.2, to aid comparison with other pricing schemes. 
\begin{proposition} In Stochastic Equilibrium, the menu cost model is equivalent to a Calvo with endogenous reset probabilities.\end{proposition} 
\begin{proof} 
At first some of the equations get more complicated, eventually they will simplify. It is fundamental to appreciate at the outset that the firms decision rule takes a so called s-S form, where the firm raise its price to the desired price $p_{t}^{*}(i)$ when it exceeds a threshold $\bar{S}_{t}(i)$ and cuts it likewise when $p_{t}^{*}(i)<\underbar{s}_{t}(i)$. The price level evolution (18) becomes 
\begin{multline}
P_{t}^{1-\theta}=
\int_{\underbar{s}_{t}(i)}^{\bar{S}_{t}(i)}\bigg(\frac{b_{t}(i)}{b}\bigg)^{\theta}(p_{t-1}(i))^{1-\theta}\, \mathrm{d}p_{t-1}(i)
+  \\ \bigg(\int_{\bar{S}_{t}(i)}^{\infty}+ \int_{0}^{\underbar{s}_{t}(i)}\bigg)\, \bigg(\frac{b_{t}(i)}{b}\bigg)^{\theta}(p_{t}^{*}(i))^{1-\theta}\, \mathrm{d}p_{t-1}(i)\end{multline}
where integrals are taken over the prevailing price distribution $\mathrm{d}p_{t-1}(i)$. In stochastic equilibrium, the relative price distribution is fixed modulo trend inflation. Hence, 
\begin{multline}1=
(1+\pi)^{\theta-1}\int_{\underbar{s}(i)}^{\bar{S}(i)}\bigg(\frac{b(i)}{b}\bigg)^{\theta}
\bigg(\frac{p(i)}{P}\bigg)^{1-\theta}\, \mathrm{d}p(i)
+  \\ \bigg(\int_{\bar{S}(i)}^{\infty}+ \int_{0}^{\underbar{s}(i)}\bigg)\, \bigg(\frac{b(i)}{b}\bigg)^{\theta}\bigg(\frac{p^{*}(i)}{P}\bigg)^{1-\theta}\, \mathrm{d}p(i)\end{multline}
There is no tendency for selection at the ergodic invariant measure, thus
\begin{equation} \int_{\underbar{s}(i)}^{\bar{S}(i)}\bigg(\frac{b(i)}{b}\bigg)^{\theta}
\bigg(\frac{p(i)}{P}\bigg)^{1-\theta}\, \mathrm{d}p(i)=\alpha
\end{equation}
Therefore, the reset price equation connecting the reset price distribution to inflation agrees with Calvo pricing modulo an allowance for heterogeniety\footnote{This equation corresponds to (32) in the more detailed exposition in SELCKE.}
\begin{equation}1-\alpha\, 
(1+\pi)^{\theta-1}=\bigg(\int_{\bar{S}(i)}^{\infty}+ \int_{0}^{\underbar{s}(i)}\bigg)\, \bigg(\frac{b(i)}{b}\bigg)^{\theta}\bigg(\frac{p^{*}(i)}{P}\bigg)^{1-\theta}\, \mathrm{d}p(i)\end{equation}
The stochastic equilibrium Phillips curve recursion takes the form 
 \begin{multline}  
 \frac{p^{*}_{t}(i)}{P_{t}}= \frac{\theta}{\theta-1}\bigg( \frac{\psi \, Y}{Y-\alpha\,  c}\bigg(\frac{b(i)}{b}\bigg)^{\theta}\nu'(\Delta \, Y) + \\ \beta \, \bigg\{\mathbb{E}\, \sum_{T=t+1}^{\infty}\alpha_{t, \, T}\int_{\bar{b}}\bigg(\frac{b(i)}{b}\bigg)^{\theta}
 f(i)\, \mathrm{d}b(i)\bigg\}
 \mathbb{E}\bigg[\frac{\psi \, Y}{Y-\alpha \, c}\frac{(1+\pi)^{\theta}\, \nu'(\Delta \, Y)}{\big\{1-\alpha \, \beta \, \mathbb{E}\, (1+\pi)^{\theta}\big\}}\bigg]\bigg) 
 \\ \Bigg/\bigg(\frac{\psi \, Y}{Y-\alpha \, c}\bigg(\frac{b(i)}{b}\bigg)^{\theta} + \\ \beta\bigg\{\mathbb{E}\sum_{T=t+1}^{\infty}\alpha_{t, \, T}\int_{\bar{b}}\bigg(\frac{b(i)}{b}\bigg)^{\theta}
 f(i)\, \mathrm{d}b(i)\bigg\} 
 \mathbb{E}\bigg[\frac{\psi \, Y}{Y-\alpha \, c}\frac{(1+\pi)^{\theta-1}}{\big\{1-\alpha \, \beta \, \mathbb{E}\, (1+\pi)^{\theta-1}\big\}}\bigg]\bigg)
 \end{multline} 
 here $\bar{b}$ represents the values of $b(i)$ consistent with the individual price staying fixed at its invariant distribution $f(i)$. 
 Note the following dichotomy
 \begin{remark} Consider the second term of the numerator and denominator. Down the stochastic equilibrium path, the first parenthesis represents terms \\ determined by the parameters and realizations of the idiosyncratic shock process and any trend inflation, whilst the subsequent product is impacted by aggregate shocks, entirely in the case of the denominator. \end{remark}
 Using (58), the stochastic equilibrium Phillips curve recursion takes the form 
\begin{multline}
 \frac{(1-\alpha)^{1/(\theta-1)}}{(1-\alpha \, (1+\pi_{t})^{\theta-1})^{1/(\theta-1)}}= \frac{\theta}{\theta-1}\bigg(\frac{\psi \, Y\, \nu'(\Delta \, Y)}{Y-\alpha \, c}\bigg\{\int_{{b}^{*}}\bigg(\frac{b(i)}{b}\bigg)^{\theta}
 f(i)\, \mathrm{d}b(i)\bigg\}  \\ + \beta \bigg\{\mathbb{E}\sum_{T=t+1}^{\infty}\alpha_{t, \, T}\int_{\bar{b}}\bigg(\frac{b(i)}{b}\bigg)^{\theta}
 f(i)\, \mathrm{d}b(i)\bigg\}
 \mathbb{E}\, \bigg[\frac{\psi \, Y}{Y-\alpha \, c}\frac{(1+\pi)^{\theta}\, \nu'(\Delta \, Y)}{\big\{1-\alpha \, \beta \, \mathbb{E}\,(1+\pi)^{\theta}\big\}}\bigg]\bigg)  \Bigg/
 \\ \bigg(\frac{\psi Y}{Y-\alpha c}\bigg\{\int_{{b}^{*}}\bigg(\frac{b(i)}{b}\bigg)^{\theta}
 f(i)\, \mathrm{d}b(i)\bigg\} +  \\ \beta\bigg\{\mathbb{E}\sum_{T=t+1}^{\infty}\alpha_{t, \, T}\int_{\bar{b}}\bigg(\frac{b(i)}{b}\bigg)^{\theta}
 f(i)\, \mathrm{d}b(i)\bigg\}\mathbb{E}\, \bigg[\frac{\psi \, Y}{Y-\alpha \, c}\frac{(1+\pi)^{\theta-1}}{\big\{1-\alpha \, \beta \, \mathbb{E}\, (1+\pi)^{\theta-1}\big\}}\bigg]\bigg)
 \end{multline} 
It is now time to further address the integrals over the individual error \\ distribution
\begin{remark}At the Stochastic Equilibrium individual relative prices are monotonic in $b_{t}(i)$. Near to ZINSS, they are in fact decreasing with the shock, the firm trades off a lower margin for greater scale. This need not be true if there were deflation or persistence. When marginal costs are elevated the relationship is always increasing. In any case, there are different aggregators allowing for alternative properties, such as \cite{kimball1995quantitative}. There is a singularity at ZINSS, reflecting the collapse of the reset probability. 
\end{remark}
Finishing off the derivation, there is a mechanical recursion connecting the different hazard rates 
\begin{equation} \alpha = (1-\alpha)(\alpha_{1} + \alpha_{1}\, \alpha_{2} + \alpha_{1}\, \alpha_{2}\, \alpha_{3} + \cdots +)\end{equation}
where $\alpha_{i}$ is the hazard rate of a price at age $i$. 
\begin{remark} The hazard profile is flat when there is no trend inflation as in Calvo. This demonstrates that trend inflation is required to fit evidence of hazard rates increasing with age (\cite{nakamura2013price}). \end{remark}
Optimal reset is determined by comparison between profits when the price is adjusted and the cost paid versus the case where the initial price stays fixed. $\alpha$ is determined by an optimal comparison of expected profit, when the price is kept fixed compared to when it varies. Copious use is made of no selection principles implied by (57) and (58). 
 \par The expected value of sticky prices is 
 \begin{multline} \bar{V}_{t}(i)=\psi_{t}\, u'(C_{t})\, Y_{t} \int_{\bar{b} _{t}(\alpha^{*} _{t})}\bigg(\frac{b(i)}{b}\bigg)^{\theta}
 f(i)\, \mathrm{d}b(i) \times \\ \Bigg[\bigg(\frac{p_{t-1}(i)}{P_{t}}\bigg)^{1-\theta}-   \frac{\theta}{\theta-1}\bigg(\frac{p_{t-1}(i)}{P_{t}}\bigg)^{-\theta} \nu'(\Delta_{t} \, Y_{t})\Bigg]
 + \\
\mathbb{E}_{t}\sum_{T=t+1}^{\infty}\alpha^{*}_{t, \, T}\, \beta^{T-t}\psi_{T}u'(C_{T})Y_{T}\int_{\bar{b} _{t}(\alpha^{*} _{t})}\bigg(\frac{b(i)}{b}\bigg)^{\theta}f(i)\, \mathrm{d}b(i) \times \\ \Bigg[\bigg(\frac{p_{t-1}(i)}{P_{T}}\bigg)^{1-\theta}-  \frac{\theta}{\theta-1}\bigg(\frac{p_{t-1}(i)}{P_{T}}\bigg)^{-\theta}\nu'(\Delta_{T} \, Y_{T})\Bigg]
\end{multline}
In stochastic equilibrium, this simplifies to 
\begin{multline} \mathbb{E}\bar{V}(i)=\alpha^{*}\Bigg\{ \psi \, u'(C)\, Y \, (1+\pi)^{\theta-1}\Bigg[1-\frac{\theta}{\theta-1} (1+\pi) \, \Delta \, \nu'(\Delta \, Y)\Bigg]
+ \\ \mathbb{E}\sum_{T=t+1}^{\infty}\alpha^{*}_{t, \, T}\, \beta^{T-t}\, \psi \, u'(C)\, Y \, (1+\pi)^{\theta-1}\Bigg[1- \frac{\theta}{\theta-1}(1+\pi)\, \Delta \, \nu'(\Delta \, Y)\Bigg]\Bigg\}
 \end{multline}
Note that ${*}$ is used to indicate the optimizing choice of reset frequency. 
\par On the other hand, the value of reset prices is\footnote{The impact of nominal rigidity is indeed discernible in stock prices (see \cite{gorodnichenko2016sticky} and \cite{faia2024cost}).} 
\begin{multline}{V}^{*}_{t}(i)=\psi_{t} \, u'(C_{t})\, Y_{t} \int_{\bar{b}_{t} (\alpha^{*}_{t})}\bigg(\frac{b(i)}{b}\bigg)^{\theta}f(i)\, \mathrm{d}b(i) \times \\ \bigg[\bigg(\frac{p^{*}_{t}(i)}{P_{t}}\bigg)^{1-\theta}-\frac{\theta}{\theta-1} \bigg(\frac{p^{*}_{t}(i)}{P_{t}}\bigg)^{-\theta}v'(\Delta_{t} \, Y_{t})\bigg]  +  \\ \mathbb{E}_{t}\sum_{T=t+1}^{\infty}\alpha^{*}_{t, \, T}\, \beta^{T-t}\, \psi \, u'(C)\, Y \int_{\bar{b}_{t}(\alpha^{*}_{t})}\bigg(\frac{b(i)}{b}\bigg)^{\theta}
 f(i)\, \mathrm{d}b(i) \times \\ \bigg[\bigg(\frac{p^{*}_{t}(i)}{P_{T}}\bigg)^{1-\theta}-\frac{\theta}{\theta-1} \bigg(\frac{p^{*}_{t}(i)}{P_{T}}\bigg)^{-\theta}\Delta_{T} \, v'(\Delta_{T} Y_{T})\bigg] -(1-\alpha_{t})\, c
\end{multline}
On average, we find 
\begin{multline}\mathbb{E}\, {V}^{*}(i)=\psi \, u'(C)\, Y
\bigg[1-\alpha^{*} \, (1+\pi)^{\theta-1}-(1-\alpha^{*})\frac{\theta}{\theta-1}\Delta \, v'(\Delta \, Y)\bigg]  + \\  \mathbb{E}\sum_{T=t+1}^{\infty}\, \alpha^{*}_{t, \, T}\, \beta^{T-t}\, (1+\pi)^{\theta-1}\, \psi \, u'(C)\, Y \times \\  \bigg[(1-\alpha^{*} \, (1+\pi)^{\theta-1})
 - (1-\alpha^{*})\frac{\theta}{\theta-1}
 (1+\pi)\, \Delta \, v'(\Delta \, Y)\bigg] -(1-\alpha)\, c
\end{multline}
The incentives in play are that a higher reset probability increases the expected stream of profits, at the expense of paying the menu cost. Optimization with respect to $\alpha^{*}$ is recursive and hence depends on all the variables of the standard Calvo Phillips curve.
The optimality condition is 
\begin{equation}\partial \bar{V}_{t}(i)/ \partial \alpha^{*} + \partial V^{*}_{t}(i)/ \partial \alpha^{*} = 0 \end{equation}
This implies a stochastic equilibrium reset probability 
$\alpha(\pi, \, Y, \, \cdot )$.
The hazard probabilities are then built by forward recursions of (61), the Phillips curve and the aggregate demand equation. This yields a hazard profile dependent on the development of inflation during the price spell  $\alpha_{t,\, T}(\pi_{t}, \, \cdots, \, \pi_{T}, \, Y, \, \cdot)$, which induces dependence on the time horizon $T-t$ and also the shape of the shock distributions. 
\par Analysis is cleanest at zero trend inflation, where the constant hazard rate ensures that 
\begin{multline} \mathbb{E}\, \bar{V}(i)=\alpha\Bigg\{ \psi \, u'(C)\, Y \bigg(1-\frac{\theta}{\theta-1} \Delta \, \nu'(\Delta \, Y)\bigg)
+ \\ \frac{\alpha \, \beta}{1-\alpha \, \beta}\mathbb{E}\, \psi \, u'(C)\, Y\,  (1+\pi)^{\theta-1}\bigg(1- \frac{\theta}{\theta-1}(1+\pi)\, \Delta \, \nu'(\Delta \, Y)\bigg)\Bigg\}\end{multline}
\begin{multline}\mathbb{E}\, {V}^{*}(i)=(1-\alpha)\, \psi \, u'(C)\, Y\, 
\bigg(1-\frac{\theta}{\theta-1}\Delta \, v'(\Delta \, Y)\bigg)  + \\ 
\frac{\alpha \, \beta}{1-\alpha \,\beta}\mathbb{E}\, \psi \, u'(C)\, Y \, (1+\pi)^{\theta-1}  \bigg[(1-\alpha \, (1+\pi)^{\theta-1})- \\ (1-\alpha)\frac{\theta}{\theta-1}(1+\pi)\, \Delta \, v'(\Delta \, Y)\bigg)\bigg] -(1-\alpha)\,c
\end{multline}
Hence, the long-run equilibrium condition is 
\begin{multline} \frac{\alpha \, \beta}{(1-\alpha \, \beta)^{2}}\mathbb{E}\, \psi \, u'(C)\, Y \, (1+\pi)^{\theta-1}\bigg(1- \frac{\theta}{\theta-1}(1+\pi)\, \Delta \, \nu'(\Delta \, Y)\bigg) + \\ \frac{1}{(1-\alpha \, \beta)^{2}}\mathbb{E}\, \psi \, u'(C)\, Y \, (1+\pi)^{\theta-1}  \times \\ \bigg[(1-\alpha \, (1+\pi)^{\theta-1})- (1-\alpha)\frac{\theta}{\theta-1}(1+\pi)\, \Delta \, v'(\Delta \, Y)\bigg)\bigg] -\\ \frac{\alpha \, \beta}{1-\alpha \, \beta}\mathbb{E}\, \psi \, u'(C)\, Y \, (1+\pi)^{\theta-1}  \, \bigg[ (1+\pi)^{\theta-1}- \frac{\theta}{\theta-1}(1+\pi)\, \Delta \, v'(\Delta \, Y)\bigg)\bigg] +\, c =0\end{multline}
\subsection{Theoretical Implications}
This conclusion of our state dependent pricing act has 
three scenes. The first completes the construction of stochastic equilibrium 
alongside consideration of the relationship between idiosyncratic heterogeneity and aggregate efficiency. The 
second business is to carry over a piece of mathematical analysis contained in SELCKE from Calvo to menu cost framework and briefly discuss prospects for improved solutions. The final heading rules out small noise limiting solutions.
\subsubsection{A Prioi Properties}
The final component of the equilibrium characterization of general interest. 
\begin{proposition} $\Delta= \Delta_{b}\, \Delta_{P}$ \end{proposition}
\begin{proof} An immediate consequence of Birkhoff's ergodic theorem, the aggregate price distribution is fixed $\mu$ a.s. then integrate 
over the individual shock \\ distribution.\end{proof}
Therefore, our equilibrium conditions are completed by. 
\begin{equation}
\Delta=\frac{\Delta_{b}}{(1-\mathbb{E}\, \alpha \, (1+\pi)^{\theta})}\mathbb{E}\, \bigg[\frac{(1-\alpha\, (1+\pi)^{\theta-1})^{\theta/(\theta-1)}}{(1-\alpha)^{1/(\theta-1)}}\bigg]
\end{equation}
where the only material change over $\Delta_{P}$ from (31) is once again a response to the endogenous frequency of price adjustment. 
In view of (26), Proposition 3 embodies the following generality 
\begin{principle} Microeconomic inefficiencies, such as inertia in the face of idiosyncratic shocks, appear like static wedges between actual and efficient output.\end{principle}
Finally, the dynamic characterization is relatively tedious. 
Once the \\ economy is close to stochastic equilibrium, it tracks Proposition 4 in SELCKE with one addition. This is the simply differentiating (70) with respect to $\mathbb{E}\alpha$ and noting that it is strictly positive.\footnote{This is because the expectation of frequent resetting reduces the value of profit foregone by not changing price today.} 
Formally, the recursive equilibrium is $\mathbb{E}_{t}\, Z_{t+1}=f(Z_{t}, \, U_{t}, \, \cdot)$ $\mu$ a.s. where $Z_{t}=(\pi_{t}, \, Y_{t}, \, \alpha_{t}, \, \Delta_{t}, \, \pi_{t-1}, \, \alpha_{t-1})'$ and $
U_{t}=(\psi_{t}, \, \psi_{t-1})'$. 
The proof of Proposition 2 is complete. 
\end{proof}
\par This immediately implies that the limiting equivalence between Calvo and Taylor pricing established in SELCKE extends to menu costs. Therefore, there is a full universality amongst the suite of sticky price models. 
Principle 1 is intuitive. It fits our understanding from non-stochastic theory and is surely more widely applicable. This derivation is surely testament to the simplifying power of stochastic equilibrium, complicated initial dynamics can die away to leave a parsimonious form. 
\par There are two further consequences of interest, both possess a topological flavor. 
\begin{corollary} Calvo and menu costs share a common wall-crossing singularity at ZINSS except the dimension of the hole is one higher. 
\end{corollary}
\begin{proof} This follows swiftly from Proposition 3 and Theorem 7 in SELCKE, where it is established that the size of the hole reflects the number of common roots arising at ZINSS. Proposition 2 implies the inflation singularity $\pi_{t-1}=\beta \, \pi_{t}$ (from Proposition 20 in SELCKE) will be inherited from Calvo. In addition, there will be one in $\alpha$, associated with the $\alpha_{t-1}$ term, introduced when lagging of the recursive form (35) is undertaken. \end{proof}
\begin{remark} The size of the hole in the model presented here would be two for Calvo and three for menu. This differs from the small noise limits in the previous section, where the shock terms cancel.
\end{remark}
This result implies that the bifurcation analysis and econometric duality apply. It also extends the equivalence noted with the Calvo Wage setting and means the same results would apply with menu costs in wages. 
A full analysis of that model and its Phillips curve is beyond these pages. 
\begin{proposition}
Stochastic Equilibrium is the site of a two-dimensional reduction in the recursive equilibrium.
\end{proposition}
\begin{proof}
This is a consequence of two stochastic equilibrium conditions, contained within (59), which replace the aggregate reset distribution with inflation rates. In particular,  
\begin{equation}1-\alpha_{t}\,
(1+\pi_{t})^{\theta-1}=\bigg(\int_{\bar{S}_{t}(i)}^{\infty}+ \int_{0}^{\underbar{s}_{t}(i)}\bigg)\, \bigg(\frac{b_{t}(i)}{b}\bigg)^{\theta}\bigg(\frac{p^{*}_{t}(i)}{P_{t}}\bigg)^{1-\theta}\, \mathrm{d}p_{t}(i)\end{equation}
\begin{multline}1-\alpha_{t-1}\, 
(1+\pi_{t-1})^{\theta-1}= \\ \bigg(\int_{\bar{S}_{t-1}(i)}^{\infty}+ \int_{0}^{\underbar{s}_{t-1}(i)}\bigg)\, \bigg(\frac{b_{t-1}(i)}{b}\bigg)^{\theta}\bigg(\frac{p^{*}_{t-1}(i)}{P_{t}}\bigg)^{1-\theta}\, \mathrm{d}p_{t-1}(i)\end{multline}
\end{proof}
\begin{remark} The meeting is transversal (\cite{hatcher2002algebraic}) because every singularity is one-to-one and onto. From an alternative algebraic standpoint, the fact that they are not intertemporal avoids spurious Zariski reduction. These features ensure that the stochastic equilibrium is smooth and thus amenable to perturbation analysis.\end{remark}
This validates the abstract 
proclamation about the capacity for stochastic equilibrium to streamline dynamics. 
\subsubsection{Existence Result}
\begin{theorem} Consider the menu cost model defined here in the limit, where $\beta \rightarrow 1$, there exists a recursive equilibrium if and only if 
the eigenvalue condition stipulated in Theorem 1 is met. 
\end{theorem}
\begin{proof}The proof is an interplay between Proposition 16 (SELCKE) and \\ Theorem 1 here. By Proposition 2, all that is required is to verify blowup at the boundary conditions. The arguments in my seminal paper map over for \\ inflation, output and price dispersion with a modification that the upper bound on inflation now depends on the endogenous reset probability $\alpha$. 
\par This only leaves sections of the boundary in the two new variables, the reset price probability $\alpha_{T}$ and the price distribution $\mathrm{d}\mu_{i}$. The latter consideration arises despite its influence decaying away in stochastic equilibrium. This is because it arises in recursive equilibrium at every finite time as a determinate of the reset rate. 
\par The task is to prove that it has to be ergodic. It is sufficient to show that this is true for the price of (almost) any particular firm $p_{i}$. This is accomplished by noting that if any positive share of firms $\mathrm{d}\mu_{j}$ had prices that tended to $0$, it would send $\pi \rightarrow -100\%$. On the other hand, if they went to $\infty$, inflation would move to its positive extreme. Both of these cases have been covered already.
\par 
Finally, turning to the reset rate. The unbounded property of the \\ idiosyncratic shock support allows us to complete the proof by noting that $\alpha_{T} \rightarrow 0$ implies $\pi_{T} \rightarrow -100\%$ or $Y_{T} \rightarrow 0$, while $\alpha_{T} \rightarrow 1$ requires that $\pi_{T} \rightarrow \bar{\pi}$, driving $Y_{T} \rightarrow 0$. In all instances, Inada condition (6) sends the objective function to $-\infty$.\end{proof}
\begin{remark}
The desired constellation is normally referred to as the Blanchard-Kahn conditions by macroeconomists, after (\cite{blanchard1980solution}). \end{remark}
\begin{remark} 
The techniques presented in SELCKE Theorem 3 can be used to generalize the result to a broader family of preferences. \end{remark}
There are no concrete existence results, without detailed analysis of the function space supporting a solution. 
\subsubsection{Non-Existence Result}
Unlike with Calvo and Taylor it is not possible to drop down to a small noise limit because of the presence of unit roots in this vicinity. This is formalized by 
\begin{proposition} Consider the menu cost model of this section where $\pi=0$ and $\beta \rightarrow 1$. With aggregate and idiosyncratic shock sizes $\vert{\varepsilon} \vert$ and $\vert{\varepsilon}_{b}(i)\vert$. There exists no recursive equilibrium in which both $\vert{\varepsilon}\vert \rightarrow 0$ and $\vert{\varepsilon}_{b}(i)\vert \rightarrow 0$. \end{proposition}
\begin{proof} In pursuit of a contradiction, focus first on $\alpha$. Consider the optimal repricing condition (70), it is clear that $\alpha=1$ delivers the highest possible profit to the firms. Hence, the De Rham Cohomology tells us that $\hat{\alpha} \rightarrow 0$, as $\alpha \rightarrow 1$. Thus, $\alpha$ is second-order and first-order dynamics converge on the Calvo model. I learn from Theorem 9 in SELCKE that in general its eigenvalues are 
\begin{multline}\lambda_{1}= 1+\frac{1}{2}\, a_{y} + \frac{(1-\alpha)^{2}}{2\alpha}\, (1+\eta) + \\ \frac{1}{2}\sqrt{\bigg( a_{y} + \frac{(1-\alpha)^{2}}{\alpha}\, (1+\eta)\bigg)^{2} + 4(1-a_{\pi})\, \frac{(1-\alpha)^{2}}{\alpha}(1+\eta)} \end{multline}
\begin{multline}\lambda_{2}= 1+\frac{1}{2}\, a_{y} + \frac{(1-\alpha)^{2}}{2\alpha}\, (1+\eta) - \\ \frac{1}{2}\sqrt{\bigg(a_{y} + \frac{(1-\alpha)^{2}}{\alpha}\, (1+\eta)\bigg)^{2} + 4(1-a_{\pi})\, \frac{(1-\alpha)^{2}}{\alpha}\, (1+\eta)}\end{multline}
\begin{equation} \lambda_{3}=\alpha \end{equation}
\begin{equation} \lambda_{4}=\frac{1}{\alpha}\end{equation}
It is clear that $\lambda_{1}, \, \lambda_{3}$ and $\lambda_{4}$ all tend to unity as $\alpha \rightarrow 1$. $\lambda_{2}$ is bounded outside the unit circle so long as $a_{y} >0$. $\lambda_{2} \rightarrow 1$ also but its effect is localized such that it only affects price dispersion $(\Delta)$ in (31). It will be treated in Remark 8. Taking a Taylor expansion, it is clear that $\lambda_{2}=1 - O(\vert \varepsilon\vert^{2})$, whilst $\lambda_{4}=1 + O(\vert \varepsilon\vert^{2})$. This reduces to an absurdity when I compute the series solution of $\mathbf{\hat{X}}_{t}=(\pi_{t}, \, \hat{y}^{e}_{t})'$ as 
$$\vert \mathbf{\hat{X}}_{t}\vert=\frac{1}{(1-\lambda_{2})}O(\vert \mathbf{\hat{e}}_{t}\vert)=\frac{O(\vert \varepsilon\vert)}{O(\vert \varepsilon\vert^{2})}=O(1/\vert \varepsilon \vert)>>O(\vert \varepsilon \vert)$$
where I have used a special case of Result 2, expressed formally as Theorem 4 in Section 7.1.
\end{proof}
Further clarification and explanation are in order.
\begin{remark} This divergence serves to rule out any perturbation solution \\ defined in this neighborhood. \end{remark}
\begin{remark} In the proposed equilibrium, as $\alpha$ tends to unity, its eigenvalue pair $(\lambda^{J,\,  M}_{n}, \, \lambda^{S, \, M}_{n}) \rightarrow (\pm \infty, \, 0)$ as $n \rightarrow \infty$, consistent with the dieing away of first-order dynamics. Price dispersion would be $O(1)$ and thus its influence would vanish asymptotically. There would be no polydromy.\end{remark}
\begin{remark} If there were no idiosyncratic noise, the price level would \\ retain the rigidity of the individual firms. There are alternative formulations of menu costs that get around this problem but they have counterfactual predictions concerning aggregate nominal rigidity, see respectively \cite{williamson2010new} (theory) and \cite{gorodnichenko2016sticky} (empirics).\end{remark}
Intuitively, noise is intrinsic to a well-behaved solution, at least when agents are patient.\footnote{I am skipping over the issue that in application, there are problems with the \\ disappearance of price flexibility. Empirically, it is difficult to defend $\alpha > 0.8$ (SELCKE) Appendix H.1.3. Usually, these models are used to fit moments of real-world price change distributions and calculate costs of sticky price adjustment.} 
I leave the more complicated case where $\beta < 1$ to others. 
\par Numerical evidence of stable simulations suggests that existence may not pose a problem in the case of passive monetary policy, as the papers cited in the introduction all use money supply shocks abstracting from systemic policy reaction. This model would be an excellent test bed for the new algorithms derived in Section 8. Mathematical analysis may have to wait until after the basic Calvo model of Section 3 is first studied, although there should
be some similarity thanks to the equivalence Proposition 2. 
\section{Taylor Contracts: Theory}
This is the first of two sections on Taylor pricing, a model of nominal rigidity where prices last for fixed intervals. The analysis
builds off Appendix G.1 in SELCKE but there is richer model development, including price dispersion and explicit error terms. This verse is divided into three stanza. The first describes salient characteristics of the general setting. I then move to the long-run. The third portion narrows in on the non-stochastic steady state.
\subsection{Environment}
\par The goal is to analyze the recursive equilibrium. The key difference from the Calvo and menu cost frameworks is that 
price-setters know the length of contract for sure, in advance. 
Contracts overlap so that in each period the same fraction comes up for renewal. This setup seems intuitively appealing.  Studies, dating back at least to \cite{alvarez2006sticky}, all find there is a large share of firms who reprice periodically. Supporting Supplementary Appendix A solves out for the general case of the Phillips curve around ZINSS. 
\par The fundamental optimization problem comparable to (17) for Calvo is 
\begin{equation}
\max _{p_{t}^{*}}\, \mathbb{E}_{t}\sum_{T=t}^{t + M} Q_{t,\, T}\bigg(\frac{p_{t}^{*}}{P_{T}}\bigg)^{-\theta}Y_{T}\bigg[\frac{p_{t}^{*}}{P_{T}}-C (y_{T})(i)\bigg]=0
\end{equation}
From which springs forth the backbone of the model
\begin{equation} \frac{p_{t}^{*}}{P_{t}}= \frac{\theta}{\theta-1}\bigg( \frac{\psi_{t}\nu'(\Delta_{t} Y_{t}) + \beta \sum_{T=t+1}^{t+M-1}  \mathbb{E}_{t}\, \psi_{T} \, \nu'(\Delta_{T} Y_{T})\prod^{T}_{j=t+1}(1+\pi_{j})^{\theta}}{\psi_{t} + \beta \sum_{T=t+1}^{t+M-1} \mathbb{E}_{t}\, \psi_{T} \, \prod^{T}_{j=t+1} (1+\pi_{t+1})^{\theta-1} }\bigg)\end{equation}
As ever, optimal pricing balances expected marginal revenue in the numerator against expected marginal costs in the denominator, suitably weighted to reflect time preference and probabilistic projections for inflation over the duration of the contract. 
\par Price dispersion is given by 
\begin{equation}\Delta_{t}=\frac{1}{M}\bigg( \frac{p^{*}_{t}}{P_{t}}\bigg)^{-\theta} + \frac{1}{M}\bigg( \frac{p^{*}_{t-1}}{P_{t}}\bigg)^{-\theta} + \cdots + \frac{1}{M}\bigg( \frac{p^{*}_{t-(M-1)}}{P_{t}}\bigg)^{-\theta}\end{equation}
whilst the price level is aggregated as follows 
\begin{equation} P^{1-\theta}_{t}=\frac{1}{M}(p^{*}_{t})^{1-\theta}
+ \frac{1}{M}(p^{*}_{t-1})^{1-\theta} + \cdots + \frac{1}{M}(p^{*}_{t-(M-1)})^{1-\theta}
\end{equation}
The other supply side conditions are common to all models. The semi-group is exemplified below. 
\par 
\begin{proposition}
An economy with Taylor contracts 
of length $M \geq 2$ periods, 
resides on a $(5M-2)$-dimensional manifold, that takes the form
$$\mathbb{E}_{t}Z_{t+1}= f(Z_{t}, \, \gamma, \, U_{t}) \: \: \: \mu \: a.s. $$
where
\begin{enumerate}[(i)]
\item  $Z_{t}=(Z^{J}_{t}, \, Z^{S}_{t} )'$, such that $\mathbf{Z}^{J}_{t}=(\pi_{t}, \, \pi_{t-1}, \,  \cdots, \, \pi_{t-(M-2)}, \,  Y_{t})'$ and \\ $Z^{S}_{t}=(\pi_{t-(M-1)},  \,  \cdots \, , \, \pi_{t-3(M-1)})'$
\item $U_{t}=(\psi_{t}, \, \psi_{t-1}, \, \cdots, \, \psi_{t-2(M-1)})'$.
\end{enumerate}
\end{proposition}
\begin{proof}The proof has two principle columns. The first sifts through a complicated system of non-linear simultaneous equations to dredge up the required curve. The second analyzes the roots of a polynomial equation to obtain the sought after cocycle. 
\par Begin by breaking open expectations, as before, using the fact that for any $k \geq 1$, $\mathbb{E}_{t}\, g(X_{t+k})=g(X_{t+k}) + v_{t+k}(\psi_{t+k})$, where $\mathbb{E}_{t}\,v_{t+k}(\psi_{t+k})=0$ and obeys any restrictions on the range of $g$. Observe that 
(79) combined with (11), (27) and 
(80) proffers a relationship 
\begin{equation} \frac{p_{t}^{*}}{P_{t}}=S_{0}\bigg(\bigg\{ \frac{p_{T}^{*}}{P_{T}}\bigg\}^{t+(M-1)}_{T=t-(M-1)}, \,
\bigg\{ \pi_{T}\bigg\}^{t+ M-1}_{T=t-(M-1)}, \, Y_{t}, 
\, \bigg\{ \psi_{T}\bigg\}^{t+M-1}_{T=t}\bigg)\end{equation}
A sophisticated strategy to eliminate the relative prices now unfurls. Begin by lagging by $M$ periods.  
\begin{equation} \frac{p_{t-M}^{*}}{P_{t-M}}=S_{0}\bigg(\bigg\{ \frac{p_{T}^{*}}{P_{T}}\bigg\}^{t-1}_{T=t-(2M-1)}, \,
\bigg\{ \pi_{T}\bigg\}^{t-1}_{T=t-(2M-1)}, \, Y_{t}, 
\, \bigg\{ \psi_{T}\bigg\}^{t-1}_{T=t-M}\bigg)\end{equation}
First, extract the price dispersion contribution from the right-hand sides to form a function. It is easy to take out the two left hand-side terms because 
\begin{equation} (1+\pi_{t})^{\theta-1}-1=\frac{1}{M}\bigg( \frac{p^{*}_{t}}{P_{t}}\bigg)^{1-\theta}-\frac{1}{M}\bigg( \frac{p^{*}_{t-M}}{P_{t-M}}\bigg)^{1-\theta}\, \prod_{k=0}^{M-1}(1+\pi_{t-k})^{\theta-1}\end{equation}
appears out of 
(81) and its immediate lag. \par The common lags 
$\big\{p_{T}^{*}/P_{T}\big\}^{t-1}_{T=t-(M-1)}$
can be deciphered, by solving the Phillips curves, intermediate between 
(82) and 
(83). This works because of the special structure of the equations. Each variable $p_{T}^{*}/P_{T}$ is increasing in its own period and decreasing elsewhere, which ensures uniqueness. Each equation has an asymptote, which ensures existence. 
This leaves a surface 
\begin{multline} S_{1}\bigg(\bigg\{ \frac{p_{T}^{*}}{P_{T}}\bigg\}^{t-M}_{T=t-(2M-1)}, \, \bigg\{\frac{p_{T}^{*}}{P_{T}}\bigg\}^{t+M-1}_{T=t}, \\ \,
\bigg\{ \pi_{T}\bigg\}^{t +M-1}_{T=t-(2M-1)}, 
\, Y_{t}, 
\, \bigg\{ \psi_{T}\bigg\}^{t+M-1}_{T=t-M}\bigg)=0\end{multline}
With the middle sector of prices safely accounted for, 
repeated applications of 
(84)
vanquishes the unwanted relative prices. 
\par The final stumbling block is to prove that 
future inflation $(\pi_{t+1})$ is uniquely \\ determined (in expectation). Two coordinate transforms act as vehicles. Surface 
(82) can be rendered implicitly as 
\begin{equation} a = \frac{h+q\, x^{\theta-1-\eta}}{r +s\, x^{\theta-1}}\end{equation}
where $a=p^{*}_{t-(M-2)}/p_{t-(M-2)}$, $x=(1+\pi_{t+1})$ and the other coefficients, which are all strictly positive, can be easily deduced but are unimportant on their own. \par Algebraic 
dexterity implies the power equation 
\begin{equation} a\, s\, x^{\theta+ 1 +\eta} + (a\, r-h)x^{1+\eta}-q\, x^{\theta}=0\end{equation}
The first assignment is to prove that solutions are unique for $x>0$. The continuity of the power relation $z \rightarrow n^{z}$, for $n >0$, allows me to 
reconfigure the right-hand side so that 
\begin{equation} g_{1}(x')\equiv a\, s\, (x')^{\theta'\, + \, 1 +\, \eta'} + (a\, r-h)(x')^{1\, +\, \eta'}-q\, (x')^{\theta'}=0\end{equation}
with $\theta'$ and $\eta'$ positive integers, 
whilst retaining all the other restrictions on variables and parameters. This creates compatibility with the classical theories locating polynomial roots. There are three eventualities. 
\newline
\newline \underline{Case 1: $\theta' < 1+ \eta'$}
\newline \\ 
\newline This is the simplest occurrence. Notice that the highest order term is positive and the lowest is negative. Therefore, no matter what the sign 
of the middle term, there can only be one sign change and hence there must be exactly one positive root, according to Descartes' law of signs. The crucial point is to notice that 
\newline 
\newline 
\underline{Case 2: $\theta' = 1+ \eta'$}
\newline 
\\ The system simplifies to 
\begin{equation} 
g_{2}(x')\equiv as(x')^{\theta'} + ar-h-q=0\end{equation}
\newline 
There can only be one positive (real) solution. This must arise for the \\ optimization problem of the firm 
(78) to have a solution. I know this is true by inspecting the first-order condition 
(79), when $p^{*}_{t}/P_{t} \rightarrow 0$, 
the right-hand side dominates the left, when it tends to infinity the situation is reversed. Thus, the intermediate value theorem suffices. Note that we have gleaned that 
\begin{equation} ar-h-q<0\end{equation}
\newline  
\underline{Case 3: $\theta' > 1+ \eta'$}
\newline
\\
This is the trickiest case. The 
Phillips curve boils down to 
\begin{equation} 
g_{3}(x')\equiv a\, s\, (x')^{\theta'} -q\, (x')^{\theta'\, -1 \, -\eta}+ a\, r-h=0\end{equation}
If $ar-h \geq 0$ then previous arguments go through. Otherwise, an unwanted second solution arises. The exercise is to rule one out. 
\begin{equation}
ar-h=\mathbb{E}_{t}\sum_{T=t}^{t + M-1} Q_{t,\, T}\, \bigg(\frac{p_{t}^{*}}{P_{T}}\bigg)^{-\theta}Y_{T}\, \bigg[\frac{p_{t}^{*}}{P_{T}}-
C (y_{T})(i)
\bigg]
\end{equation}
Hence, $a\, r-h \geq 0$ if and only if $\pi \geq 0$, which translates to 
$x' \geq 1$. \\ Substituting into (91), using (90). I uncover that the middle coefficient dominates in magnitude, thus $g(1) <0$. Since $g(0)>0$, 
this locates the invalid solution. 
\par The previous three cases serve to prove that there exists a function $\pi_{t+1}=f(\Delta_{t}, \, \cdot)$. The mathematical demonstration moves to solving the appropriate system of simultaneous equations. In all three permutations (88), (89) and (91), $g_{i}(x')$ is increasing in $x'$ and decreasing in $q$. This means there is an upward sloping 
schedule between $\pi_{t+1}$ and $\Delta_{t+1}$. This points to a negative relationship between next period inflation $(\pi_{t+1})$ and its 
contemporaneous reset price $(p^{*}_{t+1}/P_{t+1})$, due to 
(81). On the other hand, the price level construction equation showcases a positively sloped curve connecting future prices and future inflation, to be exact 
\begin{equation} \frac{p^{*}_{t+1}}{P_{t+1}}=\bigg( \frac{M}{\tilde{P}}\bigg)^{1/{(\theta-1)}}\end{equation}
where
\begin{equation} \tilde{P}=\frac{\theta}{\theta-1}\Bigg(M-(1+\pi_{t+1})^{\theta-1}\, \bigg\{ \sum_{j=0}^{M-2}\bigg(\frac{p^{*}_{t-j}}{P_{t-j}}\bigg)^{\theta-1}\, \prod^{j}_{k=0}(1+\pi_{t-k})^{\theta-1}\bigg\}\Bigg)\end{equation}
This is adequate for uniqueness of the pair $(\pi_{t+1}, \, p^{*}_{t+1}/P_{t+1})$. \par To confirm that this represents a general equilibrium, a final petition to the intermediate value theorem is required. Send $p^{*}_{t+1}/P_{t+1} \rightarrow 0$ (ceterus parebus) in the Phillips curve, whichever of (88), (89) or (91) applies, it is necessary that $x'$ and so $\pi \rightarrow \infty$. 
By contrast, this limit does not exist in the pricing duo (93) and (94)
because $\pi \rightarrow -100\%$ is associated with a finite value of $p^{*}_{t+1}/P_{t+1}$. On the other hand, these curves have an asymptote at a finite $\pi$ (a function of the reset history). Therefore, the Phillips curve must intersect the price level graph from above. This ensures $\pi_{t+1}$ can be written as a function of present and past quantities. This extends to the standard recursive equilibrium presentation under a further change of variables.
\par The other sections of the recursive equilibrium can be built by inputting these solutions. The existence of expectations logic copies over from Proposition 2 for the shared macroeconomic variables. 
After recompiling expectations, it is clear that the dimensionality is correct. From the underlying optimization problem 
(78), it is easily ascertained that there are $M$ inflation jump variables and one for output.
\end{proof}
\begin{remark}It is noteworthy that, unlike with Calvo and menu costs, price dispersion is not a state variable with Taylor contracts. This is because it can be expressed as a function of a finite sequence of inflation history, in (78). This is a figment of how prices drop out of the model once a contract becomes renegotiable. \end{remark}
\begin{remark} A further distinction from the other models is that it is impossible to remove the effect of past reset prices, which influence the economy above and beyond past rates of inflation. However, this is 
exclusive to recursive equilibrium.\end{remark}
Section 6.2 and Supplementary Appendix B contain step by step derivations of Phillips curves, albeit in rarefied circumstances. 
\subsection{Steady-State Formulations} 
The business here begins by displaying the stochastic equilibrium conditions, with brief comment. It finishes by digging down to the small noise limit and analyzing the non-stochastic steady state. 
\subsubsection{Stochastic Equilibrium}
In harmony with other Keynesian contracting models, Taylor pricing is represented by a Phillips curve and a price distortion equation.
\par The tale flows from the optimal re-pricing decision 
(79). Unlike with the other pricing formats, (30) and (61), price dispersion evolves during the course of the contract, so it is best supplanted. For readability, enroll the functional form (9). Ergo, the Phillips curve is 
\begin{equation} \nabla = \frac{\theta}{\theta-1}\frac{\aleph} {\beth} \end{equation}
In all, there are three possible singularities, owing to the possibility for unit roots in the geometric progressions.
These are located at $\pi=0$, $\pi_{1}:\mathbb{E}(1+\pi)^{\theta}=1$ and $\pi_{2}: \beta \, \mathbb{E}\, (1+\pi)^{\theta - 1}=1$.
\footnote{It is possible to gauge the position of the singularities. $\pi_{1} <0$, by Jensen's inequality because it is the expectation of an strictly increasing strictly convex function. Meanwhile, $\pi_{2} >0$, so long as noise is sufficiently small. Therefore, it is most probable that $\pi_{1} < 0 < \pi_{2}$.} All singularities in the model are removable. The component inventory is below. 

\begin{equation}
\nabla=
 \begin{cases} 
 \nabla^{SS}
 &\text{if $\pi \neq 0$} \\ 
 1 &\text{if $\pi =0$}
 \end{cases}
\end{equation}
\begin{equation} \nabla^{SS}=\bigg(\frac{(1+\pi)^{M\, (\theta-1)}-1}{M\, \big\{(1+\pi)^{\theta-1}-1\big\}}\bigg)^{1/(\theta-1)}\end{equation}

\begin{equation}
\aleph=
 \begin{cases} 
 \aleph^{SS}
 &\text{if $\pi \neq 0$} \\ 
 \aleph^{SS}(0)&\text{if $\pi =0$}
 \end{cases}
\end{equation}
where 
\begin{multline} \aleph^{SS} =  M^{(1 + \eta-\theta)/(\theta-1)}\, 
\Bigg\{\bigg(\frac{(1+\pi)^{M\theta} -1}{(1+\pi)^{\theta} -1}\bigg)^{\eta} \bigg(\frac{(1+\pi)^{\theta-1} -1}{(1+\pi)^{M(\theta-1)} -1}\bigg)^{{\eta}\, \theta/(\theta-1)}\\ \times \psi \, Y^{1+\eta} + \mathbb{E}\sum_{k=1}^{M-1}\beta^{k}(1+\pi)^{k \theta}\Bigg\{\sum_{a=1}^{k}
(1+\pi)^{a \theta}\bigg(\frac{(1+\pi)^{M(\theta-1)}-1}{M\big\{(1+\pi)^{\theta-1}-1\big\}}\bigg)^{\theta/(\theta-1)} 
\\ + (M-k)\, \bigg(\frac{(1+\mathbb{E}\, \pi)^{M\theta} -1}{(1+\mathbb{E}\, \pi)^{\theta} -1}\bigg)\bigg(\frac{(1+\mathbb{E}\, \pi)^{\theta-1} -1}{(1+\mathbb{E}\, \pi)^{M(\theta-1)} -1}\bigg)^{\theta/(\theta-1)}\Bigg\}^{\eta}\psi \, Y^{1+\eta}\end{multline}
\begin{multline} \aleph^{SS}(0) =  M^{(1 + \eta-\theta)/(\theta-1)}\, 
\Bigg\{\psi \, Y^{1+\eta} + 
\mathbb{E}\sum_{k=1}^{M-1}\beta^{k}\, (1+\pi)^{k \theta}\, \Bigg\{\sum_{a=1}^{k}
(1+\pi)^{a \theta} \times \\ \bigg(\frac{(1+\pi)^{M(\theta-1)}-1}{M, \big\{(1+\pi)^{\theta-1}-1\big\}}\bigg)^{\theta/(\theta-1)} 
+ (M-k)\Bigg\}^{\eta} \, \psi \, Y^{1+\eta}\end{multline}

\begin{equation}
\beth=
 \begin{cases} 
 \beth^{SS}
 &\text{if $\pi \neq 0$} \\ 
 \psi + (M-1)\, \mathbb{E}\, \psi \, (1+\pi)^{\theta-1}&\text{if $\pi =\pi_{2}$}
 \end{cases}
\end{equation}
\begin{equation} \beth^{SS}=\psi + \beta \, \mathbb{E}\, \psi \, (1+\pi)^{\theta-1}\, \bigg(\frac{\beta^{M-1}\,\mathbb{E}\, (1+\pi)^{(M-1)(\theta-1)}-1}{\beta \, \mathbb{E}\, (1+\pi)^{\theta-1}-1}\bigg)\end{equation}
which wraps up the Phillips curve. Proceedings culminate with price dispersion. Factorizing terms at common times unearths
\begin{equation}
\Delta=
 \begin{cases} 
 \Delta^{SS}
 &\text{if $\pi \neq 0$} \\ 
  \Delta^{SS}(1) &\text{if $\pi =\pi_{1}$}
 \end{cases}
\end{equation}
\begin{equation} \Delta^{SS}=M^{1/(\theta-1)}\bigg(\frac{\mathbb{E}\, (1+\pi)^{(M-1)\theta} -1}{\mathbb{E}\, (1+\pi)^{\theta} -1}\bigg)\mathbb{E}\, (1+\pi)^{\theta}\, \bigg(\frac{(1+\pi)^{\theta-1} -1}{(1+\pi)^{M(\theta-1)} -1}\bigg)^{\theta/(\theta-1)}\end{equation}
\begin{equation} \Delta^{SS}(1)=M^{1/(\theta-1)}(M-1) \, \mathbb{E}\, (1+\pi)^{\theta}\, \bigg(\frac{(1+\pi)^{\theta-1} -1}{(1+\pi)^{M(\theta-1)} -1}\bigg)^{\theta/(\theta-1)}\end{equation}
These equilibrium descriptions can be compared with (30) and (31) for Calvo pricing, as well as (70) and (71) for menu costs.
\footnote{For equations (96), (98), (101) and (103), there is no need to specify the value of the first function in the integral at zero because it is measure zero and therefore does not affect the expectation.}
\begin{remark} I decline to write out the stochastic equilibrium profit function here. It is established for the Calvo model that the value of any firm is positive in stochastic equilibrium and in any recursive equilibrium where households are sufficiently patient. These arguments go through for the finite horizons of a Taylor contract. Re-optimization possibilities under menu costs would increase the value of the firm, strictly because the model lives on a manifold.\end{remark} 
Readers may wish to refer back to SELCKE Section 4.4 for watertight \\ justification of claims in the previous remark.
\subsubsection{Non-Stochastic Equilibrium and Analysis}
The purpose at hand is two fold; first to create expressions for equilibrium quantities at any non-stochastic steady and then to discover the qualitative behavior of price dispersion. The solution process is simple save the ubiquitous singularity at ZINSS.
Formally, for $X=\{MC, \, \Delta, \, Y, \, L, \, \Pi \}$
\begin{equation}
X=
 \begin{cases} 
 X^{NSS}  &\text{if $\pi \neq 0$} \\ 
  X^{ZINSS} &\text{if $\pi =0$}
 \end{cases}
\end{equation}
To simplify the optimization relationships, I will eliminate time preference by sending $\beta \rightarrow 1$. The two pivotal quantities are 
\begin{equation} MC^{NSS}= \frac{\theta-1}{\theta M^{1/(\theta-1)}}\bigg(\frac{(1+\pi)^{\theta} -1}{(1+\pi)^{M\theta} -1}\bigg)\bigg(\frac{(1+\pi)^{M(\theta-1)} -1}{(1+\pi)^{\theta-1} -1}\bigg)^{\theta/(\theta-1)}\end{equation}
\begin{equation} \Delta^{NSS}= 
M^{1/(\theta-1)}\bigg(\frac{(1+\pi)^{M\theta} -1}{(1+\pi)^{\theta} -1}\bigg)\bigg(\frac{(1+\pi)^{\theta-1} -1}{(1+\pi)^{M(\theta-1)} -1}\bigg)^{\theta/(\theta-1)}\end{equation}
Using the labor supply (12), profit function from 
(78) and market equilibrium condition (26), I further uncover that 
\begin{multline}Y^{NSS}= \frac{A}{M^{1/(\theta-1)}}\bigg(\frac{\theta-1}{\theta}\bigg)^{1/(\eta+1)}\bigg( \frac{(1+\pi)^{\theta} -1}{(1+\pi)^{M\theta} -1}\bigg)\times \\ \bigg(\frac{(1+\pi)^{M(\theta-1)} -1}{(1+\pi)^{\theta-1} -1}\bigg)^{\theta/(\theta-1)}\end{multline}
\begin{equation} L^{NSS}=\bigg(\frac{\theta}{\theta-1}\bigg)^{1/(\eta+1)}\end{equation}
\begin{multline} \Pi^{NSS}= \Bigg\{ 1 - \frac{\theta-1}{\theta M^{1/(\theta-1)}}\bigg(\frac{(1+\pi)^{\theta} -1}{(1+\pi)^{M\theta} -1}\bigg)\bigg(\frac{(1+\pi)^{M(\theta-1)} -1}{(1+\pi)^{\theta-1} -1}\bigg)^{\theta/(\theta-1)}\Bigg\} \times \\ \frac{A}{M^{1/(\theta-1)}}\bigg(\frac{\theta-1}{\theta}\bigg)^{1/(\eta+1)}\bigg( \frac{(1+\pi)^{\theta} -1}{(1+\pi)^{M\theta} -1}\bigg)
\bigg(\frac{(1+\pi)^{M(\theta-1)} -1}{(1+\pi)^{\theta-1} -1}\bigg)^{\theta/(\theta-1)}\end{multline} 
\begin{remark} Section 4.2 in SELCKE discusses how to make the exact form of the labor-leisure preferences (9) consistent with the boundary condition (7).\end{remark}
The values at ZINSS are well-known and match with the neoclassical benchmark model. It is possible to impose a corrective subsidy $\tau = 1/\theta$ funded by (an unrealistic) lump sum tax to bring about the first best associated with $MC^{*}=1$. Appendix D.2.1 contains general cases of interest. The invariance of aggregate hours to trend inflation distortions is a well-understood consequence of optimal labor supply, with log-utility absent time discounting. 
\begin{proposition}
Price dispersion $(\Delta)$ is second-order around ZINSS.
\end{proposition}
I prove the simplest two period case here in the text, leaving the general $M$ period environment to Supplementary Appendix Section A.2.2.
\begin{proof}\textbf{(Two Period Model)} \\ 
By direct computation, non-stochastic steady state price dispersion takes the simple form 
\begin{equation} \Delta=\frac{2^{1/(\theta-1)}\big\{1 + (1+\pi)^{\theta} \big\}}{(1 + (1+\pi)^{\theta-1} )^{\theta/ (\theta -1)}}\end{equation}
Differentiating gives 
\begin{equation} \frac{\mathrm{d}\Delta}{\mathrm{d}\pi}= \bigg[\frac{2^{1/(\theta-1)}\, \theta \, (1+\pi)^{\theta-2}}{(1 + (1+\pi)^{\theta-1} )^{(2\theta-1)/ (\theta -1)}}\bigg]\pi\end{equation}
When it comes to applying the product rule, it is clear that, at $\pi=0$ the impact of the derivative of the bracket is nullified. Hence, 
\begin{equation} \frac{\mathrm{d}^{2}\Delta}{\mathrm{d}^{2}\pi}\bigg\rvert_{\pi=0}=
\frac{2^{1/(\theta-1)}\, \theta(1+\pi)^{\theta-2}}{(1 + (1+\pi)^{\theta-1} )^{(2\theta-1)/ (\theta -1)}}
\bigg\rvert_{\pi=0}=2^{1/(\theta-1)} \, \theta>0\end{equation}
is sufficient.\footnote{Naturally, we knew already that it had to be non-negative from the basic properties of $\Delta$. This stylized exposition may present pedagogic advantages.} 
\end{proof}
\section{Taylor Contracts: Towards Application}
This practically oriented section is split in two. The 
first piece considers business cycle 
characteristics of the Taylor Phillips curve with a view towards practical work with data. The second is taken up with an intuition building solution of the most parsimonious form. The underlying calculations for general repricing frequency are contained in 
Supplementary Appendix B. Additional facets appear in Supplementary Appendix C.
The requisite existence results are proven in Appendix D. There are theoretical and empirical challenges ahead.
\subsection{Discussion and Empirics}
\par At the outset, it is essential to demonstrate any proposed solution is
correct. I hone in on $(a_{\pi}, \, a_{y})'=(0.5, \, 0.5)'$, for comparability with Calvo. The proof is relatively straightforward for the two period toy model, coming up in the next subsection. Things appear trickier with longer contracts. I do not attempt a general rule, even around a particular parametization. Instead, I stick to a single contract length of four periods. \par The original motivation for this choice comes from \cite{alvarez2006sticky}, who show that this was close to the average frequency of adjustment in the Eurozone.
It is close to the upper end of the rigidity range considered plausible in SELCKE Supplementary Materials Section G.1.3, corresponding to $\alpha=0.75$ there. Moreover, \cite{dixon2006compare}, \cite{dixon2011contract}and \cite{dixon2012generalised} suggest that long contracts dominate short contracts, in the sense that, with realistic heterogeniety, the data is better fit by using a single Taylor contract length above the true average duration. 
\par The specific difficulty arises that there is a unit root at ZINSS, regardless of the policy stance. I am able to circumvent this problem by adding a small fringe of flexible price firms. This alteration has empirical pedigree, many prices linked to global markets like oil or agricultural goods adjust rapidly.\footnote{The empirical consensus is around $20-30\%$ of firms adjust their prices more than weighted by their share in the consumer price indices, see \cite{coenen2007identifying} in addition to previous references and https://huwdixon.org/GTE.html. } In fact, this share loads directly on to the output gap term in the Philips curve. It would be interesting to see whether these predictions are true in general, for example for any contract length $M \geq 2$. 
\par The precise mathematical statements are as follows
\begin{proposition} Consider the Taylor contracting framework laid out in Section 6.1 with contract length $M-2$, there exists a recursive equilibrium at the standard parameter setting $(a_{\pi}, \, a_{y})'=(0.5, \, 0.5)'$.\end{proposition}
\begin{construction}
Suppose a perturbation to the economy, where a fraction $\omega>0$ is introduced that re-optimize every period, so the remainder $1-\omega$ continue to reoptimize every $M$ period. Consider the limit such that \, $\omega \rightarrow 0$ whilst $1-\beta  << \vert \varepsilon \vert << \omega$.
\end{construction}
\begin{proposition}
Under Construction 1, there exists a 
recursive equilibrium 
when $M=4$ and $(a_{\pi}, \, a_{y})'=(0.5, \, 0.5)'$.
\end{proposition}
Each proposition fills its own subsection of the Supplementary Appendix. The idea behind Construction 1 is to reduce the eigenvalue polynomial to its non-stochastic counterpart and then add a small share of firms who reset their prices every period to divert the eigenvalues away from the unit circle and into the appropriate combination. However, the limit is degenerate and therefore not conducive to welfare analysis. The heightened possibility for unit roots may cause computational challenges for the standard Taylor pricing model, although the absence of polydromy works counter to this.
\par Finally, it is profitable to consider the new Phillips curves as numerical artifacts. 
\begin{equation} 
\pi_{t}=0.270\, \pi_{t-2} + 0.072\, \pi_{t-1} + 1.126\, \hat{y}^{e}_{t} - 0.649\, \mathbb{E}_{t}\, \pi_{t+1} + \hat{u}^{2}_{t}(\beta \rightarrow 1)\end{equation}
\begin{multline} \pi_{t}=0.205\, \pi_{t-4} + 0.055\, \pi_{t-3} - 0.009\, \pi_{t-2} -0.006\, \pi_{t-1} + 4.000\, \hat{y}^{e}_{t} - \\ 1.779\, \mathbb{E}_{t}\, \pi_{t+1}-0.982\, \mathbb{E}_{t}\, \pi_{t+2} -0.492\, \mathbb{E}_{t}\, \pi_{t+3} + \hat{u}^{4}_{t}(\beta \rightarrow 1)\end{multline}
To help build intuition, note that for the Taylor two period model, the coefficient expressions can be obtained from 
(152)-(156), written out in full in the next subsection.

\begin{equation} b^{2}\equiv 
2\, (2+\eta) -(1+\eta)\, a_{\pi}- \frac{2\, (1+\eta)\, a_{y}}{1 + a_{y}}=\frac{37}{6}
\end{equation}
\begin{equation} \tilde{b}^{\pi, \, 2}_{2}\equiv \frac{a_{\pi}}{(1 +a_{y})}(1+\eta)=\frac{5}{3}\end{equation} 
\begin{equation} \tilde{b}^{\pi, \, 2}_{1} \equiv 2\, (1+\eta) \, a_{\pi}-\eta - \frac{(1+\eta)\, a_{y}\, a_{\pi}}{(1+ a_{y})^{2}} - \frac{(1+\eta)\, a_{y}}{(1+  a_{y})}\bigg[ 2\,a_{\pi}-1\bigg]= \frac{4}{9}
\end{equation}
\begin{equation} \tilde{b}^{\pi, \, 2}_{-1}=\prescript{\circ}{}{b}^{\pi, \, 2}_{-1}\equiv -\eta=-4\end{equation}
\begin{equation} \tilde{b}^{y, \, 2} \equiv \frac{(1+\eta)\, a_{y}}{(1+  a_{y})^{2}}\bigg(1+2\, (1+a_{y}) + (1+a_{y})^{2} \bigg) = \frac{125}{18}\end{equation} 
Any analysis is preliminary without having solved out for the error coefficients and rough without hypothesis testing protocol. Nevertheless, two features stand out. The first is that inflation dynamics seem to be output rather than its own internal dynamics. This contrasts with Calvo where the slope of the Phillips curve at standard parameters is zero. Secondly, inflation dynamics might be counterintuitive because coefficients seem to have different sign patterns to Calvo. For example, the longest lead of inflation is always negative and in both cases negative coefficients dominate their positively signed counterparts. Moreover, there are difficulties generalizing results related to the fact that there seems to be no guarantee that the  numerator is positive, unlike with Calvo contracts (40). 
The cost channel is the culprit. When it disappears with the inactive policy $(a_{\pi}, \, a_{y})'=(0, \, 0)'$, we are back in the setting of SELCKE Supplementary Materials Section G.1 and equations (340) and (343), where the inflation coefficients are all positive and in fact sum to unity.\footnote{The common practice hitherto has been to close the model with a money demand function $\hat{y}^{e}_{t}=\hat{m}_{t}-\hat{p}_{t}$ and a shock process for the money supply, which would equally circumvent changing interest burdens.} In particular, the annual contract model looks like this 
\begin{equation} \pi_{t}=\frac{1}{12}\, \pi_{t-3} + \frac{1}{6}\, \pi_{t-2} + 
\frac{1}{4}\, \pi_{t-1} + \frac{1}{4}\, \mathbb{E}_{t}\, \pi_{t+1} + \frac{1}{6}\, \mathbb{E}_{t}\,\pi_{t+2} + \frac{1}{12}\, \mathbb{E}_{t}\, \pi_{t+3} + \hat{u}^{4}_{t}(0, \, 0)\end{equation}
The fringe of flexible prices is likely to turn out a more powerful determinant of the slope of the Phillips curve. Nevertheless, it is surely the case that the cost channel is intrinsic to the trade-offs of monetary stabilization. 
Measures may need to be taken to reduce the force of this mechanism relative to others. To this end, it would be natural to consider positive trend inflation and slow adjustment of interest rates first, although in time more sophisticated financial and behavioral facets may be worth exploring. 
\par Finally, it is easy to overlook advantages stemming from the absence of a maximum rate of inflation. This contrasts with Calvo and menu costs.\footnote{In fact with Calvo, there is an overall maximum bound, a tighter bound in non-stochastic steady state, which is even stricter with stochasticity.} 
Indeed, owing to Jensen's inequality the latter actually worsen as inflation volatility increases. With state dependent pricing, there is no upper bound but this arises at the expense of prices becoming unrealistically flexible. Thus, it is probably be desirable to consider both Taylor and menu cost models when looking at high or volatile inflation regimes. In general, moreover, it might prove propitious to combine aspects of state and time dependence, as well as reconsidering alternative real or nominal frictions like positive trend inflation.
\subsection{Basic Two Period Phillips Curve}
The step-by-step solution with two period contracts is summarized 
as follows. The mathematical purpose is to determine that all the forces are first-order. A significant observation concerning price dispersion is mentioned at the end. 
\begin{proposition}
Suppose Taylor contracts last for 
$M=2$ periods, then around ZINSS recursive equilibrium takes the form $\mathbb{E}_{t}\mathbf{\hat{Z}}_{t+1}=\mathbf{A}\mathbf{\hat{Z}}_{t} + \mathbf{\Phi}\mathbf{\hat{U}}_{t}$ 
where\\ $\mathbf{Z}_{t}=(\pi_{t}, \, \pi_{t-1}, \, , \pi_{t-2}, \,   Y^{e}_{t})'$, $\mathbf{U}_{t}=(\psi_{t}, \, \psi_{t-1}, \, \psi_{t-2})'$ 
and $\mathbf{\Phi}=\mathbf{\Phi}(\boldsymbol{\gamma})$. 
\end{proposition} 
\begin{proof}
Rather than (78) the first-order condition comes about from 
\begin{equation}
\max _{p_{t}^{*}}\mathbb{E}_{t}\sum_{T=t}^{t + 2} Q_{t,\, T}\bigg(\frac{p_{t}^{*}}{P_{T}}\bigg)^{-\theta}Y_{T}\bigg[\frac{p_{t}^{*}}{P_{T}}-\frac{\theta}{\theta-1}MC_{T}(y_{T}(i))\bigg]=0
\end{equation}
It is clear that the price change will be the average of the expected changes in marginal cost and aggregate price level, throughout the contract, suitably weighted to reflect time preference. 
For simplicity, I focus on ZINSS 
\begin{equation} \hat{p}^{*}_{t}= \frac{1}{1+\beta}\, \hat{P}_{t} + \frac{\beta}{1+\beta}\, \mathbb{E}_{t}\, \hat{P}_{t+1} + \frac{1}{1+\beta}\, \hat{mc}_{t} + \frac{\beta}{1+\beta}\, \mathbb{E}_{t}\, \hat{mc}_{t+1} \end{equation}
Lagging the relationship I obtain 
\begin{equation} \hat{p}^{*}_{t-1}= \frac{1}{1+\beta}\, \hat{P}_{t-1} + \frac{\beta}{1+\beta}\, \mathbb{E}_{t-1}\,\hat{P}_{t} +  \frac{1}{1+\beta}\, \hat{mc}_{t-1} +  \frac{\beta}{1+\beta}\, \mathbb{E}_{t-1}\, \hat{mc}_{t} \end{equation} 
With Taylor pricing, it is important to keep track of when expectations are computed. As expected,

\begin{equation} \hat{P}_{t}=\frac{1}{2}\, \hat{p}^{*}_{t-1} + \frac{1}{2}\, \hat{p}^{*}_{t}\end{equation}
Lagging this relationship leads to 
\begin{equation} \hat{P}_{t-1}=\frac{1}{2}\hat{p}^{*}_{t-2} + \frac{1}{2}\hat{p}^{*}_{t-1}\end{equation}
Thus, (126) and (127) imply 
\begin{equation} \pi_{t}= \frac{1}{2}\, (\hat{p}^{*}_{t}-\hat{p}^{*}_{t-1}) + \frac{1}{2}\, (\hat{p}^{*}_{t-1}-\hat{p}^{*}_{t-2})\end{equation}
Subtracting (124) and (125) yields 
\begin{multline} \hat{p}^{*}_{t}-\hat{p}^{*}_{t-1} = \frac{1}{1 +\beta} \, \pi_{t} + \frac{\beta}{1 + \beta} \, \mathbb{E}_{t}\, \pi_{t+1} + \frac{1}{1 +\beta} \, (\hat{mc}_{t} -\hat{mc}_{t-1}) + \\ \frac{\beta}{1 + \beta}\, (\mathbb{E}_{t}\, \hat{mc}_{t+1} - \hat{mc}_{t})  + 
\hat{v}^{0, \, 2}_{t}\end{multline}
where \begin{equation} 
\hat{v}^{0, \, 2}_{t}= \frac{\beta}{1 + \beta}\, (\pi_{t}-\mathbb{E}_{t-1}\, \pi_{t}) + \frac{\beta}{1 + \beta}\, (\hat{mc}_{t}-\mathbb{E}_{t-1}\, \hat{mc}_{t})\end{equation}
reflects differences between expected and realized outcomes. Carrying out the same steps for (125) and its lag then substituting into 
(128) reveals 
\begin{multline} \pi_{t}= \frac{1}{1 + \beta}\, \pi_{t-1} + \frac{\beta}{1 + \beta}\, 
\mathbb{E}_{t}\, \pi_{t+1} + \frac{1}{1 + \beta} \, (\hat{mc}_{t-1} -\hat{mc}_{t-2}) + \\ \frac{\beta}{1 + \beta} \, (
\mathbb{E}_{t}\, \hat{mc}_{t+1} -\hat{mc}_{t}) + 
(\hat{mc}_{t} -\hat{mc}_{t-1}) + \hat{v}^{1, \, 2}_{t}
\end{multline}
with \begin{equation} \hat{v}^{1, \, 2}_{t}=\frac{\beta}{1 + \beta}\, (\pi_{t-1}-\mathbb{E}_{t-2}\, \pi_{t-1})+ \frac{\beta}{1 + \beta}\, (\hat{mc}_{t-1}-\mathbb{E}_{t-2}\, \hat{mc}_{t-1}) + \hat{v}^{0, \, 2}_{t}\end{equation}
Deploying aggregate supply relation (33) implies that 
\begin{multline} \pi_{t}= \frac{1}{1 + \beta}\, \pi_{t-1} + \frac{\beta}{1 + \beta}\, \mathbb{E}_{t}\, \pi_{t+1} + \frac{(1+\eta)}{1 + \beta} \, (\hat{y}
_{t-1} -\hat{y}_{t-2}) + \\ \frac{\beta\, (1+\eta)}{1 + \beta}\,  (
\mathbb{E}_{t}\, \hat{y}_{t+1} -\hat{y}_{t})   + 
(1+\eta)\, (\hat{y}_{t} -\hat{y}_{t-1}) + \frac{\eta}{1 + \beta} \, (\hat{\Delta}_{t-1} -\hat{\Delta}_{t-2}) +  \\ \frac{\beta \, \eta}{1 + \beta}\,  (\mathbb{E}_{t}\, \hat{\Delta}_{t+1} -\hat{\Delta}_{t}) +  \eta\, (\hat{\Delta}_{t} -\hat{\Delta}_{t-1}) +\hat{v}^{1, \, 2}_{t}\end{multline}
\par Turning to price dispersion defined back at (24), it is evident that 
\begin{equation} \hat{\Delta}_{t}= \frac{\theta}{2}\, \bigg(\hat{p}^{*}_{t}-\hat{P}_{t}\bigg) +  \frac{\theta}{2}\, \bigg(\hat{p}^{*}_{t-1}-\hat{P}_{t}\bigg)\end{equation}
In the previous period, 
\begin{equation} \hat{\Delta}_{t-1}= \frac{\theta}{2}\bigg(\hat{p}^{*}_{t-1}-\hat{P}_{t-1}\bigg) + \frac{\theta}{2}\bigg(\hat{p}^{*}_{t-2}-\hat{P}_{t-1}\bigg)\end{equation}
Subtracting (135) from (134) with help from (126), I uncover that 
\begin{equation} \hat{\Delta}_{t}=\hat{\Delta}_{t-1}\end{equation}
It is self-evident that this unit root will destroy all the price dispersion terms in the Phillips curve (115). 
This expedites a re-designation of the demand pressure measure. \par Repeatedly plugging in the aggregate demand curve (32), down the small noise limit, 
\begin{equation} \pi_{t}=
\prescript{\circ}{}{b}^{\pi, \, 2}_{2}\, \pi_{t-2} + \prescript{\circ}{}b^{\pi, \, 2}_{1}\, \pi_{t-1} + \prescript{\circ}{}\, b^{\pi, \, 2}_{-1}\, \mathbb{E}_{t}\, \pi_{t+1} + \prescript{\circ}{}b^{y, \, 2}_{2}
\, \hat{y}^{e}_{t-2}+ \prescript{\circ}{}b^{y, \, 2}_{1}\, \hat{y}^{e}_{t-1} + \prescript{\circ}{}b^{y, \, 2}_{0}\, \hat{y}^{e}_{t} + 
\hat{v}^{2, \, 2}_{t}
\end{equation} 
where
\begin{equation} \prescript{\circ}{}b^{2}= (1 + \beta)\, (2+\eta) 
-\beta^{2}\, (1+\eta)\, a_{\pi}\end{equation}
and $\prescript{\circ}{}{b}^{j, \, 2}_{i}=\prescript{\circ}{}{\tilde{b}}^{j, \, 2}_{i}/\prescript{\circ}{}{b}^{2}$ for $j=\{ \pi, \, y, \, \psi\}$ and $i=\{1, \, 2\}$
are constructed from  
\begin{equation} \prescript{\circ}{}{\tilde{b}}^{\pi, \, 2}_{2}
=\beta\, (1+\eta) \, a_{\pi}\end{equation}
\begin{equation} \prescript{\circ}{}{\tilde{b}}^{\pi, \, 2}_{1}=\beta\, (1+\beta)\, (1+\eta) \, a_{\pi}-\eta\end{equation}
\begin{equation} \prescript{\circ}{}{\tilde{b}}^{\pi, \, 2}_{-1}=-\beta \, \eta \end{equation}
\begin{equation}\prescript{\circ}{}{\tilde{b}}^{y, \, 2}_{2}=\beta \,  (1+\eta)\, a_{y} \end{equation}
\begin{equation} \prescript{\circ}{}{\tilde{b}}^{y, \, 2}_{1}
=\beta\, (1+\beta) \, (1+\eta)\, a_{y}\end{equation}
\begin{equation} \prescript{\circ}{}{\tilde{b}}^{y, \, 2}_{0}
=\beta^{2}\, (1+\eta)\, a_{y}\end{equation}
bearing in mind the properties of the errors embodied in (28)
\begin{multline} 
\hat{v}^{2, \, 2}_{t}= \prescript{\circ}{}b^{\psi, \, 2}_{2}\, \hat{\psi}_{t-2} + \prescript{\circ}{}b^{\psi, \, 2}_{1}\, \hat{\psi}_{t-1} + \prescript{\circ}{}b^{\psi, \, 2}_{0}\, \hat{\psi}_{t}
+ \prescript{\circ}{}b^{x, \, 2}_{1}\, (\hat{y}^{e}_{t-1}-\mathbb{E}_{t-2}\, \hat{y}^{e}_{t-1}) +  \\ \prescript{\circ}{}b^{x, \, 2}_{0}\, (\hat{y}^{e}_{t}-\mathbb{E}_{t-1}\, \hat{y}^{e}_{t}) + 
\prescript{\circ}{}b^{x, \, 2}_{1}\, (\pi_{t-1}-\mathbb{E}_{t-2}\, \pi_{t-1}) + \prescript{\circ}{}b^{x, \, 2}_{0}\, (\pi_{t}-\mathbb{E}_{t-1}\, \pi_{t})  +  \frac{(1+\beta)}{\prescript{\circ}{}b^{2}}\, \hat{v}^{1, \, 2}_{t}
\end{multline}
\begin{equation} 
\prescript{\circ}{}{\tilde{b}}^{\psi, \, 2}_{2}=-(1-\rho)\, (1+\eta)\end{equation}
\begin{equation} \prescript{\circ}{}{\tilde{b}}^{\psi, \, 2}_{1}=-(1+\beta)\, (1-\rho)\, (1+\eta)\end{equation}
\begin{equation} 
\prescript{\circ}{}{\tilde{b}}^{\psi, \, 2}_{0}=-\beta\, (1-\rho)\, (1+\eta)\end{equation}
\begin{equation} \prescript{\circ}{}{\tilde{b}}^{x, \, 2}_{1}= 1+\eta \end{equation}
\begin{equation}  \prescript{\circ}{}{\tilde{b}}^{x, \, 2}_{0}=(1+\beta)\, (1+\eta)\end{equation}
More use of the aggregate demand curve gives the desired compact form. 
\begin{equation} \pi_{t}= b_{2}^{\pi, \, 2}\, \pi_{t-2} + b_{1}^{\pi, \,2}\, \pi_{t-1} +   b^{y, \, 2}\, \hat{y}^{e}_{t} +
 b^{\pi, \, 2}_{-1}\, \mathbb{E}_{t}\, \pi_{t+1}+ \hat{u}^{2}_{t}\end{equation} 
similar to previous notation 
$b^{j}_{i}=\tilde{b}^{j}_{i}/b^{2}$. Here 
\begin{equation} b^{2}= \prescript{\circ}{}{b}^{2}-\frac{\prescript{\circ}{}{\tilde{b}}^{y, \, 2}_{1}}{1 + \beta a_{y}}=
(1 + \beta)\, (2+\eta) 
-\beta^{2}\, (1+\eta)\, a_{\pi}- \beta\, (1+\beta) \, \frac{(1+\eta)\, a_{y}}{1 + \beta \, a_{y}}\end{equation}
whilst
\begin{equation} \tilde{b}^{\pi, \, 2}_{2}=\prescript{\circ}{}{b}^{\pi, \, 2}_{2}-\frac{\beta \, a_{\pi}}{(1 +\beta \, a_{y})}\prescript{\circ}{}{b}^{y, \, 2}_{2}=\frac{\beta \,  a_{\pi}}{(1 +\beta \, a_{y})}(1+\eta)\end{equation} 
\begin{multline} \tilde{b}^{\pi, \, 2}_{1}=\prescript{\circ}{}{b}^{\pi, \, 2}_{1} -\bigg( \frac{\prescript{\circ}{}{b}^{y, \, 2}_{2}}{(1+ \beta \, a_{y})^{2}} + \frac{\prescript{\circ}{}{b}^{y, \, 2}_{1}}{(1+ \beta \, a_{y})}\bigg) \beta \, a_{\pi} + \frac{\prescript{\circ}{}{b}^{y, \, 2}_{2}}{(1+ \beta \, a_{y})}= \\ \beta\, (1+\beta)\, (1+\eta) \, a_{\pi}-\eta - \frac{\beta^{2} \, (1+\eta)\, a_{y}\, a_{\pi}}{(1+ \beta \, a_{y})^{2}} 
- \frac{\beta \, (1+\eta)\, a_{y}}{(1+ \beta \, a_{y})}\bigg[ \beta\, (1+\beta)\, a_{\pi}-1\bigg]\end{multline}
\begin{equation} \tilde{b}^{\pi, \, 2}_{-1}=\prescript{\circ}{}{b}^{\pi, \, 2}_{-1}=-\beta \, \eta\end{equation}
\begin{multline} b^{y, \, 2} = \frac{\prescript{\circ}{}{b}^{y, \, 2}_{2}}{(1+ \beta \, a_{y})^{2}} + \frac{\prescript{\circ}{}{b}^{y, \, 2}_{1}}{(1+ \beta \, a_{y})} + \prescript{\circ}{}{b}^{y, \, 2}_{0}= \\ \frac{\beta \,  (1+\eta)\, a_{y}}{(1+ \beta \, a_{y})^{2}}\bigg(1+(1+\beta)\, (1+\beta \, a_{y}) + \beta\, (1+\beta \, a_{y})^{2} \bigg)\end{multline} 
Finally, the error adjustment is 
\begin{multline} 
\hat{u}^{2}_{t}=\prescript{\circ}{\circ}b^{\psi, \, 2}_{2}\, \hat{\psi}_{t-2} + \prescript{\circ}{\circ}b^{\psi, \, 2}_{1}\, \hat{\psi}_{t-1} + 
\prescript{\circ}{\circ}{b}^{x, \, 2}_{1}\, (\pi_{t-1}-\mathbb{E}_{t-2}\, \pi_{t-1}) + \prescript{\circ}{\circ}{b}^{x, \, 2}_{1}\, (\hat{y}^{e}_{t-1}-\mathbb{E}_{t-2}\, \hat{y}^{e}_{t-1}) 
\\+ \prescript{\circ}{\circ}{b}^{x, \, 2}_{0}\, (\hat{\pi}_{t}-\mathbb{E}_{t-1} \,\hat{\pi}_{t}) + \prescript{\circ}{\circ}{b}^{x, \, 2}_{0}\, (\hat{y}^{e}_{t}-\mathbb{E}_{t-1} \,\hat{y}^{e}_{t}) + \frac{\prescript{\circ}{}{b}^{2}}{b^{2}}\, \hat{v}^{2, \, 2}_{t}\end{multline}
\begin{equation} \prescript{\circ}{\circ}{\tilde{b}}^{\psi, \, 2}_{2}= \beta\, (1-\rho)\, \frac{(1+\eta)\, a_{y}}{1 +\beta \, a_{y}} \end{equation}
\begin{multline} \prescript{\circ}{\circ}{\tilde{b}}^{\psi, \, 2}_{1}= (1-\rho)\,\bigg(\frac{\prescript{\circ}{}{\tilde{b}}^{y, \, 2}_{2}}{(1+\beta \, a_{y})^{2}} + \frac{\prescript{\circ}{}{\tilde{b}}^{y, \, 2}_{1}}{(1+\beta \, a_{y})}\bigg)= \\ (1-\rho)\, \frac{\beta \, (1+\eta)\, a_{y}}{(1+\beta \, a_{y})^{2}} \bigg( 1 + (1+\beta)\, (1+\beta \, a_{y})\, \bigg)\end{multline}
\begin{equation} \prescript{\circ}{\circ}{\tilde{b}}^{x, \, 2}_{1}=-\frac{\prescript{\circ}{}{\tilde{b}}^{y, \, 2}_{2}}{1+\beta \, a_{y}}=-\frac{\beta\,  (1+\eta)\, a_{y}}{1+\beta \, a_{y}} \end{equation} 
\begin{equation} \prescript{\circ}{\circ}{\tilde{b}}^{x, \, 2}_{0}=-\bigg(\frac{\prescript{\circ}{}{\tilde{b}}^{y, \, 2}_{2}}{(1+\beta \, a_{y})^{2}} + \frac{\prescript{\circ}{}{\tilde{b}}^{y, \, 2}_{1}}{(1+\beta \, a_{y})}\bigg)=-\frac{\beta \, (1+\eta)\, a_{y}}{(1+\beta \, a_{y})^{2}} \bigg( 1 + (1+\beta)\,(1+\beta \, a_{y})\bigg)\end{equation}
\par It is not necessary for the proof to evaluate the lagged expectation \\ component of the error terms. By way of justification, it is clear, they will not contribute additional lags. The alternative possibility is that the errors would cancel out. This cannot happen because it would induce a contemporaneous relationship between the error and the endogenous variables, which would create a contradiction with (153) and (154). 
These jointly imply $(\tilde{b}^{\pi, \, 2}_{1}, \, \tilde{b}^{\pi, \, 2}_{2})\geq \bf{0}$ but not simultaneously, since $\tilde{b}^{\pi, \, 2}_{1}=0$ is associated with $\tilde{b}^{\pi, \, 1}_{1}=-\eta \neq 0$. This rules out all possible Zariski reductions. 
\par These arguments are sufficient to conclude the Taylor pricing model with two period contracts has the desired recursive equilibrium. 
\end{proof}
The long form of the errors reside in Appendix B.1
\begin{remark} Note how price dispersion vanishes at ZINSS. Consistent with Proposition 6 and the preceding discussion, however, the fact that this implies no loss of lags is specific to the shortest price spells. 
\end{remark}
\section{Convergence Results}
We arrive at the first of a pair of sections providing the titular core of the paper. There are two detachments. The first deals with the medium to long-run phenomenon of shock decay. The second relates to the conditions under which short-run dynamics can be non-monotonic. 
\subsection{Persistence and Propagation}
The first theorem of the section identifies two different lower bounds on the convergence rates, specific circumstances when they apply and the general case where both forces are present and so which ever is slowest provides the \\ supremum of the shock effect. The second guarantees swifter convergence for high-order terms. All permutations yield a dramatic improvement over \\ arithmetic convergence $O(1/T)$ under mere ergodicity. Interpretation and \\ comment finish business. 
\subsubsection{Different Convergence Rates}
\begin{theorem}
Consider a DSGE model in stochastic equilibrium with recursive equilibrium 
$$\mathbb{E}_{t}\mathbf{X}_{t+1}= f(\mathbf{X}_{t}, \, \boldsymbol{\gamma}, \, \mathbf{e}_{t}) \: \: \: \mu \: a.s. $$ where  $f \in C^{1}$ with mutual dependence, locally characterized by 
\begin{equation} \mathbb{E}_{t}\hat{\bf{X}}_{\bf{t+1}}=\bf{B}
\hat{\bf{X}}_{t} + \boldsymbol\phi\hat{\bf{e}}_{t} \end{equation}
The vector $\bf{X}_{t}=(\bf{X}^{J}_{t}, \, \bf{X}^{S}_{t})'$ contains jump 
and state variables. These vectors can be further partitioned into 
$\bf{X}^{J}_{t}=(\bf{
X}^{J}_{t, \, t}, \, \bf{X}^{J}_{t-1, \, t}, \, \cdots , \, \mathbf{X}^{J}_{t-k(X^{J}), \, t})'$ and \\ 
$\bf{X}^{S}_{t}=(\bf{X}^{S}_{t, \, t}, \, \bf{X}^{S}_{t-1, \, t}, \, \cdots , \, \mathbf{X}^{S}_{t-k(X^{S}), \, t})'$, where the first subscript always refers to when the underlying variable is realized and the second is the standard index for its date, within the recursive equilibrium.
Likewise, $\mathbf{e}_{t}=(\mathbf{e}_{t, \, t}, \cdots , 
\mathbf{e}_{t-k(e),\, t})'$ consists of exogenous variables.\, 
The underlying errors are assumed to be \\ stationary and auto-regressive $(AR(1))$ with persistence $\boldsymbol{\rho}$. $\boldsymbol{\gamma}$ groups together all the structural parameters.
\par  Define  $\bar{\lambda}_{2}=\max\{ \vert \lambda_{i}\vert : \, \vert \lambda_{i}\vert < 1\}$ and $\bar{\rho}=\max\{\vert\boldsymbol{\rho}_{i}\vert\}$.Finally, suppose $T >> k(e)+1$, then the following convergence rates apply
\begin{enumerate}[(i)]
\item In general, following an initial shock $\mathbf{e}_{0, \, 0}$: $$\vert \mathbb{E}_{0}\, \mathbf{\hat{X}}_{T}\vert \leq
C\, (\max\{\bar{\lambda}_{2}, \, \bar{\rho}\})^{T}\, 
\vert \mathbf{\hat{e}}_{0, \, 0}
\vert $$
\item If $\mathbf{\hat{X}}^{\mathbf{S}}_{t}= \emptyset$ then 
$\vert \mathbb{E}_{0}\, 
\mathbf{\hat{X}}_{T} \vert
\leq
C\, \bar{\rho}^{T}\, \vert \mathbf{\hat{e}}_{0, \, 0}
\vert $
\item The effect of the initial state of the economy $\mathbf{X}_{0, \, 0}$ 
declines such that 
\\  $\vert \mathbb{E}_{0}\, 
\mathbf{\hat{X}}_{T}\vert 
\leq C \, \bar{\lambda}^{T}_{2}\, \vert 
\mathbf{\hat{X}}_{0, \, 0}\vert $
\\ where $C>0$ is a sufficiently large constant that may vary for convenience from line to line. 
\end{enumerate}
\end{theorem}
\begin{proof}

The mixing properties inherent in stochastic equilibrium mean that $\mu$ a.s. the long-run convergence is determined by the linear approximation around the stochastic steady state. 
The proof therefore amounts to an exercise in linear algebra and basic analysis. The long-time asymptotic is designed to 
cut around the non-monotone short-run dynamics that are the focus of the next subsection.  \par The pertinent diagonalization is 
\begin{equation} 
\bf{B}=\bf{P^{-1}}\, \boldsymbol \Lambda\, \bf{P}
\end{equation} 
Instances of non-diagonalizable matrices $\bf{B}$ can be dealt with by limit taking because the singularities are removable. The process is to solve the jump and state variables separately.\footnote{The dimensional restrictions this requires are covered in detail in the proof of SELCKE Theorem 3.} It will therefore prove beneficial to work with three partitions.
$\boldsymbol{\phi}=(\boldsymbol{\phi_{1}}, \, \boldsymbol{\phi_{2}})'$ and 
\begin{equation} \bf{P}=\begin{bmatrix} \bf{P}_{11} & \bf{P}_{12} \\
\bf{P}_{21} & \bf{P}_{22}
\end{bmatrix} \end{equation}
and its inverse takes the form 
\begin{equation} \bf{P}^{-1}=\begin{bmatrix} \bf{Q}_{11} & \bf{Q}_{12} \\
\bf{Q}_{21} & \bf{Q}_{22}
\end{bmatrix} \end{equation}
\par Here, $\bf{\Lambda}$ is a special matrix. where the eigenvalues are organized in \\ descending order along the diagonal, with all other elements $0$. $\bf{\Lambda}_{1}$ incorporates the eigenvalues from outside the unit circle, while those inside lie in $\bf{\Lambda}_{2}$. It appears as 
\begin{equation} \bf{\Lambda}=\begin{bmatrix} \bf{\Lambda}_{1} & 0 \\ 0 & \bf{\Lambda}_{2}  \end{bmatrix}\end{equation}
Lastly, innovations are displayed as 
\begin{equation}\hat{\bf{e}}_{t, \, t-1} =\mathbf{R}\hat{\bf{e}}_{t-1, \, t-1} + \boldsymbol{\hat{\xi}}_{t}\end{equation}
such that 
\begin{equation} \bf{R}=\begin{bmatrix} \ddots & 0 & \cdots  & 0\\  0 &  \ddots & \cdots & 0 \\ \vdots & \vdots & \boldsymbol{\rho} &\vdots \\ 0 & 0 & \cdots & \ddots\end{bmatrix}\end{equation}
and $\boldsymbol{\xi}_{t}$ is white noise. Lastly, $\Omega_{t}$ will denote information available up to time $t$.
\par Consider an economy struck by a shock at time $0$, after a sufficiently long time $t$, SELCKE Theorem 3, underscored by its Supplementary Appendix E.3, shows that
\begin{multline}
\hat{\bf{X}}^{\bf{J}}_{\bf{t}} \approx -\sum_{i=1}^{t}\bf{Q}_{11}\, \Lambda_{2}^{i-1}\, (\bf{Q}_{21}^{-1}\, \bf{Q}_{22}-\Lambda_{2}\, \bf{Q}_{21}^{-1}\, \bf{Q}_{22}\, \Lambda^{-1}_{1}) \times \\ \sum_{k=0}^{\infty}\boldsymbol{\Lambda_{1}^{-k}}( \bf{P}_{11}\,\boldsymbol{\phi}_{2} + \bf{P}_{12}\, \boldsymbol{\phi}_{1})\, \mathbb{E}_{t}\, (\hat{\bf{e}}_{t+ k-i}\, \vert \, \Omega_{t-i}) - \\ \sum_{k=0}^{\infty}\bf{P}_{12}^{-1}\boldsymbol{\Lambda_{1}^{-
(k+1)}}(\bf{P}_{11}\,\boldsymbol{\phi}_{2} + \bf{P}_{12}\, \boldsymbol{\phi}_{1})\, \mathbb{E}_{t}\, (\hat{\bf{e}}_{t+ k}\, \vert \, \Omega_{t}) + \sum_{i=1}^{t}\bf{Q}_{11}\, \Lambda_{2}^{i-1}\, \bf{Q}_{21}^{-1}\, \boldsymbol{\phi}_{2}\, \hat{\bf{e}}_{t-i}
\end{multline}
\begin{multline}
\hat{\bf{X}}^{\bf{S}}_{\bf{t}} \approx -\sum_{i=1}^{t}\bf{Q}_{21}\, \Lambda_{2}^{i-1}\, (\bf{Q}_{21}^{-1}\, \bf{Q}_{22}-\Lambda_{2}\, \bf{Q}_{21}^{-1}\, \bf{Q}_{22}\, \Lambda^{-1}_{1}) \times \\ \sum_{k=0}^{\infty}\Lambda_{1}^{-
k}(\bf{P}_{11}\,\boldsymbol{\phi}_{2} + \bf{P}_{12}\, \boldsymbol{\phi}_{1})\, \mathbb{E}_{t}\, (\hat{\bf{e}}_{t+ k-i}\, \vert \, \Omega_{t-i}) + \sum_{i=1}^{t}\bf{Q}_{21}\, \Lambda_{2}^{i-1}\, \bf{Q}_{21}^{-1}\, \boldsymbol{\phi}_{2}\, \hat{\bf{e}}_{t-i} 
\end{multline} 
where any singular sub-matrix can be re-imagined as its pseudo-inverse \\ counterpart without trouble. 
\par To specialize to the situation of (i) it is imperative to dive down to the primitive disturbances. It is clear that 
\begin{equation} \mathbb{E}_{0}\, \hat{\bf{e}}_{k(e)+1, \,0} = \bf{R}\, \hat{\bf{e}}_{k(e), \,0}\end{equation}
To keep the derivation as simple as possible bound 
\begin{equation}\vert \hat{\bf{e}}_{t, \, 0} \vert \leq \bar{M}\, \vert \hat{\bf{e}}_{0, \, 0} \vert\end{equation}
for $0 \leq t \leq k(e)$, where $\bar{M}$ must be larger than expectation of the peak size of the error term, which must exist by assumption that we are in a first-order environment with mathematical expectations.\footnote{It is possible to obtain a quantitative handle on $\bar{M}$. 
\begin{equation*}
\bf{e}_{t}
=  \boldsymbol{\Gamma}_{0}\, \bf{e}_{t, \, t}  + \cdots + \boldsymbol{\Gamma}_{k(e)}\, \bf{e}_{t, \, t-k(e)} 
\end{equation*}
It follows that 
\begin{equation*} \mathbb{E}\, \vert \bf{e}_{t}\vert \leq (\vert \boldsymbol{\Gamma}_{0}\vert  + \cdots + \vert \boldsymbol{\Gamma}_{k(e)}\vert)\, \mathbb{E}\,\vert \bf{e}_{0,\, 0}\vert \end{equation*}
Hence, $\bar{M} \geq  \vert \bf{R} \vert^{k(e)} \, (\vert \boldsymbol{\Gamma}_{0}\vert  + \cdots + \vert \boldsymbol{\Gamma}_{k(e)}\vert)$ } \par Application to the context of 
(169)-(170) follows after taking (Euclidean) norms and a first application of their sub-additive and sub-multiplicative properties.\footnote{Strictly speaking, the matrix norms are a special type of norm induced by the vector norm. Consult \cite{meyer2023matrix} for further explanation.}
\begin{multline}
\vert \mathbb{E}_{0}\, \hat{\bf{X}}^{\bf{J}}_{T}\vert \leq 
\bar{M}\vert \sum_{i=1}^{T}\bf{Q}_{11}\, \Lambda_{2}^{T-(i-1)}\, (\bf{Q}_{21}^{-1}\, \bf{Q}_{22}-\Lambda_{2}\, \bf{Q}_{21}^{-1}\, \bf{Q}_{22}\, \Lambda^{-1}_{1})\\ \times\,(1-\Lambda^{-1}_{1})^{-1} \,(\bf{P}_{11}\,\boldsymbol{\phi}_{2} + \bf{P}_{12}\, \boldsymbol{\phi}_{1})\,(1-\bf{R})^{-1}\, \bf{R}^{i-1}\, \vert \vert\hat{\bf{e}}_{0, \, 0}\vert
+ \\ 
\bar{M}\vert\bf{P}_{12}^{-1}\boldsymbol{\Lambda_{1}^{-1}}(\bf{I}-\boldsymbol{\Lambda_{1}^{-1}})^{-1}(\bf{P}_{11}\,\boldsymbol{\phi}_{2} + \bf{P}_{12}\, \boldsymbol{\phi}_{1})\,(\bf{I}-\bf{R})^{-1}\, \bf{R}^{T}\vert \, \vert \hat{\bf{e}}_{0, \, 0} \vert
+ \\ \bar{M}\vert \sum_{i=1}^{T}\bf{Q}_{11}\, \Lambda_{2}^{T-(i-1)}\, \bf{Q}_{21}^{-1}\, \boldsymbol{\phi}_{2}\, \bf{R}^{i-1}\vert \, \vert \hat{\bf{e}}_{0, \, 0}\vert 
\end{multline}
\begin{multline} 
\vert \mathbb{E}_{0}\hat{\bf{X}}^{\bf{S}}_{T}\vert \leq 
\bar{M}\, \vert \sum_{i=1}^{T}\bf{Q}_{21}\, \Lambda_{2}^{T-(i-1)}\, (\bf{Q}_{21}^{-1}\, \bf{Q}_{22}-\Lambda_{2}\, \bf{Q}_{21}^{-1}\, \bf{Q}_{22}\, \Lambda^{-1}_{1})\\ \times\,(1-\Lambda^{-1}_{1})^{-1} \,(\bf{P}_{11}\,\boldsymbol{\phi}_{2} + \bf{P}_{12}\, \boldsymbol{\phi}_{1})\,(1-\bf{R})^{-1}\, \bf{R}^{i-1}\vert \, \vert \hat{\bf{e}}_{0, \, 0}\vert
+ \\ \bar{M}\, \vert \sum_{i=1}^{T}\bf{Q}_{21}\, \Lambda_{2}^{i-1}\, \bf{Q}_{21}^{-1}\, \boldsymbol{\phi}_{2}\, \bf{R}^{i-1}\vert \, \vert \hat{\bf{e}}_{0, \, 0}\vert
\end{multline}
Factorization and further norm bounding give
\begin{multline}
\vert \mathbb{E}_{0}\,\hat{\bf{X}}^{\bf{J}}_{T}\vert  \leq 
\bar{M} \, \bigg\{\vert \bf{Q}_{11} \vert 
\, \vert\bf{Q}_{21}^{-1}\, \bf{Q}_{22}-\Lambda_{2}\, \bf{Q}_{21}^{-1}\, \bf{Q}_{22}\, \Lambda^{-1}_{1}\vert\\ \times \vert \bf{I}-\Lambda^{-1}_{1}\vert^{-1} \vert\bf{P}_{11}\,\boldsymbol{\phi}_{2} + \bf{P}_{12}\, \boldsymbol{\phi}_{1}\vert \,\vert\bf{I}-\bf{R}\vert^{-1}
\vert \bar{V} \vert \, + \\ 
\vert \bf{P}_{12}^{-1}\vert \vert\boldsymbol{\Lambda_{1}}\vert^{-1}\, \vert \bf{I}-\boldsymbol{\Lambda_{1}^{-1}}\vert^{-1} \,
\vert \bf{P}_{11}\,\boldsymbol{\phi}_{2} + \bf{P}_{12}\, \boldsymbol{\phi}_{1} \vert \, \vert\bf{I}-\bf{R}\vert^{-1}\, \vert \bf{R}\vert^{T}\, 
+ \\ 
\vert \bf{Q}_{11}\vert \vert \bf{Q}_{21}^{-1}\vert \, \vert \boldsymbol{\phi}_{2}\vert \, 
\vert \bar{V} \vert \bigg\}\, \vert \hat{\bf{e}}_{0, \, 0} \vert
\end{multline}
where, to treat a removable singularity, I defined 
\begin{equation}
 \bf{V}=\begin{cases}
     (\boldsymbol{{\Lambda}_{2}^{T}}-\bf{R}^{T})(\bf{I}-\boldsymbol{{\Lambda}_{2}^{-1}}\, \bf{R})  & \text{if $\boldsymbol{{\Lambda}_{2}}\neq \bf{R}$}\\ 
     T\, \bf{R}^{T} & \text{otherwise}
 \end{cases}   
\end{equation}
such that 
\begin{equation}
\vert \bf{V} \vert \leq \vert \bar{V} \vert 
\end{equation}
whilst 
\begin{equation}
 \bf{\bar{V}}=\begin{cases}
(\vert\boldsymbol{{\Lambda}_{2}\vert^{T}} + \vert\bf{R}\vert^{T})\vert\bf{I}-\boldsymbol{{\Lambda}_{2}^{-1}}\, \bf{R}\vert  & \text{if $\boldsymbol{{\Lambda}_{2}}\neq \bf{R}$}\\ 
T \, \vert\bf{R}\vert^{T} & \text{otherwise}
 \end{cases}   
\end{equation}
For convenience, I also made use of the matrix norm inequality $\vert \bf{A}^{-1}\vert \leq \vert A\vert^{-1}$ for $\bf{A}\neq 0$ and the fact that $\vert \bf{D}^{n} \vert =\vert D \vert^{n}$ for any diagonal matrix raised to a power $n \in \mathbb{R}$. 
In a similar fashion, the estimate for (155) is
\begin{multline}
\vert \mathbb{E}_{0}\, \hat{\mathbf{X}}^{\mathbf{S}}_{T}\vert\leq 
\bar{M}\, \bigg\{\vert \bf{Q}_{21} \vert
\, \vert\bf{Q}_{21}^{-1}\, \bf{Q}_{22}-\Lambda_{2}\, \bf{Q}_{21}^{-1}\, \bf{Q}_{22}\, \Lambda^{-1}_{1}\vert\\ \times \vert \bf{I}-\Lambda^{-1}_{1}\vert^{-1} \vert\bf{P}_{11}\,\boldsymbol{\phi}_{2} + \bf{P}_{12}\, \boldsymbol{\phi}_{1}\vert \,\vert\bf{I}-\bf{R}\vert^{-1}
\vert \bar{V} \vert 
+ \vert \boldsymbol{\phi}_{2}\vert \, 
\vert \bar{V}\vert + \\ 
\vert \bf{P}_{12}^{-1}\vert \vert\boldsymbol{\Lambda_{1}}\vert^{-1}\, \vert \bf{I}-\boldsymbol{\Lambda_{1}^{-1}}\vert^{-1} \,
\vert \bf{P}_{11}\,\boldsymbol{\phi}_{2} + \bf{P}_{12}\, \boldsymbol{\phi}_{1} \vert \, \vert\bf{I}-\bf{R}\vert^{-1}\, \vert \bf{R}\vert^{T}\, \vert \bigg\} \, \vert \hat{\bf{e}}_{0, \, 0}\vert
\end{multline}
To complete the proof of part (i), first realize that the supremum of both 
(175) and (179) can be broken into terms of 
$O(\vert\boldsymbol{{\Lambda}_{2}\vert^{T}})$ and $O(\vert\bf{R}\vert^{T} )$, owing to the intervening relations
(176)-(178). Second, sum (175) and (179)
to obtain a bound on the aggregate norm 
$(\vert\mathbb{E}_{0}\, \hat{\bf{X}}_{T}\vert \leq \vert\mathbb{E}_{0}\,  \hat{\bf{X}}^{\bf{J}}_{T}\vert + \vert\mathbb{E}_{0} \, \hat{\bf{X}}^{\bf{S}}_{T}\vert )$, due to the triangle inequality. Finally, 
call to mind the properties that $\vert\boldsymbol{\Lambda}_{2}\vert=\bar{\lambda}_{2}$ and $\vert\bf{R}_{2}\vert=\bar{\rho}$ to compute the grand upper bound, then conflate constants to finish.
\par For parts (ii) and (iii) the mathematical analysis is unchanged. The only difference is the starting point. With regard to statement (ii), the entire system collapses down to the second term in 
(173). For part (iii), \cite{blanchard1980solution} teaches us that\footnote{By way of warning, there are some differences in notation between the two papers.}  
\begin{equation} \mathbb{E}_{0}\, \hat{\bf{X}}^{\bf{J}}_{T} \approx \bf{Q}_{21}\, \Lambda^{T}_{2}\, \bf{Q}^{-1}_{21}\, \hat{\bf{X}}^{S}_{0}\end{equation}
\begin{equation} \mathbb{E}_{0}\,\hat{\bf{X}}^{\bf{S}}_{T} \approx \bf{Q}_{11}\, \Lambda^{T}_{2}\, \bf{Q}^{-1}_{11}\, \hat{\bf{X}}^{S}_{0}\end{equation}
Following the arguments defending 
(172), it follows that there must exist some $\bar{M}$ such that  $\vert \mathbf{X}^{S}_{0}\vert \leq \bar{M}\, \vert \mathbf{X}^{S}_{0, \, 0}\vert$. 
Hence, by the estimates explained above 
\begin{equation} \vert \mathbb{E}_{0}\, \hat{\mathbf{X}}^{\mathbf{J}}_{T}\vert \leq \bar{M}\, \vert \mathbf{Q}_{21}\vert \, \bar{\lambda}^{T}_{2}\, \vert \mathbf{Q}_{21}\vert^{-1} \, \vert \hat{\mathbf{X}}^{\mathbf{S}}_{0, \, 0}\vert= C\, \bar{\lambda}^{T}_{2}\,\vert \hat{\bf{X}}^{\mathbf{S}}_{0, \, 0}\vert\end{equation}
where $C \geq \bar{M}$.\footnote{In this case no (finite) \emph{a priori} bounds are available on $\bar{M}$, although it should be possible to construct a more general solution to rectify this.} The same playbook works for $\vert \mathbb{E}_{0}\, \hat{\mathbf{X}}^{\mathbf{S}}_{T}\vert$. 
Thereby the aggregate claim is proven.
\end{proof}
\begin{remark} All estimates of this type are sharp because as $T \rightarrow \infty$ the largest exponent dominates all others. \end{remark}
This surely constitutes a powerful contribution to both theoretical and \\ econometric understanding. The salient feature is no matter how great the non-linearity, eventually the asymptotics of the linear model predominate in stochastic equilibrium.
\par It is possible to be more precise for the subset of variables  
that end up too small to feature in the relevant first-order approximation. Before embarking on this non-linear avenue new terminology is needed. 
\begin{definition}
A stochastic process $\{S_{t}\}$ is said to exhibit super-exponential convergence if for $\mu$ a.e. $S_{0}$ and any $(a, \, C)$, 
$\mathbb{E}_{0}\, \vert S_{T} -S_{0}\vert \leq C\, a^{T}$, for $T$ sufficiently large. 
\end{definition}
\begin{theorem}
Consider a limiting sequence of DSGE models $\{f_{i}\}$, each \\ consistent with the hypotheses of the previous theorem. Suppose further that 
$\bf{X}_{t}=(\bf{x}_{t}(1), \, \bf{x}_{t}(2), \, \cdots , \, 
\bf{x}_{t}(\infty))'$,  where $\bf{x}_{t}(n)$
refers to variables of order $n$ around the limiting stochastic equilibrium. The following convergence rates \\ apply 
\begin{enumerate}[(i)]
\item In general, following an initial shock $\mathbf{e}_{0, \, 0}$: $$\vert\mathbb{E}_{0}\, 
\mathbf{\hat{x}}_{T}(n)\vert
\leq 
C\, (\max\{\bar{\lambda}_{2}, \, \bar{\rho}\})^{nT}\, \vert \mathbf{\hat{e}}_{0, \, 0}\vert^{n}$$
\item If $\mathbf{\hat{X}}^{S}_{t}= \emptyset$ then 
$\vert \mathbb{E}_{0}\, 
\mathbf{\hat{x}}_{T}(n)\vert \leq
C \, \bar{\rho}^{nT}\, \vert \mathbf{\hat{e}}_{0, \, 0}\vert^{n} $
\item The effect of the initial state of the economy $\mathbf{X}_{0, \, 0}$ 
declines such that 
\\  $\vert\mathbb{E}
_{0}\, \mathbf{\hat{x}}_{T}(n)\vert \leq
C \, \bar{\lambda}^{nT}_{2}\, \vert \mathbf{\hat{X}}_{0, \, 0}(1)\vert^{n} $
\item In all cases $\vert\mathbb{E}
_{0} \, \mathbf{\hat{x}}_{T}(\infty)\vert$ 
decays super-exponentially.
\end{enumerate}
\end{theorem}
\begin{proof}
The method is an extension of the last proof. Once the expansion has been formulated, the analysis
is straightforward. I will treat parts (i)-(iii) \\ simultaneously.
\par Some new vocabulary is needed, let $\bf{X}_{t}=\big(\bf{
x}_{t}(n), \, \bf{
X}_{t}(-n)\big)'$ 
with $\bf{X}_{t}(-n)=(\bf{x}_{t}(1), \, \cdots, \, \bf{x}_{t}(n-1), \cdots)'$. 
Apply this decomposition at the sub-vector level, in particular $\bf{X}_{t, \,t}=\big(\bf{x}_{t, \, t}(n), \, \bf{X}_{t, \, t}(-n)\big)'$. 
Further define families of vectors $\{ \bf{c}_{s, \, n}(\bf{x})\}^{n-1}_{s=1}$ and $\{ \bf{c}_{s, \, n}(\bf{e})\}^{n-1}_{s=1}$ conformal with $\bf{x}_{t}$ and $\bf{e}_{t}$ respectively. 
\par In general, 
\begin{equation}
 \bf{\hat{x}}_{t}(n) \approx (\bf{\hat{x}}_{t}(1)+ \bf{\hat{e}}_{t})\prod_{s=1}^{n-1}(\bf{c}_{s,\, n}(x)\, \bf{\hat{x}}_{t}(1)+ \bf{c}_{s,\, n}(e)\, \bf{\hat{e}}_{t})
\end{equation}
Passing initial expectations and then
repeated appeals to sub-additivity and sub-multiplicity lead to the conclusion that our target $\vert\mathbb{E}_{0}\,
\mathbf{\hat{x}}_{T}(n)\vert$ is an $n^{th}$-order polynomial in $\vert\mathbb{E}_{0}\,\bf{\hat{e}_{T}}\vert$. Monotonicity implies the estimate \begin{equation}
 \vert\mathbb{E}_{0}\,
\mathbf{\hat{x}}_{T}(n)\vert \leq C\, \max\{\vert\mathbb{E}
_{0}\,\bf{\hat{x}_{T}(1)}\vert^{n}, \, \vert\mathbb{E}
_{0}\,\bf{\hat{e}_{T}}\vert^{n}\}
\end{equation}
To finish, substitute in the estimates from Theorem 4 (i)-(iii) in place of \\ $\vert\mathbb{E}_{0}\,\bf{\hat{x}_{T}(1)}\vert$ and the error decay bound. 
Finally, evaluate the maximum, factorize and merge constants.
\par Item (iv) is even easier. Statement (i) implies that \\ $\vert \mathbb{E}
_{0}\, \mathbf{\hat{X}}_{T}(\infty)\vert < C\, (\max\{\bar{\lambda}_{2}, \, \bar{\rho}\})^{nT}\, \vert \mathbf{\hat{e}}_{0, \, 0}\vert^{n}$, for all $n>0$. Hence, sending $n \rightarrow \infty$ is sufficient to demonstrate compliance with Definition 2. The entire Theorem is proven. 
\end{proof}
\begin{remark}  By convention, disturbances are always first-order. The \\ preference innovation enters linearly. This means that all high-order $(n \geq 2)$ read $O(\mathbf{\hat{e}}_{T})\, O(\mathbf{\hat{X}}^{n-1}_{T})$. This is probably advantageous from the standpoints of both small sample econometric estimation and mathematical analysis. Nevertheless, a richer selection of shocks would permit the full panopticon of non-linear effects to blossom. \end{remark}
\begin{remark} The stipulation that the underlying shocks be mutually \\ uncorrelated, so as to make $\bf{R}$ diagonal, is not an imposition, since there are classical procedures to produce orthogonal shocks that should be adaptable to the present setting (\cite{hamilton2020time}.) \end{remark}
\begin{remark} The analysis would go through if we dropped the restriction that errors were $AR(1)$, with $\bar{\rho}$ re-imagined as $\bar{\lambda(e)}$ the magnitude of the modulus of the largest eigenvalue.  \end{remark}
\subsubsection{Discussion}
Some signification of the shocks is in order. The notion that initial \\ differentials unrelated to basal economic incentives decrease over time, has been a major theme in economics, since \cite{solow1956contribution} and \cite{swan1956economic}. This result can be seen as an extension of their model along three dimensions, the first is to rigorously show that this is true for a sequence of shocks, not just a one off impulse. Secondly, I show the result is applicable to Keynesian, as well as Classical forces.\footnote{This underlines the importance of projecting stochastic equilibrium theory to economies with endogenous capital accumulation. I predict these will be available in the next few years.} Thirdly, the amplitude of the deviation is unrestricted. The historical context of these papers proffers upon us a \\ natural example of the sort of one-off plausibly unanticipated shock, appropriate to invoke the asymptotics of (iii).\footnote{These are often referred to as "MIT shocks" by macroeconomists.} 
A sizable war that causes extensive capital destruction, would be sufficiently abnormal, difficult to forecast from historical data and 
large relative to the normal business cycle to be an interesting \\ candidate shock. 
 $(\mathbf{\hat{X}}_{0, \, 0}-\mathbb{E}\, X<0)$.
\cite{crafts1996economic} and \cite{harrison2000economics} provide historical accounts pertaining to the recovery of post-war Europe. 
\par A prescient peace time example, might be a financial or government debt crisis. Narrative evidence suggest these commotions are often unanticipated by many agents, see for example the extensive survey \cite{reinhart2011time}.
A common feature appears a sudden re-evaluation of risk, often called a "Minsky moment" \cite{wray2011minsky}, after Hyman Minsky. Indeed, this entered into the formal lexicon through models, such as \cite{gertler2020macroeconomic}, which are initiated with by a downward revaluation of existing assets. 
\par The convergence analysis provides a window on the polydromy phenomenon discovered in my first paper. The theoretical experiment involves starting from a situation with shocks small enough to consider a first-order approximation \\ accurate and systematically reducing the shock size in each  simulation long enough to precisely infer statistical properties. Initially, the rate of decline of price dispersion $(\Delta)$ roughly matches that of the other endogenous variables but as volatility vanishes the rate of decline of price dispersion starts to quicken \\ relative to its counterparts, asymptotically approaching twice their speed of decay.
\par The exponential decay of initial conditions should be a burden removed from applied researchers. It can be viewed as a consequence of the 
substitutions (72) and (73). It is likely to extend to flexible price models with heterogeneity. 
\par Finally, an explicit example of (ii) is exposited in Supplementary Appendix A.3. It is of more technical than policy interest. For long-time practicioners, Theorem 4 suggests that the previous focus on exogenous persistence, as the driving force of the business cycle may have been overstated, if not misplaced.\footnote{There has been a long-running concern, elucidated famously in \cite{chari2009new}, that spurious exogenous forces have compensated for missing endogenous dynamics, that I \\ uncovered in SELCKE.}
\subsection{Impulse Response Peak Bounds}
Like the previous, this subsection has two slices. The first sets out the underlying dynamic, theory. The second angles for scientific 
content.
\subsubsection{Theory}
\begin{theorem}Consider a DSGE model in a small noise limiting stochastic \\ equilibrium $(\vert \varepsilon \vert \rightarrow 0)$ characterized ($\mu$ a.s.) by a recursive equilibrium with mutual dependence
\begin{equation} \mathbb{E}_{t}\, \hat{\bf{X}}_{t+1}=\bf{B}
\, \hat{\bf{X}}_{t} + \boldsymbol{\phi}\, \hat{\bf{e}}_{t} \end{equation}
where $\bf{\hat{X}}_{t}$ are the endogenous variables, whilst $\mathbf{\hat{e}}_{t}$ consists of $AR(1)$ exogenous variables, 
and $\boldsymbol{\gamma}$ is a vector of parameters. 
Suppose further that the following time ordering of the realizations arises 
$\bf{\hat{X}}_{t}=(\bf{\hat{X}}_{t, \, t}, \, \bf{\hat{X}}
_{t-1, \, t}, \, \cdots , \, \mathbf{\hat{X}}_{t-k(X), \, t})'$ and $\mathbf{\hat{e}}_{t}=(\mathbf{\hat{e}}_{t, \, t}, \cdots , \mathbf{\hat{e}}_{t-k(e), \, t})'$.
Following a shock, the expected time after which the impulse response peaks is sharply bounded above by $k(e) +1$, regardless of the length of lag $k(X)$ embodied in the endogenous variables. 
\end{theorem}
\begin{proof}At the outset, it is legitimate to focus on the central case of the limiting impulse response function, as though the model were certainly equivalent. This because as $(\vert \varepsilon \vert\rightarrow 0)$ the 
size of subsequent shocks becomes too small to \\ influence the peak of the impulse response, formally $\vert \varepsilon \vert << 
\vert \mathbb{E}_{t}\, \mathbf{\hat{X}}_{t+h}-\mathbb{E}_{t}\, \mathbf{\hat{X}}_{t+h'}\vert $, for any $h, \, h' >0$. 
The rest of the proof consists of carefully reviewing two \\ extreme cases. For notational simplicity, I work with scalars and send $\rho \rightarrow 0$. 
\\ \\ \textbf{Case 1: Maximal Lag Driven Error Persistence}
\\ \\ Consider the limiting system where the entire weight of expected inflation \\ dynamics rests on the error farthest in the past
\begin{equation} \mathbb{E}_{t}\, \hat{X}_{t+1} \rightarrow \hat{e}_{t-k(e)}\end{equation}
It is clear that the solution is 
\begin{equation} 
\hat{X}_{t+1} \rightarrow \hat{e}_{t-k(e)}\end{equation}
Hence, following a shock at time $t=0$
\begin{equation}
\mathbb{E}_{0}\, \hat{X}_{T} \rightarrow
 \begin{cases} 
 0
 &\text{if $T \neq k(e)$} \\ 
 \hat{e}_{0} &\text{if $T= k(e)+1$}
 \end{cases}
\end{equation}
Clearly, the impulse response peaks after $k(e)$ periods. This proves the claim justifying maximal achievable persistence. 
\par Turning to the irrelevance of lags in the endogenous variables to the bound, the maximum possible impact of endogenous persistence on the impulse \\ function must occur where all the endogenous dynamics come from the variable determined longest ago. 
\\ \\ \textbf{Case 2: Maximal Lag Driven Endogenous Persistence}\\ \\ 
For concreteness, let \begin{equation} \mathbb{E}_{t}\, \hat{X}_{t+1} \rightarrow \lambda \, \hat{X}_{t-k(e)}+ \hat{e}_{t}\end{equation}
with $\hat{X}_{t-k(e)}$ a state variable and the eigenvalue $-1< \lambda < 1$. 
Computing the exact solution yields\footnote{I solved the model by writing in lag operator terms to obtain 
$$\hat{X}_{t}=\bigg(\frac{(1+b)\, \mathbb{L}}{1-\lambda \, \mathbb{L}^{k(e)}}\bigg)\, \hat{e}_{t}$$
where $b \equiv \pi_{t}-\mathbb{E}_{t-1}\, \pi_{t}=0$,
with suitable care taken to distinguish realized from expected outcomes.} 
\begin{equation} \hat{X}_{t} \rightarrow \hat{e}_{t-1} + \lambda \, \hat{e}_{t-(k_{e}+1)} + \lambda^{2} \, \hat{e}_{t-(2k_{e}+1)} + \cdots + \lambda^{r} \, \hat{e}_{t-(rk_{e}+1)} + \cdots + \end{equation}
This implies the time profile 
\begin{equation}
\mathbb{E}_{0}\, \hat{X}_{T} \rightarrow
 \begin{cases} 
 0
 &\text{if $T \neq r\, k(e) +1$} \\ 
 \lambda^{r}\, \hat{e}_{0} &\text{if $T= r\, k(e)+1$}
 \end{cases}
\end{equation}

Thus, the implied peak arises at $t=1 \leq k(e) +1$ meaning it contributes nothing to the peak time estimate, as asked for. 
Finally, the focus on extreme cases proves sharpness. 
\end{proof}
\begin{remark} The result is robust to relaxing the $AR(1)$ assumption provided the error peaks on impact.\end{remark}
\newpage
\subsubsection{Applications}
The practical implications for the New Keynesian framework are laid out below 
\begin{theorem}
Consider the collection of economies in stochastic equilibrium detailed in this paper. Suppose further the size of aggregate shocks $(\vert \varepsilon \vert\rightarrow 0)$. Let these models be hit by an initial shock $\hat{\psi}_{0}$. The expected peak of the impulse response $H$ obeys the following bounds 
\begin{enumerate}[i]
\item For Calvo and menu cost pricing $H\leq 2$
\item For Taylor pricing $H \leq 2M-1$
\item With Rotemberg pricing $H=0$.
\end{enumerate}
\end{theorem}
\begin{proof}
All are applications of the last theorem. To prove part (i) apply Theorem 4 to Proposition 4 SELCKE for Calvo and Proposition 4 here for menu costs. For part (ii) use Proposition 5 for Taylor contracts. Part (iii) is the consequence of the exact form solution in Section A.3.
\end{proof}
\begin{remark} It can also be seen as an application of part (iii) of Theorem 4 because of the well-known absence of any state variables in the model.
\end{remark}
This justifies Result 2. Further hermeneutics are in order. The small noise limit assumption is necessary to make clean predictions, otherwise we would be in the world with a distribution of business cycle lengths. These are upper bounds that may not be met by particular models. The overarching lesson is that 
\begin{principle}
Stochastic Equilibrium, where arbitrarily low volatility common shocks are laid on top of a mean field of individual specific uncertainty, offers a powerful tool to obtain tractable dynamic predictions and to estimate causal effects of introducing heterogeniety into benchmark models of aggregate \\ uncertainty.
\end{principle}
\begin{remark}
With small aggregate noise $(\vert \varepsilon \vert \rightarrow 0)$, one can think of menu costs as the basic Calvo model with the addition of idiosyncratic demand shocks, if the cost parameter $c$ is calibrate to match, for example, the average reset frequency $(1-\alpha)$. This is common practice in the literature, where beyond its empirical merits, it is used as a device to estimate the extent of selection effects, associated with the choice of whether to reset prices, particular to menu costs.    
\end{remark}
Verily, a multiplicity of applications awaits. 
\section{Statistical and Numerical Approximation}
The first mission is to derive the peculiar limiting distribution of first-order approximate solutions when the stochastic steady state approaches its non-stochastic counterpart. The second assignment zooms out to the exceptional statistical trait possessed by any approximation local to stochastic equilibrium. The final act constructs a consistent regression based test to determine whether a DSGE model has a stochastic equilibrium.
\subsection{Super-Consistency}
There are three stages in this subsection. The first proves the Theorem. The \\second discusses potential drawbacks in terms of the required assumptions, whilst the third suggests improvements to econometric procedure. 
\begin{theorem}Consider a DSGE model consistent with the hypotheses of \\ Theorem 4. Suppose that as $\varepsilon^{2} \rightarrow 0$, there exists a subset of parameters $\boldsymbol{\gamma_{0}}
\subseteq \boldsymbol{\gamma}$, such that asymptotically (as $T \rightarrow \infty$) $\boldsymbol{\hat{\gamma}_{0}} \sim \mathcal{N}(\boldsymbol{\gamma_{0}}, \, \mathbf{K}\,(\boldsymbol{\Sigma }\boldsymbol{\Sigma}')^{-1}\, \mathbf{K}'/T)$, then in the limit where 
$\vert \varepsilon \vert \rightarrow 0$, $\boldsymbol{\hat{\gamma}_{0}} \sim \mathcal{N}(\boldsymbol{\gamma_{0}}, \, \mathbf{K}\,(\boldsymbol{\Sigma }\boldsymbol{\Sigma}')^{-1}\, \mathbf{K}'/T^{2})$ and thus \\ $\hat{\boldsymbol{\gamma}}
=\boldsymbol{\gamma} + O(1/T)$.
\end{theorem}
\begin{proof} The distributional equivalence is a consequence of the fact that \\ $\mathbf{K}\,(\boldsymbol{\Sigma }\boldsymbol{\Sigma}')^{-1}\, \mathbf{K}'/T \approx \sqrt{(\mathbf{K}\,(\boldsymbol{\Sigma }\boldsymbol{\Sigma}')^{-1}\, \mathbf{K}'/T)^{2}}$ once second-order terms have \\ vanished, which can be viewed as a Taylor expansion of variance-covariance matrix in $\sqrt{\boldsymbol{\Sigma}}/T$. Finally, the super-consistency claim follows from a first-order Taylor expansion of the normal distribution.\end{proof}
\begin{remark} The precepts covers any parameter whose estimator is \\ consistent and subject to a standard $O(1/\sqrt{T})$ central limit theorem. This will be the case for any Autoregressive Moving Average (ARMA) process, which is stable and invertible. Stability follows from stochastic equilibrium. Invertibility of the errors, requires analyzing the solutions of lag polynomial on a case-by-case basis. \cite{dias2018estimation} provide an iterative algorithm. \\ Otherwise, further economic information would be required for identification (consult \cite{shumway2025time} for an explanation of the equivalence \\ problem under non-invertiblity.) Note that the inference idea of \cite{funovits2024} is unworkable for lack of non-linearity.
\end{remark}
\begin{remark}It is critical to specify a subspace rather than the whole parameter set since in the $\sqrt{\varepsilon}$ limit the price elasticity of demand for each variety $\theta$ is unidentified 
because it does not feature in the underlying coefficient expressions (40)-(42) or (44)-(45). In actuality, parameters are typically over-identified but this is a more applied topic susceptible to coverage elsewhere.
\end{remark}
\begin{remark} The crucial asymptotic can be understood as convergence in \\ distribution of each error $\hat{e}_{i, \, t}$ to a 
Markov chain on $\{-\varepsilon_{i}, \, \varepsilon_{i}\}$ with transition matrix $$\bf{R}_{i}=\begin{bmatrix} \rho_{i} & 1-\rho_{i} \\  1-\rho_{i}& \rho_{i}\end{bmatrix}$$ 
Testing the goodness-of-fit of this distribution to the estimated errors would prove the strictest possible scientific discipline of our model. 
Testing for normality would be more forgiving (less powerful). For example, 
implementing \cite{jarque1980efficient}), focusing on the third moments, would be the most direct for the weaker condition that $\varepsilon^{2} \rightarrow 0$. \end{remark}
It has been common since \cite{hamilton1989new} to model the economy as Markov switching between recession and expansion. However, there is an implicit \\ restriction here that recessions last as long as expansions, which is 
historically challenging, as basic calculations from official NBER data (available at \\ https://www.nber.org/research/business-cycle-dating) reveal that, from \\ 1946-2025, the US economy was only in recession for $124$ months, a mere $12.9\%$ of the total. Nevertheless, (169) and (170) tell us that, the distribution of the endogenous variables would still be continuous provided there were state variables, like lags of inflation, as there are in the three main New Keynesian models (Sections 3-5). Overall, this new limit should be viewed as a lower bound on variance (upper bound on precision) and be subjected to extensive computational interrogation.
\par The new limiting argument suggests a specific optimization procedure. 
\begin{principle}
In the limit where $\vert \varepsilon \vert \rightarrow 0$ the following estimation problem emerges 
$$\min_{\boldsymbol{\gamma}}\sum_{t=0}^{T}\Vert \hat{\bf{X}}_{t+1}-\mathbb{E}_{t}\, \hat{\bf{X}}_{t+1}\Vert $$
subject to $\mathbb{E}_{t}\, \hat{\bf{X}}_{t+1}=\bf{B(\boldsymbol{\gamma})}
\, \hat{\bf{X}}_{t} + \boldsymbol{\phi(\boldsymbol{\gamma})}\, \hat{\bf{e}}_{t}$
\end{principle}
It is well-known that in the limit $\varepsilon^{2} \rightarrow 0$, Maximum Likelihood (MLE) converges on Ordinary Least Squares (OLS). The new idea here is that OLS and hence MLE converge on the minimum $(L^{1})$ distance estimator. Hence, if one believes $\vert \varepsilon \vert \rightarrow 0$ this is the efficient estimation method. Minimum distance has a long history as an outlier robust estimator (see \cite{kennedy2008guide}). It is common practice to use this objective function in indirect inference estimation, where authors have focused on first order effects at the expense of any non-linearity (\cite{le2016testing}). This could offer another theoretically justified means to obtain more precise estimates and powerful hypothesis tests when the focus is on linear relationships between aggregate variables. 
\subsection{Optimal Approximation}
The centerpiece of this subheading is the proof that the stochastic equilibrium neighborhood is always the best location from which to take perturbations. Following close behind is an easily implementable regression based strategy for consistent estimation. Lastly, discourse breaks out about the ramifications of varying the loss function.
\begin{theorem} Consider the family of polynomial approximations of degree $k \geq 1$ to a recursive equilibrium consistent with the impositions of Theorem 4
$$\mathbb{E}_{t}\, \mathbf{X}_{t+1}= f(\mathbf{X}_{t}, \, \boldsymbol{\gamma}, \, \mathbf{e}_{t}) \: \: \: \mu \: a.s.$$ 
with $f \in C^{k}$
\begin{equation}
\mathbb{E}_{t}\, \bf{\hat{X}}_{t+1}(k) =(\bf{\hat{x}}_{t}(1)+ \bf{\hat{e}}_{t})\prod_{s=1}^{k-1}
 (\boldsymbol{\varsigma}_{s,\, k}(x)\, \bf{\hat{x}}_{t}(1)+ \boldsymbol{\varsigma}_{s,\, k}(e)\, \bf{\hat{e}}_{t})
\end{equation}
Direct attention to the the $k^{th}$ order Taylor expansion in the vicinity of the Stochastic Equilibrium, which can be rendered in the multi-index notation \\ (explained in \cite{reed1980methods}) as
\begin{equation} \mathbb{E}^{*}_{t}\, \mathbf{\hat{X}}_{t+1}(k) =\sum_{a=0}^{k}
\mathbf{\hat{X}}_{t}^{a}(1)(\partial^{a}f(\mathbb{E}\mathbf{X}))/a!\end{equation} 
For every Euclidean goodness of fit metric 
$\Vert \cdot \Vert$, such that in the associated norm $\mathbb{E}\vert \mathbf{X}\vert< \infty$, ponder the appertaining optimal approximation problem. 
\begin{equation}\min_{\boldsymbol{\varsigma}_{k}}\sum_{t=0}^{T}\Vert \hat{\bf{X}}_{t+1}-\mathbb{E}_{t}\,\hat{\bf{X}}_{t+1}\Vert \end{equation}
where $\boldsymbol{\varsigma}_{k}=(\boldsymbol{\varsigma}_{1, \, k}, \, \cdots, \, \boldsymbol{\varsigma}_{k, \, k})'$. 
It has the unique solution that \\ $\hat{\mathbb{E}}_{t}\, \hat{\bf{X}}_{t+1}(\boldsymbol{\hat{\varsigma}}_{k}) \rightarrow \mathbb{E}^{*}_{t}\, \mathbf{\hat{X}}_{t+1}(k)$, as $T \rightarrow \infty$, $\mu \: a.s.$
\end{theorem}
\begin{proof}
Existence follows because strong-mixing of $\mathbb{E}_{t}\bf{X}_{t+1}$ in Stochastic \\ Equilibrium implies that $\hat{\mathbb{E}}_{t}\hat{\bf{X}}_{t+1}(\boldsymbol{\hat{\varsigma}}_{k}) \rightarrow \mathbb{E}^{*}_{t}\mathbf{\hat{X}}_{t+1}(k)$, $\mu \: a.s.$ and thus in every metric induced by a convergent norm. Uniqueness comes about by noting that for any alternative  $\mathbb{E}^{'}_{t}\mathbf{\hat{X}}_{t+1}(k) \not\equiv \mathbb{E}^{*}_{t}\mathbf{\hat{X}}_{t+1}(k)$ optimality ensures the triangle \\ inequality becomes an equality  
\begin{equation} \Vert\mathbf{\hat{X}}_{t+1} -\mathbb{E}^{*}_{t}\mathbf{\hat{X}}_{t+1}(k) \Vert= \Vert\mathbf{\hat{X}}_{t+1} -\mathbb{E}^{'}_{t}\mathbf{\hat{X}}_{t+1}(k) \Vert + \Vert \mathbb{E}^{*}_{t}\mathbf{\hat{X}}_{t+1}(k)-\mathbb{E}^{'}_{t}\mathbf{\hat{X}}_{t+1}(k)\Vert\end{equation}
The second term is strictly positive thanks to the Euclidean concordance \\ between the metric and the topology of the state space. Thus, the almost surely strict inequality 
\begin{equation} \Vert\mathbf{\hat{X}}_{t+1} -\mathbb{E}^{*}_{t}\mathbf{\hat{X}}_{t+1}(k) \Vert< \Vert\mathbf{\hat{X}}_{t+1} -\mathbb{E}^{'}_{t}\mathbf{\hat{X}}_{t+1}(k) \Vert \end{equation}
justifies the uniqueness claim and completes the entire proof. 
\end{proof}
 Besides, the obvious accommodation changes to the dependent variable, there are further substantive benefits of being able to vary the metric.\\ \cite{hansen2011robustness} explicate how high order metrics (induced from $L^{n}$ $n\geq 3$) can be useful to model attitudes to risk and uncertainty. This offers practical applications to asset pricing and financial regulation (see also \\ \cite{danielsson2011financial}), which are too subtle for the 
 representative agent framework here. 
This opens up the possibility for regression based construction of optimal approximations but I leave this development for future work. 
\subsection{Numerical Analysis}
This closing session has two goals. The first task is to carry Theorems 4 and 5 over to the entire invariant measure. The final installment uses these to build an arbitrarily powerful statistical test for whether a simulated system has a Stochastic Equilibrium. 
\subsubsection{Distribution Recovery}
\begin{theorem}
Take a standard DSGE model described in Theorem 4, 
allow that the error process $\{\bf{e}_{T}\}$ is known, at least up to its invariant measure. There exists algorithms capable of generating a sequence of approximate measures $\{ \mu_{T}\}$ where $\mu_{T} \rightarrow \mu$ with respect to Wasserstein metrics $(W)$, at the following rates
\begin{enumerate}[i]
\item    $O(\max\{\bar{\lambda}_{2}, \, \bar{\rho}\})^{T})$ in $W^{1}$
\item    $O(\max\{\bar{\lambda}_{2}, \, \bar{\rho}\})^{nT})$ in $W^{n}$
\end{enumerate}
\end{theorem}
\begin{proof}
Focus on part (i), as part (ii) is a corollary. The proof is based on a logical troika. First comes a straightforward derivation. Second is a standard programming exercise, whilst thirdly, innovative topological reasoning makes these procedures watertight. 
\par For the technical derivation, recursively substitute to obtain the infinite functional equation form of the model where
\begin{multline} \bf{X}_{t}=g_{0}(\cdots, \, \bf{X}_{t-1}, \, \bf{X}_{t+1}, \, \cdots ; \, \cdots, \, \bf{e}_{t-1}, \, \bf{e}_{t}, \, \bf{e}_{t+1}, \cdots; \\ \, \cdots, \, \mu(\bf{e}_{t-1}), \, \mu(\bf{e}_{t}), \, \mu(\bf{e}_{t+1}) \cdots , \, \boldsymbol{\gamma})\end{multline}
which lives on an appropriate Banach space. Explicitly, for $T$ very large, let $\mu(\bf{e}_{t+T}) = \mu(\bf{e})$,  $\bf{X}_{t+T} = \bar{\bf{X}}$ and 
$\bf{e}_{t-l}=\bar{e}$.  
Truncate the entire past, \\ so
$\bf{X}_{t-l} = \bar{\bf{X}}$ and $\bf{e}_{t-l}=\bar{e}$, $\forall l \geq 0$. 
The approximation simplifies to  
\begin{equation} \bf{X}_{t}=g_{1}(\bf{X}_{t+1}, \, \cdots , \,\bf{X}_{t+T-1}, \, \bf{e}_{t}; \, \boldsymbol{\mu}(\bf{e}_{t}) , \, \boldsymbol{\gamma})\end{equation}
This relationship can be carried forward, for example 
\begin{equation} \bf{X}_{t+1}=g_{1}(\bf{X}_{t}, \, \bf{X}_{t+2}, \, \cdots , \,\bf{X}_{t+T-1}, \, \bf{e}_{t}, \, \mu(\bf{e}_{t}) ; \, \boldsymbol{\gamma})\end{equation}
$$\vdots $$
\begin{equation} \bf{X}_{t+T-2}=g_{1}(\bf{X}_{t}, \,\bf{X}_{t+T-1}, \, \bf{e}_{t}; \, \mu(\bf{e}_{t}), \, \boldsymbol{\gamma})\end{equation}
Invert the last functional equation to remove terms in $\bf{X}_{t+T-1}$.
Applying this transformation successively produces an equality of the form 
\begin{equation} \bf{X}_{t}=g_{2}(\bf{X}_{t},\, \bf{e}_{t}, \, 
 ; \, \mu
(\bf{e}_{t}),\, \boldsymbol{\gamma})\end{equation}
The last line of calculation is to solve for $\bf{X}_{t}$ to yield an integral equation
\begin{equation} \bf{X}_{t}=g_{3}(\bf{e}_{t}, \, 
\mu(\bf{e}_{t}) ; \, \boldsymbol{\gamma})\end{equation}
For the second stage, sum $\bf{X}_{t}$ over a mesh in
$\epsilon$. In the limit where the distance between the points $\vert \epsilon \vert \rightarrow 0$, a candidate for a consistent approximation $\mu_{T} \rightarrow \mu$ emerges via the fundamental theorem of calculus.\footnote{To be fastidious the Lebesgue version would be required to cope with the potential for unbounded domains of integration. Nevertheless, the differential structure ensures that these can be approximated by a sequence of Riemann summands.} 
\par It is necessary to justify all the proceeding steps rigorously. The integration scheme works because (197) is a Bochner integral, thus it can be approached via Lebesgue measures (\cite{aliprantis2007infinite}). Poincare Duality ensures the abridged system yields solutions that are close to the true model (\cite{hatcher2002algebraic}). Specifically, the cohomology of the manifolds means that all the inversions will produce mappings that are accurate modulo measure zero removable singularities. Theorem 3 ensures that the distance is $O(\max\{\bar{\lambda}_{2}, \, \bar{\rho}\})^{T})$, as \\ required.
\par Multiplicity of solutions will converge towards one another and
can be \\ handled by taking a convex combination over the multitude \\
$(\mathbf{X}_{t} = \sum_{i}h_{i}\mathbf{X}_{i, \, t}: \sum_{i}h_{i}=1)$ to yield the desired approximation. Part (ii) follows swiftly from inspecting the definition in \cite{villani2021topics} and applying the same arguments as Theorem 4. The entire theorem is solemnized. 
\end{proof}
\begin{remark}In practice, the integration proposal may prove unwieldy, \\ particularly for larger models (\cite{judd2023numerical}). It may be preferable to use an alternative method, for example based on a (Schauder) basis of the function space. The most famous is the space of polynomials established by the celebrated \\ Stone-Weierstrass Theorem (see \cite{aliprantis2007infinite}); clearly $\epsilon \rightarrow 0$ corresponds to the number of terms in the approximation $N \rightarrow \infty$. There are alternative bases available, described in detail, with implementation, in \cite{kincaid2009numerical}. Machine learning methods, such as \cite{yang2025structural}, are also gaining popularity and should prove useful in this regard. Optimization of computational routines is beyond the scope of this study. 
\end{remark}
\begin{corollary}Suppose there are two approximate solutions $\hat{\mu}^{1}_{T}$ and $\hat{\mu}^{2}_{T}$, \\ as $T \rightarrow \infty$, $\vert \mu^{1}_{T}-\mu^{2}_{T}\vert \leq C \, \max\{\bar{\lambda}_{2}, \, \bar{\rho}\})^{2T})$. \end{corollary}
\begin{proof}Follows immediately from Theorem 9 and basic calculus relating minima to stationary points.
\end{proof}
\begin{remark}Suppose $f \in C^{\infty}$, the same arguments 
guarantee a convergence rate of $\max\{\bar{\lambda}_{2}, \, \bar{\rho}\})^{2nT})$ for some $n \geq 1$, depending on which even derivative is the first non-zero. This idea can be easily adapted to alternative expansions for models with less regularity.
\end{remark}
\subsubsection{Existence Testing}
\begin{theorem} For a standard DSGE model operating under the numerical scheme in Theorem 9, there exists a test with standard normal asymptotics and full asymptotic power to test the hypothesis that a Stochastic Equilibrium exists against the alternative that it does not as $\epsilon \rightarrow 0$, 
for large $T$. \end{theorem}
\begin{proof}
Theorem 3 informs us that down the prescribed limit output from the algorithm put forward in Theorem 9 is inevitably described by the relationship
\begin{equation} \hat{\bf{X}}_{t}(\epsilon, \, T)=\zeta\, \hat{\bf{X}}_{t-1}(\epsilon, \, T) + O(\epsilon, \,\max\{\bar{\lambda}_{2}, \, \bar{\rho}\})^{T})\end{equation}
This is (more than) sufficient to enlist Theorem 27.4 of \cite{billingsley1995probability}. This central limit theorem for strongly mixing processes aided by Theorem 7 imply that 
\begin{equation} \hat{\zeta} \sim \mathcal{N}(\max\{\bar{\lambda}_{2}, \, \bar{\rho}\})^{T}, \, O(\epsilon^{2}, \, \max\{\bar{\lambda}_{2}, \, \bar{\rho}\})^{T})/T^{2})\end{equation}
where $\hat{\zeta} $ is its OLS estimator. This settles the distribution claim. Finally, this leads to a standard hypothesis with $H_{0}: \zeta < 1$, which corresponds to the null hypothesis of the existence of stochastic equilibrium against \\ $H_{1}:\zeta \geq 1 $ the alternative that non-exists. The power of the test corresponds to the complement of the critical region $\Phi(z)$. The asymptotic power landmark is reached by observing that for some $c>0$
\begin{equation}
 \lim_{\epsilon \rightarrow 0, \, T \rightarrow \infty} \Phi(z)=\Phi\bigg(\frac{c\, T}{\max\{\bar{\lambda}_{2}, \, \bar{\rho}\})^{T}\, \epsilon}\bigg) \rightarrow  \Phi(\infty) \rightarrow 1
\end{equation}
\end{proof}
\begin{remark}
This test is efficient assuming the dominant eigenvalue is real. Otherwise, the model will miss the sinusoidal pattern associated with complex roots. In which case, one could start by regressing 
\begin{equation} \hat{\bf{X}}_{t}(\epsilon, \, T)=\zeta_{1}\, \hat{\bf{X}}_{t-1}(\epsilon, \, T) + \zeta_{2}\, \hat{\bf{X}}_{t-2}(\epsilon, \, T) +O(\epsilon, \,\max\{\bar{\lambda}_{2}, \, \bar{\rho}\})^{T})\end{equation}
where the pretest $H_{0}:\zeta_{2}=0$ versus $H_{1}:\zeta_{2}\neq 0$ tests the null of a real lead characteristic root against the alternative that it is complex, at no loss of asymptotic power. Note that $\zeta=\max\{\bar{\lambda}_{2}, \, \bar{\rho}\}$ could be estimated by indirect least squares by $\hat{\zeta}= \sqrt{\hat{\zeta}_{2}}$ in a consistent and asymptotically normal fashion. For details consult \cite{simon1994mathematics} (second-order difference equations) and \cite{greene2003econometric}(asymptotic theory). \end{remark}
It would be interesting to implement this kind of examination on real-world data. Although, in this case one would have to rely on standard asymptotics $O(1/\sqrt{T})$ and use the perturbation solutions laid forth as part of Theorem 8. 
\section{Conclusions}
This paper constitutes substantial progress in numerical approximation and \\ statistical analysis of DSGE. It is remarkable to uncover 
such precise \\ convergence outcomes across such a wide class of models.
The optimality of perturbations taken from around the stochastic equilibrium is surely a milestone in approximation theory. Moreover, the mathematically watertight algorithm to compute and test for stochastic equilibrium should prove highly significant for nonlinear simulation and estimation. The new lower bound on the asymptotic variance will surely provide crucial additional precision and power for the most basic tests of the benchmark New Keynesian framework. These constitute clear cut gains from rigorous stochastic equilibrium theory. I anticipate significant opportunities for mathematical refinement and computational investigation. 
\par Finally, I make two further theoretical advances. I derive sharp bounds on the peaks of impulse response functions, which decisively favor greater consideration of Taylor-style contracts in applied work. I also extend \\ commonality within the modern Keynesian framework to include the popular menu cost model. There are many new avenues to explore. Our understanding of inflation dynamics is progressing fast.
\bibliographystyle{plainnat} 
\bibliography{sncse.bib}
\appendixpage
\appendix
As presaged in the main text (Section 2.2), these Supplementary Materials follow a fourfold path. The first section contains an array of 
items missing from the main corpus. The final three sections are given over to Taylor contracts. The first two sections alight upon variants of its Phillips curve. The last stanza relates to the eigenvalue calculations cementing Propositions 7 and 8.
\section{Miscellaneous Items}
This inaugural section of the Supplementary Materials possesses three distinct parts. 
The first derives the menu cost demand aggregator from Section 4.1. The second gives additional results for the non-stochastic steady state of the Taylor pricing, pertaining to Section 5.2.2. The last subsection takes up solving a forward-looking toy model mentioned in Section 6 and resolving potential singularities.  
\subsection{Demand with Idiosyncratic Shocks}
This section tracks the presentation in SELCKE Appendix B. The consumption objective (utility from consumption) is
\begin{equation}
u_{t}=\bigg[\int_0^1 b_{t}(i)\, {c_{t}(i)}^{(\theta -1)/\theta} \,{d}i\bigg]^{\theta/(\theta-1)}\end{equation}
where $\theta > 1$ is required for a well-behaved problem. The household's problem is to maximize consumption utility subject to an expenditure constraint
\begin{equation}P_{t}\, C_{t}=\int_0^1 p_{t}(i)\, c_{t}(i)\,{d}i\end{equation}
For any two varieties $i$ and $i'$, this yields relative demand. 
\begin{equation}\frac{{c}_{t}(i)}{c_{t}(i')}=\bigg(\frac{b_{t}(i)}{b_{t}(i')}\bigg)^{\theta} \bigg(\frac{p_{t}(i)}{p_{t}(i')}\bigg)^{-\theta} \end{equation}
Now a little manipulation and then integration with respect to $i'$ yields 
\begin{equation} \int_0^1 p_{t}(i')\, c_{t}(i')\,{d}i'=\int_0^1 \bigg(\frac{b_{t}(i)}{b_{t}(i')}\bigg)^{-\theta}c_{t}(i)\, p_{t}(i)^{\theta}\, p_{t}^{1-\theta}(i')\,{d}i' \end{equation} 
Guessing and verifying yields the demand curve price level pair 
\begin{equation}
P_{t}=\bigg[\int_0^1 \bigg(\frac{b_{t}(i)}{b}\bigg)^{\theta}\, {p_{t}(i)}^{1-\theta}{d}i\bigg]^{1/(1-\theta)}\end{equation}
\begin{equation}c_{t}(i)= \bigg(\frac{b_{t}(i)}{b}\bigg)^{\theta}\, \bigg(\frac{p_{t}(i)}{P_{t}}\bigg)^{-\theta} \, C_{t}\end{equation}
substituting in the resource constraint (50) completes the task of deriving \\ expression (51). 
\qed
\subsection{Taylor Price Dispersion and Steady State}
The first action is to provide expressions for the non-stochastic steady state with general time preferences. The second delivers the proof of the order of magnitude for price dispersion in the most pressing neighborhood.
\subsubsection{Non-Stochastic Steady State with Time Discounting}
This item marches in lockstep with the text, save for the minor inconvenience of the reemergence of the singularities from Section 5.2.1. For exactitude, 
\begin{equation}
X=
 \begin{cases} 
 X^{NSS}  &\text{if $\pi \neq 0, \, \pi_{1}, \, \pi_{2}$} \\ 
 X^{1}  &\text{if $\pi =\pi_{1}$} \\ 
  X^{2}  &\text{if $\pi =\pi_{2}$} \\
  X^{ZINSS} &\text{if $\pi =0$}
 \end{cases}
\end{equation}
and as before $X=\{MC, \, \Delta, \, Y, \, L, \, \Pi \}$. 
The following correspond with the special cases (107), (109) and (111). (108) and (110) are of course unchanged.
\begin{multline} MC^{NSS}= \frac{\theta-1}{M^{1/(\theta-1)}\, \theta}\bigg(\frac{\beta (1+\pi)^{\theta} -1}{\beta^{M}(1+\pi)^{M\theta} -1}\bigg)\bigg(\frac{\beta^{M}(1+\pi)^{M(\theta-1)} -1}{\beta (1+\pi)^{\theta-1} -1}\bigg)\times \\ \bigg(\frac{(1+\pi)^{M(\theta-1)} -1}{(1+\pi)^{\theta-1} -1}\bigg)^{1/(\theta-1)}\end{multline}
\begin{equation} MC^{1}= \frac{\theta-1}{M^{\theta/(\theta-1)}\, \theta}\bigg(\frac{\beta\, ^{M}(1+\pi)^{M(\theta-1)} -1}{\beta \, (1+\pi)^{\theta-1} -1}\bigg)\bigg(\frac{(1+\pi)^{M(\theta-1)} -1}{(1+\pi)^{\theta-1} -1}\bigg)^{1/(\theta-1)}\end{equation}
\begin{multline} MC^{2}=M^{(\theta-2)/(\theta-1)}\bigg( \frac{\theta-1}{\theta}\bigg)\bigg(\frac{\beta \, (1+\pi)^{\theta} -1}{\beta^{M}\, (1+\pi)^{M\theta} -1}\bigg) \times \\ \bigg(\frac{(1+\pi)^{M(\theta-1)} -1}{(1+\pi)^{\theta-1} -1}\bigg)^{1/(\theta-1)}\end{multline}
\begin{multline}Y^{NSS}= 
\frac{A}{M^{1/(\theta-1)}}\bigg(\frac{\theta-1}{\theta}\bigg)^{1/(\eta+1)}\bigg(\frac{\beta \, (1+\pi)^{\theta} -1}{\beta^{M}\, (1+\pi)^{M\theta} -1}\bigg)\times \\ \bigg(\frac{\beta^{M}\, (1+\pi)^{M(\theta-1)} -1}{\beta \, (1+\pi)^{\theta-1} -1}\bigg) \bigg(\frac{(1+\pi)^{M(\theta-1)} -1}{(1+\pi)^{\theta-1} -1}\bigg)^{1/(\theta-1)}
\end{multline}
\begin{multline}Y^{1}= 
\frac{A}{M^{\theta/(\theta-1)}}\bigg(\frac{\theta-1}{\theta}\bigg)^{1/(\eta+1)} \bigg(\frac{\beta^{M}\, (1+\pi)^{M(\theta-1)} -1}{\beta \, (1+\pi)^{\theta-1} -1}\bigg) \times \\ \bigg(\frac{(1+\pi)^{M(\theta-1)} -1}{(1+\pi)^{\theta-1} -1}\bigg)^{1/(\theta-1)}
\end{multline}
\begin{multline}Y^{2}= 
A \, M^{(\theta-2)/(\theta-1)}\bigg(\frac{\theta-1}{\theta}\bigg)^{1/(\eta+1)}\bigg(\frac{\beta \, (1+\pi)^{\theta} -1}{\beta^{M}\, (1+\pi)^{M\theta} -1}\bigg)\times \\
\bigg(\frac{(1+\pi)^{M(\theta-1)} -1}{(1+\pi)^{\theta-1} -1}\bigg)^{1/(\theta-1)}
\end{multline}
\begin{multline} L^{NSS}=\bigg(\frac{\theta-1}{\theta}\bigg)^{1/(\eta+1)}\bigg(\frac{\beta \, (1+\pi)^{\theta} -1}{\beta^{M}\, (1+\pi)^{M\theta} -1}\bigg) \bigg(\frac{\beta^{M}\, (1+\pi)^{M(\theta-1)} -1}{\beta \, (1+\pi)^{\theta-1} -1}\bigg) \times \\ 
\bigg(\frac{(1+\pi)^{M\theta} -1}{(1+\pi)^{\theta} -1}\bigg)
\bigg(\frac{(1+\pi)^{\theta-1} -1}{(1+\pi)^{M(\theta-1)} -1}\bigg)
\end{multline}
\begin{multline} L^{1}=\frac{1}{M}
\bigg(\frac{\theta-1}{\theta}\bigg)^{1/(\eta+1)}\bigg(\frac{\beta^{M}(1+\pi)^{M(\theta-1)} -1}{\beta (1+\pi)^{\theta-1} -1}\bigg) 
\bigg(\frac{(1+\pi)^{M\theta} -1}{(1+\pi)^{\theta} -1}\bigg)\times \\
\bigg(\frac{(1+\pi)^{\theta-1} -1}{(1+\pi)^{M(\theta-1)} -1}\bigg)
\end{multline}
\begin{multline} L^{2}=M\bigg(\frac{\theta-1}{\theta}\bigg)^{1/(\eta+1)}\bigg(\frac{\beta \, (1+\pi)^{\theta} -1}{\beta^{M}\, (1+\pi)^{M\theta} -1}\bigg) 
\bigg(\frac{(1+\pi)^{M\theta} -1}{(1+\pi)^{\theta} -1}\bigg) \times \\ 
\bigg(\frac{(1+\pi)^{\theta-1} -1}{(1+\pi)^{M(\theta-1)} -1}\bigg)
\end{multline}
\begin{multline} \Pi^{NSS}= \Bigg\{ 1 - \frac{A}{M^{2/(\theta-1)}}\bigg(\frac{\theta-1}{\theta}\bigg)^{(\eta+2)/(\eta+1)}\, \bigg(\frac{\beta \, (1+\pi)^{\theta} -1}{\beta^{M}\, (1+\pi)^{M\theta} -1}\bigg)\times \\ \bigg(\frac{\beta^{M}\, (1+\pi)^{M(\theta-1)} -1}{\beta \, (1+\pi)^{\theta-1} -1}\bigg)  \bigg(\frac{(1+\pi)^{M(\theta-1)} -1}{(1+\pi)^{\theta-1} -1}\bigg)^{1/(\theta-1)}\Bigg\}\bigg(\frac{\beta (1+\pi)^{\theta} -1}{\beta^{M}(1+\pi)^{M\theta} -1}\bigg)\times \\ \bigg(\frac{\beta^{M}(1+\pi)^{M(\theta-1)} -1}{\beta (1+\pi)^{\theta-1} -1}\bigg) \bigg(\frac{(1+\pi)^{M(\theta-1)} -1}{(1+\pi)^{\theta-1} -1}\bigg)^{1/(\theta-1)}\end{multline}
$\Pi^{1}$ and $\Pi^{2}$, which take after the 
common pattern, are suppressed for brevity.
\subsubsection{Second-Order Price Dispersion at Arbitrary Horizons}
This part carries forward the analysis of Section 5.2.2, by furnishing a full proof of Proposition 7, covering the more cumbersome instances of arbitrary contract length. 
\newline 
\newline 
\textbf{Proof of Proposition 7} \textbf{(Taylor} $\mathbf{M}$\textbf{)}
\newline
\newline 
The proof is a careful calculation.
\begin{proof}  
The stated result is a consequence of the following fact 
derived from (108)
\begin{multline} \frac{\mathrm{d}^{2}\Delta^{NSS}}{\mathrm{d}^{2}\pi}\bigg\rvert_{\pi=0}= \\ \bigg( \frac{
M \, \theta\, (\theta-1)\, \pi\, \big\{ 2\, (1+\pi) + (M\, \theta-1)\, (\theta-1)\, \pi \big\}-(\theta-2)\big\{ (1+\pi)^{M\theta}-1\big\} }{\big\{(1+\pi)^{M(\theta-1)} -1\big\}^{\theta/(\theta-1)}\, \pi^{(2\theta-3)/(\theta-1)}}\bigg) \\ -\frac{2 \, M \, \theta}{\theta -1}\bigg( \frac{(1+\pi)^{M(\theta-1)-1}\big\{ \theta \, (\theta-1)\, M \, (1+\pi)^{M(\theta-1)}\, \pi -\big\{ (1+\pi)^{M\theta}-1\big\}\big\}}{\big\{(1+\pi)^{M(\theta-1)} -1\big\}^{(2\, \theta-1)/(\theta-1)}\, \pi^{(\theta-2)/(\theta-1)}}\bigg) \\ 
+ \theta \, M \bigg(\frac{
\pi^{1/(\theta-1)}\, (1+\pi)^{M(\theta-1)-2}\, \big\{(1+\pi)^{M\theta} -1\big\}}{\big\{(1+\pi)^{M(\theta-1)} -1\big\}^{(3\theta-2)/(\theta-1)}}\bigg)\times \\ \bigg[ (2\, \theta -1)\, M -\big\{M\, (\theta-1)-1\big\}\big\{(1+\pi)^{M(\theta-1)} -1\big\}\bigg]\\ = 
(M-1)\, \frac{\theta}{6}\, \bigg\{ (2\, M-1)\, \theta -3\bigg\}-\frac{2\, (M-1)^{2}}{M^{2}}\, \theta\, (\theta-1) + \\ \frac{1}{M^{2}}\bigg[2\, (M-1)^{2}\, (\theta-1)^{2}-M^{2}\, (M-1)\, \frac{(\theta-1)}{6}\, \bigg\{ (2\, M-1)\, (\theta-1) -3\bigg\}\bigg]
\\ =\frac{(M-1)}{6M^{2}}\, \bigg( 2M^{2}\, (M+1)\, \theta + (M-2)\, \big\{ 2\, (M-1)^{2} + 7\, (M-2) +6 \big\}\, (\theta-1)\bigg)>0\end{multline}
\end{proof}
\begin{remark} The proceeding succession of identities 
was irreplacable to my \\ derivation of (224). 
\begin{multline*} \frac{\mathrm{d}^{2}}{\mathrm{d}x^{2}}\bigg( \frac{f(x)}{g(x)}\bigg)=\frac{1}{g(x)}\, \frac{\mathrm{d}^{2}(f(x))}{\mathrm{d}x^{2}} -\frac{2}{g^{2}(x)}\, \frac{\mathrm{d}(f(x))}{\mathrm{d}x}\, \frac{\mathrm{d}(g(x))}{\mathrm{d}x} + \\ \frac{f(x)}{g^{3}(x)}\, \bigg[ 2\, \bigg( \frac{\mathrm{d}(g(x))}{\mathrm{d}x}\bigg)^{2} -g(x)\, \frac{\mathrm{d}^{2}(g(x))}{\mathrm{d}x^{2}} \bigg]\end{multline*}
$$ \sum_{k=1}^{n}k^{2}=\frac{n}{6}\, (n+1)\, (2n+1)$$
which was used to sum second derivative expressions. Finally, although I used 
(108) in the first part of 
(224), I chose as my starting point series expansions as in the two period case (112) rather than attempting to apply L'H{\^o}pital's rule.
\end{remark}
\subsection{Forward Looking Alternative (Rotemberg Pricing)}
The Rotemberg model of sticky nominal adjustment is in the process of falling out of favor, in the light of the disproof of equivalence with Calvo in SELCKE.
Nevertheless, it remains an interesting theoretical device. I mention it here as it is apposite in the arenas of Theorem 3 (ii) and Theorem 6 (iii).
I execute the convenient closed-form analysis suggested by the latter, avoiding the rigmarole of calculating eigenvalues or performing complicated inter-temporal \\ maximization. 
Lastly, the analytical significance of a possible singularity is explained. 
\par The starting point is that in every period each firm faces a convex cost of changing prices 
\begin{equation}
C^{a}_{t}(i)=\frac{c_{p}}{2}\, \bigg(\frac{p_{t}(i)}{p_{t-1}(i)}-1\bigg)^{2}
\end{equation}
with $c_{p}$ the magnitude of the cost of changing a price. 
This implies the first-order condition 
\begin{multline}
y_{t}(i)- \theta \, p^{*}_{t}(i)\, \bigg(\frac{p^{*}_{t}(i)}{P_{t}}\bigg)^{-(\theta +1)}\, \frac{Y_{t}}{P_{t}} + \theta \, \bigg(\frac{p^{*}_{t}(i)}{P_{t}}\bigg)^{-(\theta +1)} MC_{t}\, Y_{t} - \\ c_{p}\, \bigg(\frac{p^{*}_{t}(i)}{p_{t-1}(i)}-1\bigg)\frac{P_{t}}{p^{*}_{t-1}(i)}Y_{t} + c_{p} \, \mathbb{E}_{t}\, Q_{t, \, t+1}\, \bigg(\frac{p^{*}_{t+1}(i)}{p^{*}_{t}(i)}-1\bigg)\, Y_{t+1}\, \frac{p^{*}_{t+1}(i)}{(p^{*}_{t}(i))^{2}}\, P_{t+1}=0
\end{multline}
Firms select the same price.  
This means the dynamics are considerably simpler with 
\begin{equation}\frac{p_{t}(i)}{p_{t-1}(i)}=1+ \pi_{t}\end{equation}
\begin{equation}\Delta =1 \end{equation}
Hence, the first-order conditions simplify to 
\begin{equation}
(1- \theta)  + \theta  \, MC_{t} 
-  c_{p}\, \pi_{t} \, (1+\pi_{t}) +  c_{p}\,  \mathbb{E}_{t}\, Q_{t, \, t+1}\, \pi_{t+1} \, (1+\pi_{t+1})\, \frac{Y_{t+1}}{Y_{t}}=0
\end{equation}
Finally, the resource constraint reflecting a wedge between production and \\ consumption, caused by the cost of price 
changes, replaces (13)
\begin{equation}
Y_{t}=C_{t} + \frac{1}{2}c_{p}\, \pi^{2}_{t}\, Y_{t}
\end{equation}
Therefore, around ZINSS $\hat{c}_{t}=\hat{y}^{e}_{t}$.\, 
Here, the dynamic linearzied Phillips curve solves out to 
\begin{equation}\pi_{t}=\tilde{\omega}\, \hat{y}^{e}_{t} + \beta \, \mathbb{E}_{t}\, \pi_{t+1} \end{equation}
where 
\begin{equation}
\tilde{\omega}=\frac{(1 + \eta)\,(\theta -1)}{c_{p}}
\end{equation}
The basic demand relation (38) is in place. 
From Theorem 9, it can be deduced that 
\begin{equation} \mathbb{E}_{t}\, \pi_{t+1}=\rho \, \pi_{t}=k_{\pi}\,\hat{\psi}_{t}\end{equation}
\begin{equation} \mathbb{E}_{t}\, \hat{y}^{e}_{t+1}=\rho \, \hat{y}^{e}_{t}=k_{y}\, \hat{\psi}_{t}\end{equation}
From (38), (233) and (234) result the impulse response coefficients 
\begin{equation} k_{\pi}= \frac{(1-\rho)\, \tilde{\omega}}{(1+\beta \, a_{y})(1-\beta \, \rho)-(\beta \, a_{\pi}-\rho)\, \tilde{\omega} }\end{equation}
\begin{equation} k_{y}= \frac{(1-\beta \, \rho)\,  (1-\rho)}{(1+\beta \, a_{y})(1-\beta \, \rho)-(\beta \, a_{\pi}-\rho)\, \tilde{\omega} }\end{equation}
for $ (1+\beta \, a_{y})(1-\beta \, \rho)-(\beta \, a_{\pi}-\rho)\, \tilde{\omega}\neq 0$. Otherwise, there is a singularity because the aggregate demand equation is inconsistent. This means that the limiting stochastic equilibrium where $\vert \varepsilon\vert \rightarrow 0$ degenerates. However, Theorem 8 guarantees this problem will not arise for any particular shock size $\vert \bar{\varepsilon} \vert >0$. Therefore, it will be possible to resolve this singularity with higher-order perturbations, whenever $a_{\pi} >1$, as stipulated by Theorem 5 in SELCKE. 
\section{Phillips Cure for Longer Contracts}
This second section of the Supplementary Materials is devoted to the extraction of a working Phillips curve, 
admitting less frequent repricing. It can be seen as an extension of the short contracts model, the subject of Proposition 9 and a development upon the non-linear context of Proposition 5. It is integral to the proof of Proposition 8 in Section D. The notational convention in place is that $X_{t\pm (M-k)}$ is nonzero only if $k \leq M-1$.
\begin{proposition}
Consider an economy with Taylor contracts
reset every $M \geq 2$ periods, then around ZINSS 
recursive equilibrium takes the form \\ $\mathbb{E}_{t}\, \mathbf{\hat{Z}}_{t+1}=\mathbf{B}\, 
\mathbf{\hat{Z}}_{t} + \mathbf{\Phi}\, \mathbf{\hat{U}}_{t}$, $\mu$ a.s. where
\begin{enumerate}[(i)]
\item $\mathbf{Z}_{t}=(\pi_{t}, \, \pi_{t-1}, \,  \cdots \, \pi_{t-2(M-1)}, \,  Y^{e}_{t})'$ 
\item  $\mathbf{Z}_{t}=(\mathbf{Z}^{J}_{t}, \, \mathbf{Z}^{S}_{t} )'$, such that $\mathbf{Z}^{J}_{t}=(\pi_{t}, \, \pi_{t-1}, \,  \cdots \, \pi_{t-(M-2)}, \,  Y^{e}_{t})'$ and \\ $\mathbf{Z}^{S}_{t}=(\pi_{t-(M-1)},  \,  \cdots \, \pi_{t-2(M-1)})'$
\item $\mathbf{U}_{t}=(\psi_{t}, \, \psi_{t-1}, \, \cdots \, \psi_{t-2(M-1)})'$
\item $\mathbf{\Phi}=\mathbf{\Phi}(\boldsymbol{\gamma})$ and $\mathbf{B}=\mathbf{B}(\boldsymbol{\gamma})$ announced herein.
\end{enumerate}
\end{proposition}
\begin{proof}
Linearizing (79) yields the subsequent association
corresponding to (124)
\begin{multline} \hat{p}^{*}_{t}= \frac{(1-\beta)}{1-\beta^{M}}\, \hat{P}_{t} + \frac{\beta\, (1-\beta)}{1-\beta^{M}}\, \mathbb{E}_{t}\, \hat{P}_{t+1} + \cdots + \frac{\beta^{M-1}\, (1-\beta)}{1-\beta^{M}}\, \mathbb{E}_{t}\, \hat{P}_{t+M-1}  + \\ \frac{(1-\beta)}{1-\beta^{M}}\, \hat{mc}_{t} + \frac{\beta\, (1-\beta)}{1-\beta^{M}}\, \mathbb{E}_{t}\, \hat{mc}_{t+1} + \cdots + \frac{\beta^{M-1}\, (1-\beta)}{1-\beta^{M}}\, \mathbb{E}_{t}\, \hat{mc}_{t+M-1}\end{multline} 
where I have evaluated the sum of the geometric progression $1 + \beta + \cdots + \beta^{M-1}$.\footnote{Each coefficient approaches the limit $1/M$ as $\beta \rightarrow 1$ thanks to L'H{\^o}pital's rule, as in SELCKE Section G.1.}
As expected, the price level 
(126) generalizes to 
\begin{equation} P_{t}=\frac{1}{M}\, (p^{*}_{t-(M-1)})^{1-\theta} + \cdots + \frac{1}{M}\, (p^{*}_{t})^{1-\theta}\end{equation}
Taking a lag and then subtracting reveals that
\begin{equation} \pi_{t}=\hat{P}_{t}-\hat{P}_{t-1}=\frac{1}{M}\bigg[\hat{p}^{*}_{t-(M-1)}-\hat{p}^{*}_{t-M}+ \cdots + \hat{p}^{*}_{t}-\hat{p}^{*}_{t-1}\bigg]=\frac{1}{M}\bigg(\hat{p}^{*}_{t}-\hat{p}^{*}_{t-M}\bigg)\end{equation}

\par That $M^{th}$ lag reads 
\begin{multline} \hat{p}^{*}_{t-M}= \frac{(1-\beta)}{1-\beta^{M}}\, \hat{P}_{t-M} + \frac{\beta\, (1-\beta)}{1-\beta^{M}}\, \mathbb{E}_{t-M}\, \hat{P}_{t-(M-1)} + \cdots + \\ \frac{\beta^{M-1}\, (1-\beta)}{1-\beta^{M}}\, \mathbb{E}_{t-M}\, \hat{P}_{t-1}  +  \frac{(1-\beta)}{1-\beta^{M}}\, \hat{mc}_{t-M} + \frac{\beta\, (1-\beta)}{1-\beta^{M}}\, \mathbb{E}_{t-M}\, \hat{mc}_{t-(M-1)} + \\ \cdots + \frac{\beta^{M-1}\, (1-\beta)}{1-\beta^{M}}\, \mathbb{E}_{t-M}\, \hat{mc}_{t-1}\end{multline} 
\par Combining (238) and (240) offers up
\begin{multline} \pi_{t}=\frac{1-\beta}{1-\beta^{M}}\, \frac{1}{M-1}\bigg(\pi_{t-(M-1)} + (1+\beta) \, \pi_{t-(M-2)} + \cdots + (1 + \cdots + \beta^{M-2})\, \pi_{t-1} 
\\ + 
(\beta + \cdots + \beta^{M-1})\, \mathbb{E}_{t}\, \pi_{t+1} +   \cdots + (\beta^{M-2} +\beta^{M-1}) \, \mathbb{E}_{t}\, \pi_{t+M-2} + 
\beta^{M-1}\, \mathbb{E}_{t}\, \pi_{t+M-1} 
\\ + (\hat{mc}_{t-(M-1)} - \hat{mc}_{t-M})+ 
(1+\beta)\, (\hat{mc}_{t-(M-2)} - \hat{mc}_{t-(M-1)}) +    \cdots + \\ (1 + \cdots + \beta^{M-2})\, (\hat{mc}_{t-1} - \hat{mc}_{t-2})  
+(1 + \beta + \cdots + \beta^{M-1})\, (\hat{mc}_{t} - \hat{mc}_{t-1}) +  \\ (\beta + \cdots + \beta^{M-1})\, \mathbb{E}_{t}\, (\hat{mc}_{t+1} - \hat{mc}_{t})+  \cdots + \\  (\beta^{M-2} +\beta^{M-1}) \, \mathbb{E}_{t}\, (\hat{mc}_{t+M-2} - \hat{mc}_{t+M-3}) 
 + \beta^{M-1}\, \mathbb{E}_{t}\, (\hat{mc}_{t+M-1} - \hat{mc}_{t+M-2}) + \\ 
\hat{v}^{1, \, M}_{t} \bigg)\end{multline}
with 
\begin{multline} \hat{v}^{1, \, M}_{t}= -\bigg(\pi_{t-(M-1)}- \mathbb{E}_{t-M}\,\pi_{t-(M-1)} \bigg)- \\ (1 +\beta) \bigg(\pi_{t-(M-2)}- \mathbb{E}_{t-M}\, \pi_{t-(M-2)} \bigg) 
- \cdots -   \frac{(1-\beta^{M-1})}{(1-\beta)}\bigg( \pi_{t-1}-\mathbb{E}_{t-M}\, \pi_{t-1}\bigg) - \\ \frac{(1-\beta^{M})}{(1-\beta)}\bigg( \pi_{t}-\mathbb{E}_{t-M}\, \pi_{t}\bigg)
+ \beta \bigg(\hat{mc}_{t-(M-1)} - \mathbb{E}_{t-M}\, \hat{mc}_{t-(M-1)}\bigg) + \\ \beta^{2} \bigg(\hat{mc}_{t-(M-2)} - \mathbb{E}_{t-M}\, \hat{mc}_{t-(M-2)}\bigg) + \cdots + \\ 
\beta^{M-1}\bigg(\hat{mc}_{t-1} - \mathbb{E}_{t-M}\, \hat{mc}_{t-1}\bigg)\end{multline} 
where I have taken the opportunity to pursue notational compactness. 
The next pocket of the derivation mimics (134)-(136), from the 
two-period case
\begin{equation} \hat{\Delta}_{t}= \frac{\theta}{M}\bigg(\hat{p}^{*}_{t}-\hat{P}_{t}\bigg) +  \frac{\theta}{M}\bigg(\hat{p}^{*}_{t-1}-\hat{P}_{t}\bigg) + \cdots +  \frac{\theta}{M}\bigg(\hat{p}^{*}_{t-(M-1)}-\hat{P}_{t}\bigg)\end{equation}
\begin{equation} \hat{\Delta}_{t-1}= \frac{\theta}{M}\bigg(\hat{p}^{*}_{t-1}-\hat{P}_{t-1}\bigg) + \frac{\theta}{M}\bigg(\hat{p}^{*}_{t-2}-\hat{P}_{t-1}\bigg) + \cdots + \frac{\theta}{M}\bigg(\hat{p}^{*}_{t-M}-\hat{P}_{t}\bigg)\end{equation}
Terms cancel when I take (244) from (243), so as before 
\begin{equation} \hat{\Delta}_{t}=\hat{\Delta}_{t-1}\end{equation}
The multiple output gaps form, extending 
(137), is expressed as 
\begin{multline} \pi_{t}= 
\prescript{\circ}{}{b}^{\pi, \, M}_{M}\, \pi_{t-M}+ \cdots + \prescript{\circ}{}b^{\pi, \, M}_{1}\, \pi_{t-1} + \prescript{\circ}{}b^{\pi, \, M}_{-1}\, \mathbb{E}_{t}\, \pi_{t+1} + \cdots + \prescript{\circ}{}b^{\pi, \, M}_{-(M-1)}\, \mathbb{E}_{t}\, \pi_{t+M-1} \\  + \prescript{\circ}{}b^{y, \, M}_{M}\, \hat{y}^{e}_{t-M} + \cdots + \prescript{\circ}{}b^{y, \, M}_{1}\, \hat{y}^{e}_{t-1} + \prescript{\circ}{}b^{y, \, M}_{0}\, \hat{y}^{e}_{t} + \\ \prescript{\circ}{}b^{y, \, M}_{-1}\, \mathbb{E}_{t}\, \hat{y}^{e}_{t+1} + \cdots + \prescript{\circ}{}b^{y, \, M}_{-(M-2)}\, \mathbb{E}_{t}\, \hat{y}^{e}_{t+M-2} + \hat{v}^{2, \, M}_{t}\end{multline}
Each coefficient $\prescript{\circ}{}b^{j}_{i}=\prescript{\circ}{}{\tilde{b}}^{j}_{i}/\prescript{\circ}{}b^{M}$, such that 
\begin{equation} \prescript{\circ}{}b^{M} =\frac{(1-\beta^{M})}{1-\beta}\, (M + \eta)- \frac{\beta^{2}\, (1-\beta^{M-1})}{1-\beta}\, (1+\eta)\, a_{\pi}\end{equation}
At the start there is a cost channel but no expected inflation effect. Thus,
\begin{equation} \prescript{\circ}{}{\tilde{b}}^{\pi, \, M}_{M}= 
\beta\, (1+\eta)\, a_{\pi}\end{equation}
subsequently, both are present so 
\begin{equation}  \prescript{\circ}{}{\tilde{b}}^{\pi, \, M}_{M-1}= \beta\, (1 + \beta) \, (1+\eta)\, a_{\pi} -\eta
\end{equation}
\begin{equation}  \prescript{\circ}{}{\tilde{b}}^{\pi, \, M}_{M-2}= \beta\, (1 + \beta + \beta^{2}) \, (1+\eta)\, a_{\pi} -(1 + \beta)\, \eta
\end{equation}
$$ \vdots $$
\begin{multline}  \prescript{\circ}{}{\tilde{b}}^{\pi, \, M}_{1}= 
\beta(1 + \cdots + \beta^{M-1}) \, (1+\eta)\, a_{\pi} - (1 + \cdots + \beta^{M-2})\, \eta = \\ \beta \, \bigg(\frac{1-\beta^{M}}{1-\beta} \bigg)(1+\eta)\, a_{\pi} - \bigg(\frac{1-\beta^{M-1}}{1-\beta}\bigg) \eta
\end{multline}
\begin{multline}  \prescript{\circ}{}{\tilde{b}}^{\pi, \, M}_{-1}= 
(\beta^{M} + \cdots + \beta^{3})\, (1+\eta)\, a_{\pi} -(\beta^{M-1} + \cdots + \beta)\, \eta = \\ \beta \bigg[ \beta^{2}\bigg(\frac{1-\beta^{M-2}}{1-\beta}\bigg)(1+\eta)\, a_{\pi} - \bigg(\frac{1-\beta^{M-1}}{1-\beta}\bigg) \eta \bigg]
\end{multline}
\begin{multline}  \prescript{\circ}{}{\tilde{b}}^{\pi, \, M}_{-2}= (\beta^{M} + \cdots + \beta^{4}) \, (1+\eta)\, a_{\pi} -(\beta^{M-1} + \cdots + \beta^{2})\, \eta = \\ 
\beta^{2} \bigg[ \beta^{2}\bigg(\frac{1-\beta^{M-3}}{1-\beta}\bigg)(1+\eta)\, a_{\pi} - \bigg(\frac{1-\beta^{M-2}}{1-\beta}\bigg) \eta \bigg]
\end{multline}
$$\vdots$$
\begin{equation} \prescript{\circ}{}{\tilde{b}}^{\pi, \, M}_{-(M-2)}=(1+\eta)\, \beta^{M}\, a_{\pi}-(\beta^{M-1} +\beta^{M-2} )\, \eta\end{equation}
whilst at the end, there is only the expected inflation effect but no cost channel so 
\begin{equation} \prescript{\circ}{}{\tilde{b}}^{\pi, \, M}_{-(M-1)}=-\beta^{M-1}\eta\end{equation} 
The pattern for the output gap coefficients is a little simpler 
\begin{equation} \prescript{\circ}{}{\tilde{b}}^{y, \, M}_{M}=
\beta\, (1+\eta)\, a_{y}\end{equation}
\begin{equation} \prescript{\circ}{}{\tilde{b}}^{y, \, M}_{M-1}=
\beta\, (1 + \beta)\, (1+\eta)\, a_{y}\end{equation}
$$\vdots$$
\begin{equation}\prescript{\circ}{}{\tilde{b}}^{y, \, M}_{1}=(\beta^{M} + \cdots + \beta) \, (1+\eta)\, a_{y}=\beta \bigg(\frac{1-\beta^{M}}{1-\beta}\bigg) (1+\eta)\, a_{y} \end{equation}
\begin{equation}\prescript{\circ}{}{\tilde{b}}^{y, \, M}_{0}=(\beta^{M} + \cdots + \beta^{2}) \, (1+\eta)\, a_{y}=\beta^{2}\bigg(\frac{1-\beta^{M-1}}{1-\beta}\bigg) (1+\eta)\, a_{y} \end{equation}
\begin{equation} \prescript{\circ}{}{\tilde{b}}^{y, \, M}_{-1}=\beta^{3}\bigg(\frac{1-\beta^{M-2}}{1-\beta}\bigg) (1+\eta)\, a_{y}\end{equation}
$$\vdots$$
\begin{equation} \prescript{\circ}{}{\tilde{b}}^{y, \, M}_{-(M-2)}=\beta^{M} \, (1+\eta)\, a_{y}\end{equation}
The error adjustment is 
\begin{multline} \hat{v}^{2, \, M}_{t}=
\frac{(1+\eta)}{\prescript{\circ}{}b^{M}}\Bigg[\hat{y}^{e}_{t-(M-1)} - \mathbb{E}_{t-M}\,\hat{y}^{e}_{t-(M-1)}+ 
\\ 
(1 +\beta) \, (\hat{y}^{e}_{t-(M-2)} - \mathbb{E}_{t-(M-1)}\,\hat{y}^{e}_{t-(M-2)}) + \\ \cdots + 
\bigg(\frac{1-\beta^{M-1}}{1-\beta}\bigg)(\hat{y}^{e}_{t-1} - \mathbb{E}_{t-2}\,\hat{y}^{e}_{t-1}) + \bigg(\frac{1-\beta^{M}}{1-\beta}\bigg)(\hat{y}^{e}_{t} - \mathbb{E}_{t-1}\,\hat{y}^{e}_{t}) + \\ \pi_{t-(M-1)} - \mathbb{E}_{t-M}\,\pi_{t-(M-1)}
+ (1 +\beta) (\pi_{t-(M-2)} - \mathbb{E}_{t-(M-1)}\,\pi_{t-(M-2)}) + \cdots + \\
\bigg(\frac{1-\beta^{M-1}}{1-\beta}\bigg)(\pi_{t-1} - \mathbb{E}_{t-2}\,\pi_{t-1}) + \bigg(\frac{1-\beta^{M}}{1-\beta}\bigg)(\pi_{t} - \mathbb{E}_{t-1}\,\pi_{t}) 
-\\ 
\prescript{\circ}{}{\tilde{b}}^{\psi, \, M}_{0}\, \hat{\psi}_{t} -
\bigg(\frac{1-\beta^{M}}{1-\beta}\bigg)(1-\rho)\, \hat{\psi}_{t-1} - 
\bigg(\frac{1-\beta^{M-1}}{1-\beta}\bigg)(1-\rho)\, \hat{\psi}_{t-2} - \cdots - \\ (1 +\beta)\, (1-\rho)\, \hat{\psi}_{t-(M-1)} -
(1-\rho)\, \hat{\psi}_{t-M}
\Bigg] + \frac{1}{\prescript{\circ}{}b^{M}}
\, \hat{v}^{1, \, M}_{t}\end{multline}
where 
\begin{multline}
 \prescript{\circ}{}{\tilde{b}}^{\psi, \, M}_{0}=(\beta + \cdots + \beta^{M-1})\, (1-\rho) + (\beta^{2} + \cdots + \beta^{M-1})\, \rho\, (1-\rho) + \cdots + \\  (\beta^{M-2} +\beta^{M-1})\, \rho^{M-3}\, (1-\rho) + \beta^{M-1}\, \rho^{M-2}\, (1-\rho) \equiv \beta\, (1-\rho)\, \bigg( \frac{1-\beta^{M-1}}{1-\beta} + \\ \beta \, \rho \, \frac{(1-\beta^{M-2})}{(1-\beta)} + \beta^{2} \, \rho^{2} \, \frac{(1-\beta^{M-3})}{(1-\beta)} + \cdots + \beta^{M-3}\, \rho^{M-3}\, \frac{(1-\beta^{2})}{1-\beta} + \beta^{M-2}\, \rho^{M-2}\bigg) \\ \equiv  
 \frac{\beta\, (1-\rho)}{(1-\beta)} \bigg( \frac{1-\beta^{M-1}\, \rho^{M-1}}{1-\beta \rho} - \beta^{M-1}\, \frac{(1-\rho^{M-1})}{(1-\rho)}\bigg)
\end{multline}
\par As with (137), all that remains is to repetitively substitute in the aggregate demand system 
\begin{multline} \pi_{t}= 
b^{\pi, \, M}_{M}\, \pi_{t-M}+ \cdots + b^{\pi, \, M}_{1}\, \pi_{t-1} + b^{y, \, M}\, \hat{y}^{e}_{t} + \\ b^{\pi, \, M}_{-1}\, \mathbb{E}_{t}\, \pi_{t+1} + \cdots + b^{\pi, \, M}_{-(M-1)}\, \mathbb{E}_{t}\, \pi_{t+M-1}  +\hat{u}^{M}_{t} 
\end{multline}
As usual, 
$b^{j, \, M}_{i}=\tilde{b}^{j, \, M}_{i}/b^{M}$. I will make repeated use of 
The \\ inflation coefficients are relatively straightforward.  
The denominator equals the previous expression (247)
plus two additional components. There is a cost \\ channel entering through the aggregate demand equation used to eliminate $\mathbb{E}_{t}\, \hat{y}^{e}_{t+1}$ and an expectations channel arising from substituting out $\hat{y}^{e}_{t-1}$, \\ specifically  
\begin{multline} 
b^{M} = \prescript{\circ}{}b^{M}- \bigg(\prescript{\circ}{}{\tilde{b}}^{y, \, M}_{-1}
+ \prescript{\circ}{}{\tilde{b}}^{y, \, M }_{-2}\, (1+\beta \, a_{y})
+  \cdots + 
\\ \prescript{\circ}{}{\tilde{b}}^{y}_{-(M-2)}\, (1+\beta \, a_{y})^{M-3}
\bigg)\beta \, a_{\pi} 
-\bigg(\frac{\prescript{\circ}{}{\tilde{b}}^{y, \, M}_{M}}{(1+ \beta \, a_{y})^{M}} + \cdots + \frac{\prescript{\circ}{}{\tilde{b}}^{y, \, M}_{1}}{1+ \beta \, a_{y}}\bigg)\equiv \\ 
\frac{(1-\beta^{M})}{1-\beta}\, (M + \eta)- \frac{\beta^{2}\, (1-\beta^{M-1})}{1-\beta}(1+\eta)\, a_{\pi}
- \\ 
\Bigg[ \frac{1-\beta^{M-2}}{1-\beta}
+ \beta \, \bigg(\frac{1-\beta^{M-3}}{1-\beta}\bigg)(1+\beta \, a_{y}) + \cdots \\ + \beta^{M-3}\, (1+\beta \, a_{y})^{M-3}\Bigg]\beta^{4}\, (1+\eta)
\, a_{\pi}\, a_{y}  \\ - \Bigg[ \frac{1-\beta}{(1+ \beta \, a_{y})^{M}} + \frac{1-\beta^{2}}{(1+ \beta \, a_{y})^{M-1}} + \cdots + \frac{1-\beta^{M}}{(1+ \beta \, a_{y})}\Bigg]\frac{\beta}{1-\beta}(1+\eta)\, a_{y} \equiv \\ 
\frac{(1-\beta^{M})}{1-\beta}\, (M + \eta)- \frac{\beta^{2}\, (1-\beta^{M-1})}{1-\beta}\, (1+\eta)\, a_{\pi} - \\ 
\Bigg[ a_{y}\bigg(\frac{\beta^{M-2} \, (1+\beta \, a_{y})^{M-2}-1}{\beta\, (1+\beta \, a_{y})-1}\bigg)
-\beta^{M-3}\bigg((1+\beta \, a_{y})^{M-2}-1\bigg)\Bigg]\frac{\beta^{4}\, (1+\eta)}{1-\beta}
\, a_{\pi} -\\ 
\Bigg[(1+\beta \, a_{y})^{M}-1- \beta^{2}\, a_{y}\bigg( \frac{ \beta^{M}\, (1+\beta \, a_{y})^{M}-1}{\beta\, (1+\beta \, a_{y})-1}\bigg)\Bigg]\frac{(1+\eta)}{(1-\beta)\, (1+\beta \, a_{y})^{M}}
\end{multline}
For each series, the simplifying procedure is to
separate each term into its \\ positive and negative components, then note that these form distinct geometric progressions. This tactic will reoccur throughout the derivation of (264).
\par For the last lag impact is only through the cost channel
 \begin{equation} \tilde{b}^{\pi, \, M}_{M}=\prescript{\circ}{}{\tilde{b}}^{\pi, \, M}_{M} -\bigg(\frac{\beta \, a_{\pi}}{1+ \beta \, a_{y}}\bigg) \prescript{\circ}{}{\tilde{b}}^{y, \, M}_{M}=
\beta \, a_{\pi}\, \frac{(1+\eta)}{1 + \beta \, a_{y}}\end{equation}
Subsequently, there are both cost channel and expected inflation effects 
\begin{multline} \tilde{b}^{\pi, \, M}_{M-1}=\prescript{\circ}{}{\tilde{b}}^{\pi, \, M}_{M-1} 
-\bigg[\frac{\prescript{\circ}{}{\tilde{b}}^{y, \, M}_{M}}{(1+ \beta a_{y})^{2}} + \frac{\prescript{\circ}{}{\tilde{b}}^{y, \, M}_{M-1}}{1+ \beta \, a_{y}}\bigg]\beta \, a_{\pi}
+ \frac{\prescript{\circ}{}{\tilde{b}}^{y, \, M}_{M}}{(1+ \beta \, a_{y})} \equiv \\ 
\beta \, 
(1 + \beta)\, (1+\eta)\, a_{\pi} 
 -\eta - \beta \, (1+\eta)\, a_{y}\Bigg[ \beta \, a_{\pi}
\bigg(\frac{1}{(1+ \beta \, a_{y})^{2}} + \frac{
(1+ \beta)}{1+ \beta \, a_{y}} 
\bigg)- \frac{1}{1+\beta \, a_{y}}\Bigg]
\end{multline}
\begin{multline} \tilde{b}^{\pi, \, M}_{M-2}=\prescript{\circ}{}{\tilde{b}}^{\pi, \, M}_{M-2} - \bigg[\frac{\prescript{\circ}{}{\tilde{b}}^{y, \, M}_{M}}{(1+ \beta \, a_{y})^{3}} + \frac{\prescript{\circ}{}{\tilde{b}}^{y, \, M}_{M-1}}{(1+ \beta \, a_{y})^{2}} + \frac{\prescript{\circ}{}{\tilde{b}}^{y, \, M}_{M-2}}{1+ \beta \, a_{y}}\bigg]\beta \, a_{\pi}
+ \frac{\prescript{\circ}{}{\tilde{b}}^{y, \, M}_{M}}{(1+ \beta \, a_{y})^{2}} + \\ \frac{\prescript{\circ}{}{\tilde{b}}^{y, \, M}_{M-1}}{(1+ \beta \, a_{y})} \equiv   
\beta \, (1 + \beta + \beta^{2})\, (1+\eta)\, a_{\pi} -
(1 + \beta)\, \eta 
- \\ \beta\, (1+\eta)\, a_{y}\Bigg[ \beta \, a_{\pi}\bigg( \frac{
1}{(1+ \beta \, a_{y})^{3}} + \frac{(1+ \beta)}{(1+ \beta\,  a_{y})^{2}} + \frac{(1+ \beta+ \beta^{2})}{(1+ \beta \, a_{y})}\bigg) - \\ 
\bigg( \frac{1}{(1+ \beta \, a_{y})^{2}} + \frac{(1+ \beta)}{(1+ \beta \, a_{y})} \bigg)\Bigg]
\end{multline}
$$\vdots$$
\begin{multline} \tilde{b}^{\pi, \, M}_{1}=\prescript{\circ}{}{\tilde{b}}^{\pi, \, M}_{1} - \bigg[\frac{\prescript{\circ}{}{\tilde{b}}^{y, \, M}_{M}}{(1+ \beta \, a_{y})^{M}} + \cdots + \frac{\prescript{\circ}{}{\tilde{b}}^{y, \, M}_{1}}{1+ \beta \, a_{y}}\bigg] \beta \, a_{\pi} + \frac{\prescript{\circ}{}{\tilde{b}}^{y, \, M}_{M}}{(1+ \beta \, a_{y})^{M-1}} + \\ \cdots + \frac{\prescript{\circ}{}{\tilde{b}}^{y, \, M}_{2}}{1+ \beta \, a_{y}}
\equiv \beta \, 
(1 + \beta + \cdots +\beta^{M-1}
)\, (1+\eta)\, a_{\pi} -(1 + \beta + \cdots + \beta^{M-2})\ \eta 
- \\ 
\beta \, \frac{(1+\eta)}{1-\beta} \, a_{y}\Bigg[ 
\beta \, a_{\pi} \bigg(\frac{1-\beta}{(1+ \beta \, a_{y})^{M}} + \frac{1-\beta^{2}}{(1+ \beta \, a_{y})^{M-1}} + \cdots +  \frac{1-\beta^{M}}{(1+ \beta \, a_{y})} \bigg)- \\ \bigg(\frac{1-\beta}{(1+\beta \, a_{y})^{M-1}} + \frac{1-\beta^{2}}{(1+ \beta \, a_{y})^{M-2}} + \cdots +  \frac{1-\beta^{M-1}}{(1+ \beta \, a_{y})}\bigg)\Bigg]
\equiv 
\\ \frac{1}{1-\beta}
\Bigg\{\beta \, (1-\beta^{M})\, (1+\eta)\, a_{\pi} - 
(1-\beta^{M-1})\, \eta 
- \\ \frac{\beta \, a_{\pi}\, (1+\eta)}{(1+ \beta \, a_{y})^{M}}\Bigg[ (1+ \beta \, a_{y})^{M}-1- \beta^{2}\, a_{y} \, \bigg\{\frac{\beta^{M}\, (1+\beta \, a_{y})^{M}-1}{\beta\, (1+\beta \, a_{y})-1}\bigg\}\Bigg] + \\ 
\frac{(1+\eta)}{(1+ \beta \, a_{y})^{M-1}}\Bigg[ (1+ \beta \, a_{y})^{M-1}-1- \beta^{2}\, a_{y} \, \bigg\{\frac{\beta^{M-1}\, (1+\beta \, a_{y})^{M-1}-1}{\beta\, (1+\beta \, a_{y})-1}\bigg\}\Bigg]\Bigg\}
\end{multline} 
\begin{multline} 
\tilde{b}^{\pi, \, M}_{-1}=\prescript{\circ}{}{\tilde{b}}^{\pi, \, M}_{-1} + \bigg(\prescript{\circ}{}{\tilde{b}}^{y, \, M}_{-2} + \prescript{\circ}{}{\tilde{b}}^{y, \, M}_{-3}\, (1+\beta \, a_{y}) +  \cdots + \prescript{\circ}{}{\tilde{b}}^{y, \, M}_{-(M-2)}\, (1+\beta \, a_{y})^{M-4}\bigg)\beta \, a_{\pi} \\ - \bigg(\prescript{\circ}{}{\tilde{b}}^{y, \, M}_{-1} + \prescript{\circ}{}{\tilde{b}}^{y, \, M}_{-2}\, (1+\beta \, a_{y}) +  \cdots + \prescript{\circ}{}{\tilde{b}}^{y, \, M}_{-(M-2)}\, (1+\beta \, a_{y})^{M-3}\bigg) \equiv\\ \beta \, \bigg[ \beta^{2}\, \bigg(\frac{1-\beta^{M-2}}{1-\beta}\bigg)(1+\eta)\, a_{\pi} - \bigg(\frac{1-\beta^{M-1}}{1-\beta}\bigg) \eta \bigg] + \\ 
\Bigg[ \beta^{4}\, \bigg(\frac{1-\beta^{M-3}}{1-\beta}\bigg) + \beta^{5}\, \bigg(\frac{1-\beta^{M-4}}{1-\beta}\bigg)(1+\beta a_{y}) + \beta^{6}\, \bigg(\frac{1-\beta^{M-5}}{1-\beta}\bigg)(1+\beta a_{y})^{2} + \cdots + \\   \beta^{M}(1+\beta a_{y})^{M-4} \Bigg]\beta\, (1+\eta)\, a_{y}\, a_{\pi} - \Bigg[ \beta^{3}\, \bigg(\frac{1-\beta^{M-2}}{1-\beta}\bigg) + \\\beta^{4}\, \bigg(\frac{1-\beta^{M-3}}{1-\beta}\bigg)(1+\beta \, a_{y}) +  \cdots + \beta^{M}\, (1+\beta \, a_{y})^{M-3}\Bigg](1+\eta)\, a_{y} \equiv 
\\ \frac{\beta}{1-\beta} \Bigg[ \beta^{2}\, (1+\eta)\, \Bigg\{\bigg[1-\beta^{M-2}
+ \beta^{2} \, a_{y}\, \bigg(\frac{\big\{ \beta\, (1+\beta \, a_{y})\big\}^{M-3}-1}{\beta\, (1+\, \beta \, a_{y})-1} \bigg) - \\ \beta^{M-2}\, \bigg( (1+\beta \, a_{y})^{M-3}-1\bigg)\bigg]
a_{\pi} -\bigg[a_{y}\, \bigg(\frac{\big\{ \beta\, (1+\beta\,  a_{y})\big\}^{M-2}-1}{\beta\, (1+\beta \, a_{y})-1} \bigg)
- \\ \beta^{M-3}\, \bigg( (1+\beta \, a_{y})^{M-2}-1\bigg)\bigg]\Bigg\}-(1-\beta^{M-1}) \, \eta \Bigg]
\end{multline}
The cost channel term comes about from substituting out future output \\ expectations in the Euler equation until one gets to $\mathbb{E}_{t}\, \hat{y}^{e}_{t+2}$. The expectation channel reflects the final step when one has reached $\mathbb{E}_{t}\, \hat{y}^{e}_{t+1}$. Once again, the tactic was 
to break each bracket into two geometric progressions. In a similar vein,  
\begin{multline} \tilde{b}^{\pi, \, M}_{-2}=\prescript{\circ}{}{\tilde{b}}^{\pi, \, M}_{-2} + \bigg(\prescript{\circ}{}{\tilde{b}}^{y, \, M}_{-3} + \prescript{\circ}{}{\tilde{b}}^{y, \, M}_{-4}\, (1+\beta \, a_{y}) +  \cdots + \prescript{\circ}{}{\tilde{b}}^{y, \, M}_{-(M-2)}(1+\beta \, a_{y})^{M-5}\bigg)\beta \, a_{\pi}  \\ -\bigg(\prescript{\circ}{}{\tilde{b}}^{y, \, M}_{-2} + \prescript{\circ}{}{\tilde{b}}^{y, \, M}_{-3}\, (1+\beta \, a_{y}) +  \cdots + \prescript{\circ}{}{\tilde{b}}^{y, \, M}_{-(M-2)}(1+\beta \, a_{y})^{M-4}\bigg) \equiv \\\beta^{2} \bigg[ \beta^{2}\, \bigg(\frac{1-\beta^{M-3}}{1-\beta}\bigg)(1+\eta)\,  a_{\pi} - \bigg(\frac{1-\beta^{M-2}}{1-\beta}\bigg) \eta \bigg] + \\ 
\Bigg[ \beta^{5}\, \bigg(\frac{1-\beta^{M-4}}{1-\beta}\bigg) + \beta^{6}\, \bigg(\frac{1-\beta^{M-5}}{1-\beta}\bigg)(1+\beta \, a_{y}) + \beta^{7}\, \bigg(\frac{1-\beta^{M-6}}{1-\beta}\bigg)(1+\beta \, a_{y})^{2} + \\ \cdots + \beta^{M}\, (1+\beta \, a_{y})^{M-5} \Bigg]\beta\, (1+\eta)\, a_{y}\, a_{\pi} - \Bigg[ \beta^{4}\, \bigg(\frac{1-\beta^{M-3}}{1-\beta}\bigg) + \\ \beta^{5}\, \bigg(\frac{1-\beta^{M-4}}{1-\beta}\bigg)\, (1+\beta \, a_{y}) +  \cdots + \beta^{M}\,(1+\beta \, a_{y})^{M-4}\Bigg]\, (1+\eta)\, a_{y} \equiv 
\\ \frac{\beta^{2}}{1-\beta} \Bigg[ \beta^{2}\, (1+\eta)\, \Bigg\{\bigg[1-\beta^{M-3}
+ \beta^{2} \, a_{y}\, \bigg(\frac{\big\{ \beta\, (1+\beta \, a_{y})\big\}^{M-4}-1}{\beta\, (1+\beta \, a_{y})-1} \bigg) - \\ \beta^{M-3}\bigg( (1+\beta \, a_{y})^{M-4}-1\bigg)\bigg]\, 
a_{\pi} -\bigg[a_{y}\, \bigg(\frac{\big\{ \beta\, (1+\beta \, a_{y})\big\}^{M-3}-1}{\beta\, (1+\beta \, a_{y})-1} \bigg)
- \\ 
\beta^{M-4}\, \bigg( (1+\beta \, a_{y})^{M-3}-1\bigg)\bigg]\Bigg\}-(1-\beta^{M-2}) \, \eta \Bigg]
\end{multline}
$$\vdots$$
The last term with both a cost and expectation channel is 
\begin{multline}\tilde{b}^{\pi, \, M}_{-(M-3)}=\prescript{\circ}{}{\tilde{b}}^{\pi, \, M}_{-(M-3)} + \beta \, a_{\pi}\, \prescript{\circ}{}{\tilde{b}}^{y, \, M}_{-(M-2)} -\bigg(\prescript{\circ}{}{\tilde{b}}^{y, \, M}_{-(M-3)} + (1+\beta \, a_{y})\, \prescript{\circ}{}{\tilde{b}}^{y, \, M}_{-(M-2)}\bigg)
\equiv  \\ (\beta^{M+1}\, a_{y} + \beta^{M} + \beta^{M-1})\, (1+\eta)\, (a_{\pi}-a_{y}) - \\ (\beta^{M-1} + \beta^{M-2} + \beta^{M-3})\, \eta - \beta^{M}\, (1+\eta)\, a_{y}
\end{multline}
At the penultimate lead there is just the expectations channel 
\begin{equation} \tilde{b}^{\pi, \, M}_{-(M-2)}=\prescript{\circ}{}{\tilde{b}}^{\pi, \, M}_{-(M-2)} 
-\prescript{\circ}{}{\tilde{b}}^{y, \, M}_{-(M-2)}=\beta^{M}\, (1+\eta)\, (a_{\pi}-a_{y}) - (\beta^{M-1} +\beta^{M-2})\, \eta 
\end{equation}
Furthest into the future neither are present
\begin{equation} \tilde{b}^{\pi, \, M}_{-(M-1)}=\prescript{\circ}{}{\tilde{b}}^{\pi, \, M}_{-(M-1)}=-\beta^{M-1}\eta\end{equation}
The single output gap term is a little more complicated, featuring terms from both the past and future.\footnote{It is formed out of two pairs of geometric progressions. 
First, substitute in 
(256)-(261) in expanded form.  
Next, separate out the expression into terms derived from lags 
$\bigg(\prescript{\circ}{}{\tilde{b}}^{y, \, M}_{i >0}\bigg)$
and those emanating from current and expected effects $\bigg(\prescript{\circ}{}{\tilde{b}}^{y, \, M}_{i \leq 0}\bigg)$
\begin{multline*}\tilde{b}^{y, \, M}
= \frac{\beta (1+\eta)a_{y}}{(1+\beta a_{y})^{M}} + \frac{(\beta^{2} + \beta)(1+\eta)a_{y}}{(1+\beta a_{y})^{M-1}} + \cdots + \frac{
(\beta^{M} + \cdots + \beta)(1+\eta)a_{y}}{(1+\beta a_{y})} + \\  (\beta^{M} + \cdots + \beta^{2}) (1+\eta)a_{y} + (\beta^{M} + \cdots + \beta^{3}) (1+\eta)a_{y}(1+\beta a_{y}) + \cdots + \beta^{M}(1+\eta)a_{y}(1+\beta a_{y})^{M-2} 
\\ \equiv \frac{\beta a_{y}(1+\eta)}{(1-\beta)(1+ \beta a_{y})}\bigg[ 1-\beta + \frac{1-\beta^{2}}{(1+ \beta a_{y})} + \cdots + \frac{1-\beta^{M-1}}{(1+ \beta a_{y})^{M-1}}\bigg] 
+ \\ \frac{\beta^{2} a_{y}(1+\eta)}{(1-\beta)}\bigg[ 1-\beta^{M-1} + \beta(1-\beta^{M-2})(1+ \beta a_{y}) + \cdots + \beta^{M-2}(1+ \beta a_{y})^{M-2}\bigg]
\end{multline*}
Inside each component there are two geometric progressions one for positive and the other featuring negative terms.
Finally, routine algebra leads to 
(275) in the text.} 
\begin{multline} \tilde{b}^{y, \, M}
= \frac{\prescript{\circ}{}{\tilde{b}}^{y, \, M}_{M}}{(1+\beta a_{y})^{M}} + \frac{\prescript{\circ}{}{\tilde{b}}^{y, \, M}_{M-1}}{(1+\beta \, a_{y})^{M-1}} + \cdots + \frac{\prescript{\circ}{}{\tilde{b}}^{y, \, M}_{1}}{(1+\beta \, a_{y})} + \prescript{\circ}{}{\tilde{b}}^{y, \, M}_{0} + \\ \prescript{\circ}{}{\tilde{b}}^{y, \, M}_{-1}\,  (1+\beta \, a_{y}) +  \cdots + \prescript{\circ}{}{\tilde{b}}^{y, \, M}_{-(M-2)}\, (1+\beta \, a_{y})^{M-2} \equiv \\  
\Bigg[ (1+ \beta \, a_{y})^{M}-1 - \beta^{2}\, a_{y} \, \bigg( \frac{\beta^{M}\, (1+ \beta \, a_{y})^{M}-1}{\beta\, (1+ \beta \, a_{y})-1}\bigg)\Bigg]\frac{(1+\eta)}{(1-\beta)\, (1+\beta \, a_{y})^{M}} + \\ \frac{\beta^{2}\, (1+\eta)}{(1-\beta)}\Bigg[ a_{y} \, \bigg( \frac{\beta^{M-1}\, (1+ \beta \, a_{y})^{M-1}-1}{\beta\, (1+ \beta \, a_{y})-1}\bigg)-\beta^{M-2}\, \bigg( (1+ \beta \, a_{y})^{M-1}-1\bigg)\Bigg]
\end{multline}
The errors are constructed as follows
\begin{multline} 
\hat{u}^{M}_{t}=
\prescript{\circ}{\circ}{b}^{x, \, M}_{M-1}\, (1-\rho)\, \hat{\psi}_{t-M} + \prescript{\circ}{\circ}{b}^{x, \, M}_{M-2}\, (1-\rho)\, \hat{\psi}_{t-(M-1)} + \cdots + \prescript{\circ}{\circ}{b}^{x, \, M}_{0}\, (1-\rho)\, \hat{\psi}_{t-1} - \\ \prescript{\circ}{\circ}{b}^{\psi, \, M}_{0}\, 
\hat{\psi}_{t} 
- \prescript{\circ}{\circ}{b}^{x, \, M}_{M-1}\bigg(\pi_{t-(M-1)}- \mathbb{E}_{t-M}\, \pi_{t-(M-1)}\bigg) - \cdots - \prescript{\circ}{\circ}{b}^{x, \, M}_{0}\bigg(\pi_{t}- \mathbb{E}_{t-1}\, \pi_{t}\bigg) - \\ \prescript{\circ}{\circ}{b}^{x, \, M}_{M-1}\bigg(\hat{y}^{e}_{t-(M-1)}- \mathbb{E}_{t-M}\, \hat{y}^{e}_{t-(M-1)}\bigg) - \cdots - \prescript{\circ}{\circ}{b}^{x, \, M}_{0}\bigg(\hat{y}^{e}_{t}- \mathbb{E}_{t-1}\, \hat{y}^{e}_{t}\bigg) +\frac{\prescript{\circ}{}b^{M}}{b^{M}} \, \hat{v}^{2, \, M}_{t}\end{multline}
In keeping with the previous notation, 
$\prescript{\circ}{\circ}b^{j, \, M}_{i}=\prescript{\circ}{\circ}{\tilde{b}}^{j, M}_{i}/b^{M}$.
Notice that the coefficients 
$b^{x, \, M}_{i}$at each horizon match those for the expectation effect in (267)-(269). As before, this is a consequence of switching between expected and observed quantities. Output gap and inflation coefficients differ in sign because of opposing expectation effects in the Euler equation. Formally, 
\begin{equation} \prescript{\circ}{\circ}{\tilde{b}}^{x, \, M}_{M-1}= \frac{\prescript{\circ}{}{\tilde{b}}^{y, \, M}_{M}}{(1+\beta \, a_{y})}\equiv \frac{\beta \, (1+\eta)\, a_{y}}{1+ \beta \, a_{y}}
\end{equation}
\begin{multline} \prescript{\circ}{\circ}{\tilde{b}}^{x, \, M}_{M-2}=\frac{\prescript{\circ}{}{\tilde{b}}^{y, \, M}_{M}}{(1+ \beta \, a_{y})^{2}} + \frac{\prescript{\circ}{}{\tilde{b}}^{y, \, M}_{M-1}}{1+ \beta \, a_{y}}\equiv \frac{\beta \, (1+\eta)\, a_{y}}{(1+ \beta a_{y})^{2}} + \frac{\beta\, (1 + \beta)\, (1+\eta)\, a_{y}}{(1+ \beta \, a_{y})} \\  = \frac{\beta\, (1+\eta)\, a_{y}}{(1+ \beta \, a_{y})^{2}} \bigg(1+ (1 + \beta)\, (1+ \beta \, a_{y})\bigg)
\end{multline}
$$ \vdots $$
\begin{multline} \prescript{\circ}{\circ}{\tilde{b}}^{x, \, M}_{0}=\frac{\prescript{\circ}{}{\tilde{b}}^{y, \, M}_{M}}{(1+ \beta \, a_{y})^{M}} + \cdots + \frac{\prescript{\circ}{}{\tilde{b}}^{y, \, M}_{1}}{1+ \beta \, a_{y}}\equiv 
\\ \frac{(1+\eta)}{(1-\beta)\, (1+ \beta \, a_{y})^{M}}
\Bigg[ (1+ \beta \, a_{y})^{M}-1 - \beta^{2}\, a_{y} \, \bigg( \frac{\beta^{M}\, (1+ \beta \, a_{y})^{M}-1}{\beta\, (1+ \beta \, a_{y})-1}\bigg)\Bigg]\end{multline}
\begin{multline} \prescript{\circ}{\circ}{\tilde{b}}^{\psi, \, M}_{0}= 
(1-\rho)\bigg( \prescript{\circ}{}{\tilde{b}}^{y, \, M}_{-1}+ \prescript{\circ}{}{\tilde{b}}^{y, \, M}_{-2}\, (1+\beta \, a_{y}) + \cdots + \prescript{\circ}{}{\tilde{b}}^{y, \, M}_{-(M-2)}\, (1+\beta \, a_{y})^{M-3}\bigg) 
+ \\ \rho\, (1-\rho)\bigg( \prescript{\circ}{}{\tilde{b}}^{y, \, M}_{-2}+ \prescript{\circ}{}{\tilde{b}}^{y, \, M}_{-3}\, (1+\beta \, a_{y}) + \cdots + \prescript{\circ}{}{\tilde{b}}^{y, \, M}_{-(M-2)}(1+\beta \, a_{y})^{M-4}\bigg) 
+ \cdots 
+ \\ \rho^{M-3}\, (1-\rho)\, \prescript{\circ}{}{\tilde{b}}^{y, \, M}_{-(M-2)} \equiv 
\frac{\beta^{3}\, (1-\rho)}{1-\beta}\, (1+\eta)\, a_{y}\Bigg[
1-\beta^{M-2} + \\ \rho \, \beta\, (1-\beta^{M-3})\, (1+\beta \, a_{y})   + \rho^{2} \, \beta^{2}\, (1-\beta^{M-4})\, (1+\beta \, a_{y})^{2}  + \cdots +  \\ \rho^{M-3} \, \beta^{M-3}\, (1-\beta)\, (1+\beta \, a_{y})^{M-3}\Bigg] \equiv  \\ 
\frac{\beta^{3}\, (1-\rho)}{1-\beta}(1+\eta)\, a_{y}\Bigg[\bigg(\frac{\big\{\rho \, \beta\, (1+\beta \, a_{y})\big\}^{M-2}-1}{\rho \, \beta\, (1+\beta \, a_{y})-1}\bigg)-  \\ \beta^{M-2}\, \bigg( \frac{\big\{\rho\, (1+\beta \, a_{y})\big\}^{M-2}-1}{\rho\, (1+\beta \, a_{y})-1} \bigg)\Bigg]
\end{multline}
Thus, a convenient Phillips curve interrelationship \begin{equation}\pi_{t}=\tilde{g}_{0}(\mathbb{E}_{t}\, \pi_{t+M-2}, \, \cdots , \, \mathbb{E}_{t}\, \pi_{t+1}, \, \pi_{t-1}, \,  \cdots ,  \, \pi_{t-M}, \,  Y^{e}_{t}, \, \psi_{t}, \, \psi_{t-1}, \, \cdots , \, \psi_{t-M})
\end{equation}
appears. A couple of hurdles remain on the way to recursive equilibrium. First, lag the relationship and expand $M-3$ periods, then expand the expectations to get  
\begin{multline}\pi_{t}=\tilde{g}_{1}(\mathbb{E}_{t}, \pi_{t+1}, \, \pi_{t-1}, \,  \cdots ,  \, \pi_{t-2(M-1)}, \\ \,  Y^{e}_{t}, \, \cdots, \, Y^{e}_{t-(M-1)}, \, \psi_{t}, \, \psi_{t-1}, \, \cdots , \, \psi_{t-2(M-1)})\end{multline}
To reach the due destination, make 
tedious substitutions of the aggregate \\ demand system to remove lagged output, in the same way that these past terms were removed from 
(246). Lastly, rearrange for $\mathbb{E}_{t}\pi_{t+1}$,
which is valid because $\tilde{b}^{\pi, \, M}_{-(M-1)}<0$, in 
(274). The error process (276) is unchanged. The terminal argument to prevent the system collapsing in on its self applies mutatis mutandis from the twice annual price changing realm via (266), (267) and (276).
\end{proof}
\section{Taylor Contracts (Additional Results)}
This appendix features two sets of interesting derivations, for which there was no space in the main text, spread across three subsections. The first processes disturbances into forms fit for future investigation. The second gives general \\ patient limit Phillips curves, extending (115), (116) and (264). The third \\ completes the process for the error components.  
\subsection{Full Error Coefficients}
This subsection contains the lengthy expressions for the two final error terms from the Taylor pricing Phillips curves (157) and (276). For ease of location shorter and longer contracts each have their own heading. 
\subsubsection{Taylor (2) Disturbances}
Here is the expression.
\begin{multline} \hat{u}^{2}_{t}=b^{\psi, \, 2}_{2}\, \hat{\psi}_{t-2} + b^{\psi, \, 2}_{1}\, \hat{\psi}_{t-1} + b^{\psi, \, 2}_{0}\, \hat{\psi}_{t} + b^{\pi(e), \, 2}_{1}\, (\pi_{t-1}-\mathbb{E}_{t-2}\, \pi_{t-1}) + \\ b^{\pi(e), \, 2}_{0}\, (\pi_{t}-\mathbb{E}_{t-1}\, \pi_{t}) + 
b^{y(e)\, 2}_{1}\, (\hat{y}^{e}_{t-1}-\mathbb{E}_{t-2}\, \hat{y}^{e}_{t-1}) + b^{y(e), \, 2}_{0}\, (\hat{y}^{e}_{t}-\mathbb{E}_{t-1}\, \hat{y}^{e}_{t})\end{multline}
Throughout this subsection, normalization conventions follow the equation in the main text. The procedure is to unfurl earlier error terms (130), (132) and (145) inside (157).
 \begin{equation} \tilde{b}^{\psi, \, 2}_{2}= \prescript{\circ}{\circ}{\tilde{b}}^{\psi, \, 2}_{2}+ \prescript{\circ}{}{\tilde{b}}^{\psi, \, 2}_{2}=-\frac{(1-\rho)\, (1+\eta)}{1 +\beta \, a_{y}}\end{equation}
\begin{equation} \tilde{b}^{\psi, \, 2}_{1}=\prescript{\circ}{\circ}{\tilde{b}}^{\psi, \, 2}_{1} + \prescript{\circ}{}{\tilde{b}}^{\psi, \, 2}_{1} =-\frac{(1-\rho)\, (1+\eta)}{(1 +\beta \, a_{y})^{2}}\bigg( 1 +\beta + \beta^{2}\, a_{y}\bigg)\end{equation}
\begin{equation} \tilde{b}^{\psi, \, 2}_{0}=\prescript{\circ}{}{\tilde{b}}^{\psi, \, 2}_{0}=-\beta\, (1-\rho)\, (1+\eta)\end{equation}
\begin{equation} \tilde{b}^{\pi(e), \, 2}_{1}= 
\prescript{\circ}{\circ}{\tilde{b}}^{x, \, 2}_{1} + \prescript{\circ}{}{\tilde{b}}^{x, \, 2}_{1} + \beta = \frac{(1+\eta)\, a_{y}}{ 1+ \beta \, a_{y}}+ \beta\end{equation}
\begin{equation} \tilde{b}^{\pi(e), \, 2}_{0}=  \prescript{\circ}{\circ}{\tilde{b}}^{x, \, 2}_{0} + \prescript{\circ}{\circ}{\tilde{b}}^{x, \, 2}_{0}+\beta= 
\bigg( 1 + \beta + \beta \big\{ 1 + \beta(1-\beta)\big\}a_{y}\bigg)\frac{(1+\eta)}{(1+\beta a_{y})^{2}} + \beta\end{equation}
\begin{equation} \tilde{b}^{y(e), \, 2}_{1}=   \prescript{\circ}{\circ}{\tilde{b}}^{x, \, 2}_{1} + \prescript{\circ}{}{\tilde{b}}^{x, \, 2}_{1}  +\beta\, (1+\eta) = \frac{(1+\eta)\, a_{y}}{ 1+ \beta }\bigg( \beta + (1+\beta^{2})\, a_{y}\bigg)\end{equation}
\begin{equation} \tilde{b}^{y(e), \, 2}_{0}=  \prescript{\circ}{\circ}{\tilde{b}}^{x, \, 2}_{0}+ \prescript{\circ}{}{\tilde{b}}^{x, \, 2}_{0} +\beta\, (1+\eta)= 
\bigg( 2 + \beta + \beta \, \big\{ 3 + \beta\, (1-\beta)\big\}a_{y} + \beta^{2} \, a_{y}^{2}\bigg)\frac{(1+\eta)}{(1+\beta \, a_{y})^{2}} \end{equation}
Naturally, in order, (158), (146), (159), (147), (148), (160), (149), (161) and (150) played a role in reaching the 
clean result sought. Notice that all the expectation adjustment \\ coefficients are positive, whilst all the error term coefficients are negative. There is an opportunity here for computation of an exact form solution, which I offer for others.
\subsubsection{Taylor (M) Disturbances}
In the same spirit, I proceed from 
\begin{multline} \hat{u}^{M}_{t}= b^{\psi, \, M}_{1}\, \hat{\psi}_{t-(M-1)}+ \cdots + b^{\psi, \, M}_{1}\, \hat{\psi}_{t-1} + b^{\psi, \, M}_{0}\, \hat{\psi}_{t} + \\ b^{\pi(e), \, M}_{M-1}\, (\pi_{t-(M-1)}-\mathbb{E}_{t-M}\, \pi_{t-(M-1)}) + b^{\pi(e), \, M}_{M-2}(\pi_{t-(M-2)}-\mathbb{E}_{t-(M-1)}\, \pi_{t-(M-2)}) \\ + \cdots +  b^{\pi(e), \, M}_{1}\, (\pi_{t-1}-\mathbb{E}_{t-2}\, \pi_{t-1}) + b^{\pi(e), \, M}_{0}\, (\pi_{t}-\mathbb{E}_{t-1}\, \pi_{t}) + \\ b^{y(e), \, M}_{M-1}\, (\hat{y}^{e}_{t-(M-1)}-\mathbb{E}_{t-M}\, \hat{y}^{e}_{t-(M-1)}) + b^{y(e), \, M}_{M-2}\, (\hat{y}^{e}_{t-(M-2)}-\mathbb{E}_{t-(M-1)}\, \hat{y}^{e}_{t-(M-2)}) \\ + \cdots +  b^{y(e), \, M}_{1}\, (\hat{y}^{e}_{t-1}-\mathbb{E}_{t-2}\, \hat{y}^{e}_{t-1}) + b^{y(e), \, M}_{0}\, (\hat{y}^{e}_{t}-\mathbb{E}_{t-1}\, \hat{y}^{e}_{t})\end{multline}
whose coefficients are 
\begin{equation}
 \tilde{b}^{\psi, \, M}_{M}= \prescript{\circ}{\circ}{\tilde{b}}^{x, \, M}_{M-1}\, (1-\rho) 
 - (1-\rho)\, (1+\eta)=  -\frac{(1-\rho)\, (1+\eta)}{(1+ \beta\,  a_{y})}\end{equation}

\begin{multline}
 \tilde{b}^{\psi, \, M}_{M-1}= \prescript{\circ}{\circ}{\tilde{b}}^{x, \, M}_{M-2}\, (1-\rho) 
 - (1+\beta)\, (1-\rho)\, (1+\eta)= \\ 
-\frac{(1-\rho)\, (1+\eta)}{(1+ \beta \, a_{y})^{2}}\bigg( 1+ \beta\, (1+ \beta \, a_{y})\bigg)
\end{multline}
$$\vdots$$
\begin{multline}
 \tilde{b}^{\psi, \, M}_{1}= \prescript{\circ}{\circ}{\tilde{b}}^{x, \, M}_{0}(1-\rho) 
 - \bigg(\frac{1-\beta^{M}}{1-\beta}\bigg)(1-\rho)(1+\eta)= 
\\ -\frac{(1-\rho)\, (1+\eta)}{(1+ \beta \, a_{y})^{M}}\bigg( \frac{\beta^{M}(1+ \beta \, a_{y})^{M}-1}{\beta\, (1+ \beta \, a_{y})-1}\bigg)
\end{multline}
\begin{multline} \tilde{b}^{\psi, \, M}_{0}= \prescript{\circ}{\circ}{b}^{\psi, \, M}_{0} + (1+\eta)\, \prescript{\circ}{}{\tilde{b}}^{\psi, \, M}_{0} = \\ \frac{\beta \, (1+\eta)}{(1-\beta)}\Bigg(\frac{1-\rho - \beta^{M-1}\, \big\{ 1-\beta \, \rho-\rho^{M}\, (1-\beta)\big\} }{1-\beta \, \rho} + \\ \frac{\beta^{2}\, a_{y} \, (1-\rho)}{\big\{\rho \, \beta\, (1+\beta \, a_{y})-1\big\}\, \big\{\rho \, (1+\beta\,  a_{y})-1\big\}}\times \\ \bigg[ 1-\beta^{M-2} + \big\{\rho (1+\beta \, a_{y})-1\big\}\big\{ \beta^{M-2}\, \big\{\rho \, (1+\beta \, a_{y})-1\big\}^{M-2}(1-\beta)-1-\beta^{M-1}\big\}\bigg]\Bigg)
\end{multline}
\begin{equation}
 \tilde{b}^{\pi(e), \, M}_{M-1}= -\prescript{\circ}{\circ}{\tilde{b}}^{x, \, M}_{M-1}
 + (1+\eta)=  \frac{(1+\eta)}{(1+ \beta \, a_{y})}
\end{equation}
\begin{equation}
 \tilde{b}^{\pi(e), \, M}_{M-2}= -\prescript{\circ}{\circ}{\tilde{b}}^{x, \, M}_{M-2}
 + (1+ \beta )\, (1+\eta)=  
 \frac{(1+\eta)}{(1+ \beta \, a_{y})^{2}}\bigg( 1+ \beta\, (1+ \beta\,  a_{y})\bigg)
 \end{equation}
 $$\vdots$$
 \begin{multline}
 \tilde{b}^{\pi(e), \, M}_{1}= 
 - \prescript{\circ}{\circ}{\tilde{b}}^{x, \, M}_{1}
 + \bigg(\frac{1-\beta^{M-1}}{1-\beta}\bigg)(
 1+\eta)
 = \\ \frac{(1+\eta)}{(1+ \beta \,a_{y})^{M-1}}\bigg( \frac{\beta^{M-1}\, (1+ \beta \, a_{y})^{M-1}-1}{\beta\, (1+ \beta \, a_{y})-1}\bigg)
\end{multline} 
\begin{multline}
 \tilde{b}^{\pi(e), \, M}_{0}= 
 - \prescript{\circ}{\circ}{\tilde{b}}^{x, \, M}_{0}
 + \bigg(\frac{1-\beta^{M}}{1-\beta}\bigg)(
 1+\eta)
 = \\ 
\frac{(1+\eta)}{(1+ \beta \, a_{y})^{M}}\bigg( \frac{\beta^{M}\, (1+ \beta \, a_{y})^{M}-1}{\beta\, (1+ \beta \, a_{y})-1}\bigg)
\end{multline} 
\begin{equation}
 \tilde{b}^{y(e), \, M}_{M-1}= \tilde{b}^{\pi(e), \, M}_{M-1}=  \frac{(1+\eta)}{(1+ \beta \, a_{y})}
\end{equation}
\begin{equation}
 \tilde{b}^{y(e), \, M}_{M-2}=  \frac{(1+\eta)}{(1+ \beta \, a_{y})^{2}}\bigg( 1+ \beta\, (1+ \beta \, a_{y})\bigg)
 \end{equation}
 $$\vdots$$
 \begin{equation}
 \tilde{b}^{y(e), \, M}_{1}=\tilde{b}^{\pi(e), \, M}_{1}= \frac{(1+\eta)}{(1+ \beta \, a_{y})^{M-1}}\bigg( \frac{\beta^{M-1}\, (1+ \beta \, a_{y})^{M-1}-1}{\beta\, (1+ \beta \, a_{y})-1}\bigg)
\end{equation} 
\begin{equation}
 \tilde{b}^{y(e), \, M}_{0}=\tilde{b}^{\pi(e), \, M}_{0}= 
\frac{(1+\eta)}{(1+ \beta \, a_{y})^{M}}\bigg( \frac{\beta^{M}\, (1+ \beta \, a_{y})^{M}-1}{\beta\, (1+ \beta \, a_{y})-1}\bigg)
\end{equation} 
The derivation involves recursively replacing previous error terms in (276) with their first-principled counterpart, using 
(242) and (262), then simplifying. 
\\ (277)-(280) were mustered in this effort. 
\subsection{Patient Limit Slope Coefficients (Taylor M)}
This subsection is divided into two subsubsections. The first deals with the multiple output gap form (246), whilst the second treats the main result (264). The errors are omitted for brevity. There is no need to bother with (241), where the coefficients agree with (367) in SELCKE.  
\subsubsection{Multiple Output Gap Form}
Each coefficient is built from the denominator in (247), which becomes
\begin{equation} \prescript{\circ}{}b^{M} =
M\, (M + \eta)- (M-1)\, (1+\eta)\, a_{\pi}\end{equation}
The inflation numerators, corresponding to (248)-(255), are 
\begin{equation} \prescript{\circ}{}{\tilde{b}}^{\pi, \, M}_{M}= (1+\eta)\, a_{\pi}\end{equation}
\begin{equation}  \prescript{\circ}{}{\tilde{b}}^{\pi, \, M}_{M-1}= 2\,(1+\eta)\, a_{\pi} -\eta
\end{equation}
\begin{equation}  \prescript{\circ}{}{\tilde{b}}^{\pi, \, M}_{M-2}= 3\, (1+\eta)\, a_{\pi} -2\, \eta\end{equation}
$$ \vdots $$
\begin{equation}  \prescript{\circ}{}{\tilde{b}}^{\pi, \, M}_{1}= M\, (1+\eta)\, a_{\pi} - (M-1)\, \eta \end{equation}
\begin{equation}  \prescript{\circ}{}{\tilde{b}}^{\pi, \, M}_{-1}= (M-2)\, (1+\eta)\, a_{\pi} -
(M-1) \, \eta 
\end{equation}
\begin{equation}  \prescript{\circ}{}{\tilde{b}}^{\pi, \, M}_{-2}= (M-3)\, (1+\eta)\, a_{\pi} -(M-2)\, \eta 
\end{equation}
$$\vdots$$
\begin{equation} \prescript{\circ}{}{\tilde{b}}^{\pi, \, M}_{-(M-2)}=(1+\eta)\, a_{\pi}-2\, \eta\end{equation}
\begin{equation} \prescript{\circ}{}{\tilde{b}}^{\pi, \, M}_{-(M-1)}=-\eta\end{equation} 
On the other hand, for the output gap, (256)-(261) become
\begin{equation} \prescript{\circ}{}{\tilde{b}}^{y, \, M}_{M}=
(1+\eta)\, a_{y}\end{equation}
\begin{equation} \prescript{\circ}{}{\tilde{b}}^{y, \, M}_{M-1}=2\, (1+\eta)\, a_{y}\end{equation}
$$\vdots$$
\begin{equation}\prescript{\circ}{}{\tilde{b}}^{y, \, M}_{1}=
M\, (1+\eta)\, a_{y} \end{equation}
\begin{equation}\prescript{\circ}{}{\tilde{b}}^{y, \, M}_{0}=
(M-1)\, (1+\eta)\, a_{y} \end{equation}
\begin{equation} \prescript{\circ}{}{\tilde{b}}^{y, \, M}_{-1}=
(M-2)\, (1+\eta)\, a_{y}\end{equation}
$$\vdots$$
\begin{equation} \prescript{\circ}{}{\tilde{b}}^{y, \, M}_{-(M-2)}=(1+\eta)\, a_{y}\end{equation}
\subsubsection{Main Phillips Curve}
The denominator (265) takes the form  
 \begin{multline} b^{M} = 
M\, (M + \eta)- 
(M-1)\, (1+\eta)\, a_{\pi} -\frac{(1+\eta)}{a
_{y}}\Bigg\{a_{\pi}\bigg(a_{y}\, \big\{(1+a_{y})^{M-2} -(M-2)\big\} \\ + (1+a_{y})\big\{ (1+a_{y})^{M-3} -1\big\}\bigg) + \\ \frac{a_{y}\, \big\{M\, (1+a_{y})^{M}-1\big\} + 
 (1 +a_{y})\big\{(1+a_{y})^{M-1}-1\big\}}{(1+a_{y})^{M}}
 \Bigg\}\end{multline}
where expressions are valid so long as $a_{y} \neq 0$. The numerators, matching (266)-(275), are as follows 
\begin{equation} {\tilde{b}}^{\pi, \, M}_{M}= 
a_{\pi}\frac{(1+\eta)}{1+a_{y}}\end{equation}
\begin{equation} \tilde{b}^{\pi, \, M}_{M-1}=
2\, (1+\eta)\, a_{\pi} -\eta 
- (1+\eta)\, a_{y}\Bigg[ a_{\pi}
\bigg(\frac{1}{(1+ 
a_{y})^{2}} + \frac{
2}{1+ a_{y}} \bigg)- \frac{1}{1+a_{y}}\Bigg]
\end{equation}
\begin{multline} \tilde{b}^{\pi, \, M}_{M-2}=
3\, (1+\eta)\, a_{\pi} -2 \, \eta
- (1+\eta)\, a_{y}\, \Bigg[ 
a_{\pi}\bigg( \frac{1}{(1+ a_{y})^{3}} + \frac{2}{(1+ a_{y})^{2}} + \frac{3}{(1+ a_{y})}\bigg) - \\ \bigg( \frac{
1}{(1+ a_{y})^{2}} + \frac{2}{(1+ a_{y})} \bigg)\Bigg]
\end{multline}
$$\vdots$$
\begin{multline} \tilde{b}^{\pi, \, M}_{1}= M\, (1+\eta)\, a_{\pi} - (M-1)\, \eta - \frac{a_{\pi}\, (1+\eta)}{a_{y}\, (1+a_{y})^{M}}\bigg((M\, a_{y}-1)\, (1+a_{y})^{M} + 1 \bigg)
+ \\ \frac{(1+\eta)}{a_{y}\, (1+a_{y})^{M-1}}\bigg(\big\{(M-1)\, a_{y}-1\big\}(1+a_{y})^{M-1} + 1 \bigg)
\end{multline}
\begin{multline} \tilde{b}^{\pi, \, M}_{-1}= 
(M-2)\, (1+\eta)\, a_{\pi} -
(M-1) \, \eta -\frac{(1+\eta)}{a_{y}}\Bigg\{a_{y}\, \big\{(1+a_{y})^{M-2} -(M-2)\big\} 
 \\ + (1+a_{y})\, \big\{ (1+a_{y})^{M-3} -1\big\}
- \\ a_{\pi}\bigg(a_{y}\, \big\{(1+a_{y})^{M-3} -(M-3)\big\}  
+ (1+a_{y})\, \big\{ (1+a_{y})^{M-4} -1\big\}\bigg)\Bigg\}\end{multline}
$$\vdots$$
\begin{equation}\tilde{b}^{\pi, \, M}_{-(M-3)}=
(2 +a_{y} )\, (1+\eta)\, (a_{\pi}-a_{y}) - 
3\, \eta - (1+\eta)\, a_{y}
\end{equation}
\begin{equation} \tilde{b}^{\pi, \, M}_{-(M-2)}=
(1+\eta)\, (a_{\pi}-a_{y}) - 2\, \eta 
\end{equation}
\begin{equation} \tilde{b}^{\pi, \, M}_{-(M-1)}=-\eta\end{equation}
\begin{multline} {\tilde{b}}^{y, \, M}
= \frac{(1+\eta)}{a_{y}}\Bigg\{\frac{(M\, a_{y}-1)(1+a_{y})^{M} +1}{(1+a_{y})^{M}}
 + (M-1)\, a^{2}_{y} + \\ 
 (1+a_{y})\bigg(a_{y}\, \big\{(1+a_{y})^{M-2} -(M-2)\big\} 
+ (1+a_{y})\, \big\{ (1+a_{y})^{M-3} -1\big\}\bigg)\Bigg\}\end{multline}
\subsection{Patient Errors}
The limiting shock coefficients as $\beta \rightarrow 1$ are again grouped according to contract length. 
\subsubsection{Taylor (2)}
\begin{equation} \tilde{b}^{\psi, \, 2}_{2}=-\frac{(1-\rho)\, (1+\eta)}{1 + a_{y}}\end{equation}
\begin{equation} \tilde{b}^{\psi, \, 2}_{1} =-(1-\rho)\, (1+\eta)\frac{(2 + a_{y})}{(1 +a_{y})^{2}}\end{equation}
\begin{equation} \tilde{b}^{\psi, \, 2}_{0}=-(1-\rho)\, (1+\eta)\end{equation}
\begin{equation} \tilde{b}^{\pi(e), \, 2}_{1} = \frac{(1+\eta)\, a_{y}}{ 1+ a_{y}}+ 1\end{equation}
\begin{equation} \tilde{b}^{\pi(e), \, 2}_{0}= 
(1+\eta)\, \frac{(2 + a_{y})}{(1 +a_{y})^{2}} + 1 \end{equation}
\begin{equation} \tilde{b}^{y(e), \, 2}_{1} = \frac{(1+\eta)}{2}a_{y}(1 + 2a_{y})\end{equation}
\begin{equation} \tilde{b}^{y(e), \, 2}_{0}= 
(1+\eta)\, \frac{(3+ 3\, a_{y} + a_{y}^{2})}{(1+a_{y})^{2}}\end{equation}
which were a rendering of (284)-(290). 
\subsubsection{Taylor (M)}
Performing the same operation on (292)-(303) yields the following expressions.\footnote{The most difficult answer to obtain was (339). Rather than resorting to L'H{\^o}pital's rule with (295), I regressed to the original expressions in search of an elegant solution. It is easy to perceive that as time preference recedes, the limit of (263) is  
$$\prescript{\circ}{}{\tilde{b}}^{\psi, \, M}_{0}=(1-\rho)\bigg[M-1 + (M-2)\, \rho + (M-3)\, \rho^{2} + \cdots + 2\, \rho^{M-3} + \rho^{M-2}\bigg]$$
distinguishing between terms dependent and independent of $M$ 
gives $$(1-\rho)\bigg[ M\, (1+ \rho + \cdots + \rho^{M-2})-\underbrace{\big\{1+2\,\rho + 3\, \rho^{2} + \cdots + (M-1)\, \rho^{M-2}\big\}}_{S(\rho)}\bigg]$$
The first bracket is an ordinary geometric progression. A second geometric series emerges along with two further terms when $S(\rho)$ is evaluated, via subtracting $\rho S(\rho)$ and telescoping. After some simplification, I arrive at the 
bottom row of (339). The method for finding $\prescript{\circ}{\circ}{b}^{\psi, \, M}_{0}$ is the same. Notice that $\rho =1$ is still a factor of the second row and thus $\tilde{b}^{\psi, \, M}_{0}$ itself, as a consequence of L'H{\^o}pital's rule and the factor theorem.}
\begin{equation} \tilde{b}^{\psi, \, M}_{M}= -\frac{(1-\rho)\, (1+\eta)}{(1+ a_{y})}\end{equation}
\begin{equation}
 \tilde{b}^{\psi, \, M}_{M-1}= -(1-\rho)\,(1+\eta)\frac{(2 + a_{y})}{(1+ a_{y})^{2}}
\end{equation}
$$\vdots$$
\begin{equation}
 \tilde{b}^{\psi, \, M}_{1}= -(1-\rho)(1+\eta)\bigg( \frac{(1+ a_{y})^{M}-1}{a_{y}(1+ a_{y})^{M}}\bigg) 
\end{equation}
\begin{multline} \tilde{b}^{\psi, \, M}_{0}= \prescript{\circ}{\circ}{b}^{\psi, \, M}_{0} + (1+\eta)\, \prescript{\circ}{}{\tilde{b}}^{\psi, \, M}_{0} = (1+\eta) \Bigg[ 
(M-1)\,\rho^{M-1}-1 -\rho \bigg( \frac{1 -\rho^{M-2}}{1 -\rho}\bigg) + \\ (1 -\rho)\, a_{y}\, \Bigg\{ (M-2)\, \rho^{M-2}\, (1+a_{y})^{M-2}-1 - \\ \rho\, (1+a_{y})\bigg( \frac{\rho^{M-3}\, (1+a_{y})^{M-3}-1}{\rho\, (1+a_{y})-1}\bigg)\Bigg\}\Bigg]
\end{multline}
\begin{equation}
 \tilde{b}^{\pi(e), \, M}_{M-1}=  \frac{(1+\eta)}{(1+ a_{y})}
\end{equation}
\begin{equation}
 \tilde{b}^{\pi(e), \, M}_{M-2}= (1+\eta)\frac{(2 + a_{y})}
 {(1 +a_{y})^{2}} \end{equation}
 $$\vdots$$
 \begin{equation}
 \tilde{b}^{\pi(e), \, M}_{1}= (1+\eta)\bigg( \frac{(1+ a_{y})^{M-1}-1}{a_{y}(1+ a_{y})^{M-1}}\bigg)\end{equation} 
\begin{equation}
 \tilde{b}^{\pi(e), \, M}_{0}= (1+\eta)\bigg( \frac{(1+ a_{y})^{M}-1}{a_{y}(1+ a_{y})^{M}}\bigg)
\end{equation} 
\begin{equation}
 \tilde{b}^{y(e), \, M}_{M-1}= \frac{(1+\eta)}{(1+ a_{y})}
\end{equation}
\begin{equation}
 \tilde{b}^{y(e), \, M}_{M-2}= (1+\eta)\, \frac{(2 + a_{y})}{(1 +a_{y})^{2}} 
 \end{equation}
 $$\vdots$$
 \begin{equation}
 \tilde{b}^{y(e), \, M}_{1}=(1+\eta)\bigg( \frac{(1+ a_{y})^{M-1}-1}{a_{y}\, (1+ a_{y})^{M-1}}\bigg)
\end{equation} 
\begin{equation}
 \tilde{b}^{y(e), \, M}_{0}=(1+\eta)\bigg( \frac{(1+ a_{y})^{M}-1}{a_{y}\, (1+ a_{y})^{M}}\bigg)
\end{equation} 
\section{Eigenvalue Configuration Results}
This final leg of the paper completes the proof of the internal consistency of the two Phillips curves (115) and (116). Both sojourns chart the same course, beginning by deriving the relevant characteristic equation and ceasing with an eigenvalue evaluation.
\subsection{Proof of Proposition 8: Existence Taylor $(2)$}
Following Theorem 3 (SELCKE), all that is needed is to verify eigenvalue \\ conditions.
\begin{lemma} The characteristic polynomial of the Taylor $(2)$ period contract model is given by 
\begin{multline} \lambda^{4} 
+ \bigg(\frac{\big\{4 + \eta -(1+\eta)\, a_{\pi} +a_{y}\big\}\, (1+a_{y})^{2}+ 
(1+\eta)\, a_{y} }{\eta\, (1+a_{y})^{2}}\bigg)\lambda^{3}- \\ 
\Bigg\{a_{y}\, (1+\eta) + \big\{a_{\pi}\, a_{y}\, (1+\eta) + 2\, (a_{\pi}-a_{y})\, (1+\eta) -\eta\big\}\, (1+a_{y}) + \\ \big\{ 2\, (2+\eta)-(1+\eta)\, a_{\pi}\big\}\, (1+a_{y})^{2}\Bigg\}\frac{\lambda^{2}}{\eta(1+a_{y})} 
+ \\ \bigg(\frac{ \big\{2\, (1+\eta) \, a_{\pi}-\eta \big\}\, (1 + a_{y})- (1+\eta)\, \big\{ a_{\pi} + a_{y}\, (2\, a_{\pi}-1)\big\}}{\eta}\bigg)\lambda + \\ \frac{a_{\pi}\, (1+\eta)}{\eta}= 0\end{multline}
\end{lemma}
\begin{proof}The proof is an exercise in lag operator manipulation. For greater \\ succinctness and generality I will keep parametric coefficients in place until the final step. Errors and expectations can be ignored. 
\par The consolidated Phillips curve corresponding to 
(151) is $$\pi_{t}\, (1-  b_{1}^{\pi, \,2}\, \mathbb{L} -b_{2}^{\pi, \, 2}\,\mathbb{L}^{2} -b^{\pi, \, 2}_{-1}\, \mathbb{L}^{-1})=b^{y, \, 2}\,\hat{y}^{e}_{t}$$
whilst the aggregate demand side, represented by (11) and (27), compacts to 
\begin{equation}\hat{y}^{e}_{t}=\frac{(\mathbb{L}^{-1}-a_{\pi})\, \pi_{t}}{1+ a_{y} -\mathbb{L}^{-1}}\end{equation}
combining then clearing out all the negative exponents yields the general \\ characteristic equation.\footnote{The work inside the brackets is routine. Here are some more steps. 

$$\pi_{t}\bigg[(1+ a_{y} -\mathbb{L}^{-1})\, (1-  b_{1}^{\pi, \,2}\, \mathbb{L} -b_{2}^{\pi, \, 2}\, \mathbb{L}^{2} -b^{\pi, \, 2}_{-1}\, \mathbb{L}^{-1})
+ a_{\pi}\, b^{y, \, 2} - b^{y, \, 2}\, \mathbb{L}^{-1}\bigg]=0$$
$$\bigg(\big\{1+ a_{y}\big\}\, \mathbb{L}-1\bigg)\bigg(\mathbb{L}-  b_{1}^{\pi, \,2}\, \mathbb{L}^{2} -b_{2}^{\pi, \, 2}\, \mathbb{L}^{3} -b^{\pi, \, 2}_{-1}\, \bigg)+ a_{\pi}\, b^{y, \, 2}\, \mathbb{L}^{2}- b^{y, \, 2}\, \mathbb{L}=0$$
\begin{multline*}b_{2}^{\pi, \, 2}\, \big\{1+ a_{y}\big\} \mathbb{L}^{4} + \bigg(b_{1}^{\pi, \,2}\big\{1+ a_{y}\big\} -b_{2}^{\pi, \, 2}\bigg)\mathbb{L}^{3} 
- \bigg( a_{\pi}\,b^{y, \, 2} + {b}_{1}^{\pi, \,2} + 1+a_{y}\bigg)\mathbb{L}^{2} + \\  \bigg(1 + b_{-1}^{\pi, \,2}\, \big\{1+ a_{y}\big\} + b^{y, \, 2}\bigg)\mathbb{L} - b_{-1}^{\pi, \,2}=0 \end{multline*}
(350) is attained via the substitution $\lambda = 1/ \mathbb{L}$.} 

\begin{multline} \lambda^{4} - \frac{\bigg(b^{2} + \tilde{b}_{-1}^{\pi, \,2}\, \big\{1+ a_{y}\big\} + \tilde{b}^{y, \, 2}\bigg)}{\tilde{b}_{-1}^{\pi, \,2}}\lambda^{3} + \frac{\bigg( a_{\pi}\,\tilde{b}^{y, \, 2} + \tilde{b}_{1}^{\pi, \,2} + b^{2}(1+a_{y})\bigg)}{\tilde{b}_{-1}^{\pi, \,2}}\lambda^{2}
- \\ \frac{\bigg( \tilde{b}_{1}^{\pi, \,2}\, \big\{1+ a_{y}\big\}- \tilde{b}_{2}^{\pi, \, 2}\bigg)}{\tilde{b}_{-1}^{\pi, \,2}}\lambda - \frac{\tilde{b}_{2}^{\pi, \, 2}\, \big\{1+ a_{y}\big\}}{\tilde{b}_{-1}^{\pi, \,2}} =0\end{multline}
Inputting the parametric expressions (117)-(121) completes the task. \end{proof}
\begin{lemma} 
There exists a recursive equilibrium at standard parameters. 
\end{lemma}
\begin{proof}After replacing the parameters with their numerical counterparts,  the equation
\begin{equation}\lambda^{4}+ \frac{16}{9}\lambda^{3} -
\frac{79}{24}\lambda^{2} - \frac{1}{4}\lambda + \frac{5}{8}=0 \end{equation}
materializes. The roots are $\lambda \approx -0.435, \, 0.478, \, 1.048, \, -2.869$. This \\configuration, with two inside and two outside the unit circle, conforms with the presence of two state variables $(\pi_{t-1}, \, \pi_{t-2})'$ and two jump variables $(\pi_{t}, \, \hat{y}^{e}_{t})'$, established in 
Proposition 10. \par All root approximations can be rigorously justified with reference to the intermediate value theorem. I will only depict the charge closest to the unit circle. Label the left-hand side of (351) $F$. Numerical computations show that $F(1.0485)=0.0019 >0$, whilst $F(1.0475)=-0.0014 <0$. Hence, there is a root $1.0475 < \lambda < 1.0485$, which is outside the unit circle. Similar bounds for the other roots serve to make the arguments mathematically concrete. \end{proof}
\subsection{Proof of Proposition $9$: Existence Taylor $(4)$ }%
\begin{lemma} 
The characteristic polynomial for the Taylor (4) takes the form 
\begin{multline} \lambda^{8}- \bigg(\frac{(1+\eta)\, a_{\pi}
-\eta -a_{y}}{\eta}\bigg)\lambda^{7} 
-\bigg(\frac{(1+\eta)\, a_{\pi} -\eta -2\, a_{y}}{\eta}\bigg)\lambda^{6} + \\ \bigg(\frac{16 + \eta -(1+\eta)\, a_{\pi} + 3\, a_{y}}{\eta}\bigg)\lambda^{5} - \bigg(\frac{16 + \eta + (1+\eta)\, a_{\pi} + 12\, a_{y}}{\eta}\bigg)\lambda^{4} \\ 
+ \bigg(\frac{(1+\eta)\, a_{\pi} + 3\, a_{y}-\eta}{\eta}\bigg)\lambda^{3}
+\bigg(\frac{(1+\eta)\, a_{\pi} + 2\, a_{y}-\eta}{\eta}\bigg)\lambda^{2} +
\\ \bigg(\frac{(1+\eta)\, a_{\pi} + a_{y}-\eta}{\eta}\bigg)\lambda + \frac{a_{\pi}\, (1+\eta)}{\eta}=0
\end{multline}
\end{lemma}
\begin{proof} 
In lag operator terms, the endogenous dynamics of the Phillips curve (116) appear as 
$$\pi_{t}(1-  b_{1}^{\pi, \, 4}\, \mathbb{L} -b_{2}^{\pi, \, 4}\,\mathbb{L}^{2}-b_{3}^{\pi, \, 4}\, \mathbb{L}^{3}-b_{4}^{\pi, \, 4}\,\mathbb{L}^{4} -b^{\pi, \, 4}_{-1}\, \mathbb{L}^{-1}-b^{\pi, \, 4}_{-2}\, \mathbb{L}^{-2})=b^{y, \, 4}\, \hat{y}^{e}_{t}$$ (349) is unchanged.\footnote{Missing steps include 
\begin{multline*} \bigg(\big\{1+ a_{y}\big\}\, \mathbb{L} -1\bigg)\bigg(\mathbb{L}^{3}- b_{4}^{\pi, \, 4}\, \mathbb{L}^{7} -  b_{3}^{\pi, \, 4}\, \mathbb{L}^{6}- b_{2}^{\pi, \, 4}\, \mathbb{L}^{5} -b_{1}^{\pi, \,4}\, \mathbb{L}^{4} -  b^{\pi, \, 4}_{-1}\,\mathbb{L}^{2} - b^{\pi, \, 4}_{-2}\,\mathbb{L}-b^{\pi, \, 4}_{-3}\bigg)
+ \\ a_{\pi}\, b^{y, \, 4}\, \mathbb{L}^{4}- b^{y, \, 4}\, \mathbb{L}^{3}=0\end{multline*}
which collates to 
\begin{multline*} b_{4}^{\pi, \, 4}\,\big\{1+ a_{y}\big\}\, \mathbb{L}^{8} +\bigg(b_{3}^{\pi, \, 4}\,\big\{1+ a_{y}\big\}- b_{4}^{\pi, \, 4}\bigg)\mathbb{L}^{7} + \bigg(b_{2}^{\pi, \, 4}\,\big\{1+ a_{y}\big\}- b_{3}^{\pi, \, 4}\bigg)\mathbb{L}^{6} + \\ \bigg(b_{1}^{\pi, \, 4}\,\big\{1+ a_{y}\big\}- b_{2}^{\pi, \, 4}\bigg)\mathbb{L}^{5}- \bigg(a_{\pi}b^{y, \, 4} + b_{1}^{\pi, \, 4} + 1+ a_{y}\bigg)\mathbb{L}^{4} + \bigg(b^{y, \, 4} + b_{-1}^{\pi, \, 4}\big\{1+ a_{y}\big\} + 1\bigg)\mathbb{L}^{3} + \\ \bigg(b_{-2}^{\pi, \, 4}\,\big\{1+ a_{y}\big\}- b_{-1}^{\pi, \, 4}\bigg)\mathbb{L}^{2} + \bigg(b_{-3}^{\pi, \, 4}\,\big\{1+ a_{y}\big\}- b_{-2}^{\pi, \, 4}\bigg)\mathbb{L}-b_{-3}^{\pi, \, 4}=0\end{multline*}} The lag polynomial analagous to (350) is 
\begin{multline} b_{4}^{\pi, \, 4}\,\big\{1+ a_{y}\big\}\, \mathbb{L}^{8} +\bigg(b_{3}^{\pi, \, 4}\,\big\{1+ a_{y}\big\}- b_{4}^{\pi, \, 4}\bigg)\mathbb{L}^{7} + \bigg(b_{2}^{\pi, \, 4}\,\big\{1+ a_{y}\big\}- b_{3}^{\pi, \, 4}\bigg)\mathbb{L}^{6} + \\ \bigg(b_{1}^{\pi, \, 4}\,\big\{1+ a_{y}\big\}- b_{2}^{\pi, \, 4}\bigg)\mathbb{L}^{5}- \bigg(a_{\pi}\, b^{y, \, 4} + b_{1}^{\pi, \, 4} + 1+ a_{y}\bigg)\mathbb{L}^{4} + \bigg(b^{y, \, 4} + b_{-1}^{\pi, \, 4}\big\{1+ a_{y}\big\} + 1\bigg)\mathbb{L}^{3} \\  + \bigg(b_{-2}^{\pi, \, 4}\,\big\{1+ a_{y}\big\}- b_{-1}^{\pi, \, 4}\bigg)\mathbb{L}^{2} + \bigg(b_{-3}^{\pi, \, 4}\,\big\{1+ a_{y}\big\}- b_{-2}^{\pi, \, 4}\bigg)\mathbb{L}-b_{-3}^{\pi, \, 4}=0\end{multline}
\begin{multline}
\lambda^{8} - \frac{\bigg(\tilde{b}_{-3}^{\pi, \, 4}\,\big\{1+ a_{y}\big\}- \tilde{b}_{-2}^{\pi, \, 4}\bigg)}{\tilde{b}_{-3}^{\pi, \, 4}}\lambda^{7} -\frac{\bigg(\tilde{b}_{-2}^{\pi, \, 4}\,\big\{1+ a_{y}\big\}- \tilde{b}_{-1}^{\pi, \, 4}\bigg)}{\tilde{b}_{-3}^{\pi, \, 4}}\lambda^{6} -\\ \frac{\bigg(\tilde{b}^{y, \, 4} + \tilde{b}_{-1}^{\pi, \, 4}\, \big\{1+ a_{y}\big\} + b^{4}\bigg)}{\tilde{b}_{-3}^{\pi, \, 4}}\lambda^{5} + \frac{\bigg(a_{\pi}\, \tilde{b}^{y, \, 4} + \tilde{b}_{1}^{\pi, \, 4} + b^{4}\, \big\{1+ a_{y}\big\}\bigg)}{\tilde{b}_{-3}^{\pi, \, 4}}\lambda^{4} - \\ 
\frac{\bigg(\tilde{b}_{1}^{\pi, \, 4}\,\big\{1+ a_{y}\big\}-\tilde{b}_{2}^{\pi, \, 4} \bigg)}{\tilde{b}_{-3}^{\pi, \, 4}}\lambda^{3} -
 \frac{\bigg(\tilde{b}_{2}^{\pi, \, 4}\,\big\{1+ a_{y}\big\}-\tilde{b}_{3}^{\pi, \, 4} \bigg)}{\tilde{b}_{-3}^{\pi, \, 4}}\lambda^{2} 
 - \\ \frac{\bigg(\tilde{b}_{3}^{\pi, \, 4}\, \big\{1+ a_{y}\big\}-\tilde{b}_{4}^{\pi, \, 4}\bigg)}{\tilde{b}_{-3}^{\pi, \, 4}}\lambda-\frac{\tilde{b}_{4}^{\pi, \, 4}\,\big\{1+ a_{y}\big\}}{\tilde{b}_{-3}^{\pi, \, 4}}=0\end{multline}
Digging down to primitives, we uncover that, due to (316), 
\begin{multline}b^{4} \equiv  
4\, (4 + \eta)- 
3\, (1+\eta)\, a_{\pi} - \frac{(1+\eta)}{a_{y}}\Bigg\{a_{\pi}\, a_{y}^{2}\, (3+a_{y}) +  \frac{(4\, a_{y}-1)(1+a_{y})^{4} +1}{(1+a_{y})^{4}}\Bigg\} \\ =  
\frac{5281}{648} \approx 8.150
\end{multline}
It is easiest to deduce from (320)-(322), alongside (118)-(119), that \newline 
${\tilde{b}}^{\pi, \, 4}_{4}={\tilde{b}}^{\pi, \, 2}_{2}=5/3$, ${\tilde{b}}^{\pi, \, 4}_{3}={\tilde{b}}^{\pi, \, 2}_{1}=4/9$, ${\tilde{b}}^{\pi, \, 4}_{2}=-1/27$, whilst (323) implies 
\begin{multline} \tilde{b}^{\pi, \, 4}_{1}\equiv 
4\, (1+\eta)\, a_{\pi} -3\, \eta - \frac{a_{\pi}\, (1+\eta)}{a_{y}(1+a_{y})^{4}}\bigg((4\, a_{y}-1)(1+a_{y})^{4} + 1\bigg) + \\ \frac{(1+\eta)}{a_{y}\, (1+a_{y})^{3}}
\bigg((3a_{y}-1)(1+a_{y})^{3} + 1\bigg) = -\frac{2}{81} \approx -0.025
\end{multline}
The expectation terms can be obtained from (325)-(327), specifically, \newline
$\tilde{b}^{\pi, \, 4}_{-1}=-29/2$, $\tilde{b}^{\pi, \, 4}_{-2}=-8$ and $\tilde{b}^{\pi, \, 4}_{-3}=-4$. Finally, the output gap term (328) comes to 
\begin{equation}{\tilde{b}}^{y, \, 4}\equiv
\frac{(1+\eta)}{a_{y}}\Bigg\{\frac{(4\, a_{y}-1)\, (1+a_{y})^{4} +1}{(1+a_{y})^{4}}+ a^{2}_{y}\, (6 + 4\, a_{y} + a^{2}_{y})\Bigg\}=\frac{21125}{648} \approx 32.600
 \end{equation}
putting in these parametric functions and simplifying achieves the targeted \newline result.
\end{proof}
\begin{lemma} Under Construction 1, the limiting eigenvalue configuration is consistent with recursive equilibrium, 
at the preferred parameter selection. \end{lemma}
\begin{proof} With numerical values in place the eigenvalue equation works out as
\begin{equation}
\lambda^{8} + \frac{1}{2}\lambda^{7} + \frac{5}{8}\lambda^{6} + \frac{19}{4}\lambda^{5} - \frac{57}{8}\lambda^{4} -\frac{1}{8}\lambda^{2}-\frac{1}{4}\lambda + \frac{5}{8}=0
\end{equation}
The solution set is $\lambda \approx -0.526, \, 0.545\, e^{\pm1.604i},\, 1, \, 1.827\, e^{\pm1.415}, \, -2.082$. \par Concerning the precise root location, relative to the unit circle, I will only display the less familiar arguments concerning complex roots. The procedure is to break the characteristic equation into real and imaginary parts; both are eighth-order polynomials. The fundamental theorem of algebra implies they will both have eight roots and they will coincide. Moreover, they will occur in conjugate pairs. Therefore, it is sufficient to focus on the positive conjugate and show each desired region contains a real zero. For the first complex root $\Re{(0.5445\, e^{1.6035i})}>0$, whilst $\Re{(0.5455\, e^{1.6045i})}<0$. Obversely, for the larger root $\Re{(1.8265\, e^{1.4145i})}<0$ but $\Re{(1.8275\, e^{1.4155i})}>0$, sealing the argument. 
\par This assures us that there are four roots inside the unit circle, matching the four state variables 
$(\pi_{t-1}, \, \pi_{t-2}, \, \pi_{t-3}, \, \pi_{t-4})'$. The only obstacle is the root on the unit circle. This is where the unique aspect of Construction 1 comes into play. 
\par With a share $\omega >0$ of flexible price firms, inflation is the weighted average of Classical and Keynesian determinants
\begin{multline} \pi_{t}= \omega\, (\pi_{t} + \hat{mc}_{t}) + (1-\omega)\bigg(b^{\pi, \, 4}_{4}\, \pi_{t-4}+ \cdots + b^{\pi, \, 4}_{1}\, \pi_{t-1} + b^{y, \, 4}\, \hat{y}^{e}_{t} + \\ b^{\pi, \, 4}_{-1}\, \mathbb{E}_{t}\, \pi_{t+1} + \cdots + b^{\pi,\, 4}_{-3}\, \mathbb{E}_{t}\, \pi_{t+3}  
+\hat{u}^{4}_{t}\bigg) \end{multline}
Upon simplification, it is immediate that among the endogenous variables only the output gap coefficient $b^{y, \, 4}$ will change. Labelling the characteristic function $G(\omega)$, I can deduce from 
(354) that
\begin{multline}
G(\omega) \equiv \lambda^{8} + \frac{1}{2}\lambda^{7} + \frac{5}{8}\lambda^{6} + \bigg\{\frac{19 +\omega/(1-
\omega)}{4}\bigg\}\lambda^{5} - \\ 
\bigg\{ \frac{57 +\omega/(1-\omega)}{8}\bigg\}\lambda^{4} -\frac{1}{8}\lambda^{2}-\frac{1}{4}\lambda + \frac{5}{8}\end{multline}
In response to my refinement
$$\frac{\mathrm{d}G(\omega)}{\mathrm{d}\lambda}\Bigg\rvert_{
\omega=0, \, \lambda=1} \equiv 8\, \lambda^{7}-\frac{7}{2}\lambda^{6} + \frac{15}{4}\lambda^{5} + 
\frac{95}{4}\lambda^{4} -\frac{57}{2}\lambda^{3} -\frac{1}{4}\lambda -\frac{1}{4}=3> 0$$
$$\frac{\mathrm{d}G} {\mathrm{d}\omega}\Bigg\rvert_{\omega=0, \, \lambda =1}\equiv \frac{\lambda^{5}}{4\, (1-
\omega)^{2}}-\frac{\lambda^{4}}{8\, (1-\omega)^{2}}=
\frac{1}{4}>0$$
Thus, by the chain rule $$\frac{\mathrm{d}\lambda} {\mathrm{d}
\omega}\Bigg\rvert_{\omega=0, \, \lambda =1}>0$$
and hence a flexible fringe precipitates the correct eigenvalue constellation. \end{proof}
\begin{remark} The application of the Fundamental Theorem of algebra is valid for any order of polynomial, making the technique generic. \end{remark}
\end{document}